\documentclass[a4paper, accepted=2026-04-07]{quantumarticle}
\pdfoutput=1
\usepackage[numbers]{natbib}
\usepackage{graphicx}  
\usepackage{dcolumn}   
\usepackage{bm}        
\usepackage{bbold}
\usepackage{amssymb}
\usepackage{physics}   
\usepackage{mathtools}
\usepackage{esvect}
\usepackage{amsmath}

\usepackage{algorithm}
\usepackage[noend]{algpseudocode}
\usepackage{setspace}
\algblock[Class]{Class}{End}

\usepackage{amsthm}
\usepackage{ragged2e}
\usepackage{verbatim}
\usepackage[colorlinks=true,plainpages=false,pdfpagelabels]{hyperref}
\makeatletter
\newcommand\setcurrentname[1]{\def\@currentlabelname{#1}}
\makeatother
\usepackage{color,xcolor,colortbl}
\usepackage{multirow}
\usepackage{boxedminipage}
\usepackage{changepage}
\usepackage{lipsum}
\usepackage{array}
\usepackage{etoolbox}
\usepackage[caption=false]{subfig}
\usepackage{appendix}
\usepackage{float}
\usepackage{dsfont}

\usepackage[mathscr]{euscript}
\usepackage{enumitem}

\usepackage[normalem]{ulem}
\usepackage{xpatch}

\usepackage{tikz-cd}

\usepackage{algorithm}
\usepackage[noend]{algpseudocode}
\usepackage{setspace}
\algblock[Class]{Class}{End}

\def\Tr{\operatorname{Tr}}
\def\sc{\operatorname{sc}}
\def\IR{\operatorname{IR}}
\def\IR{\operatorname{SIR}}

\def\SEP{\operatorname{SEP}}

\def\PPT{\operatorname{PPT}}
\def\supp{\operatorname{supp}}
\def\LOCC{\operatorname{LOCC}}

\def\>{\rangle}
\def\<{\langle}

\def\id{\mathds{1}}

\def\({\left(}
\def\){\right)}
\def\[{\left[}
\def\]{\right]}

\let\emptyset\varnothing

\newcommand{\mc}[1]{\mathcal{#1}}
\newcommand{\wt}[1]{\widetilde{#1}}

\algnewcommand{\Inputs}[1]{
  \State \textbf{Inputs:}
  \Statex \hspace*{\algorithmicindent}\parbox[t]{.8\linewidth}{\raggedright #1}
}

\algnewcommand{\Initialize}[1]{
  \State \textbf{Initialize:}
  \Statex \hspace*{\algorithmicindent}\parbox[t]{.8\linewidth}{\raggedright #1}
}

\newtheorem{theorem}{Theorem}

\newtheorem{corollary}{Corollary}

\newtheorem{observation}{Observation}
\newtheorem{definition}{Definition}

\newtheorem{lemma}{Lemma}

\newtheorem{proposition}{Proposition}
\newtheorem{remark}{Remark}

\renewcommand{\nocontentsline}[3]{}
\renewcommand{\tocless}[2]{\bgroup\let\addcontentsline=\nocontentsline#1{#2}\egroup}

\allowdisplaybreaks

\begin{document}

\widetext

\title{Cost of quantum secret key}

\author{Karol Horodecki}\email{karol.horodecki@ug.edu.pl} 

\affiliation{Institute of Informatics, National Quantum Information Centre, Faculty of Mathematics, Physics and Informatics, University of Gda\'nsk, Wita Stwosza 57, 80-308 Gda\'nsk, Poland}

\affiliation{International Centre for Theory of Quantum Technologies, University of Gda\'nsk, Wita Stwosza 63, 80-308 Gda\'nsk, Poland}
\affiliation{School of Electrical and Computer Engineering, Cornell University, Ithaca, New York 14850, USA}

\author{Leonard Sikorski }\email{leonard.sikorski@phdstud.ug.edu.pl}
\affiliation{Institute of Informatics, National Quantum Information Centre, Faculty of Mathematics, Physics and Informatics, University of Gda\'nsk, Wita Stwosza 57, 80-308 Gda\'nsk, Poland}

\author{Siddhartha Das}\email{das.seed@iiit.ac.in}
\affiliation{q4i, Centre for Quantum Science and Technology (CQST), Center for Security, Theory and Algorithmic Research (CSTAR), International Institute of Information Technology Hyderabad, Gachibowli 500032, Hyderabad, Telangana, India}

\author{Mark M. Wilde}\email{wilde@cornell.edu}
\affiliation{School of Electrical and Computer Engineering, Cornell University, Ithaca, New York 14850, USA}

\begin{abstract}
In this paper, we develop the resource theory of quantum secret key.
Operating under the assumption that entangled states with zero distillable key \textit{do not} exist, we define the \textit{key cost} of a quantum state, and device. We study its properties through the lens of a quantity that we call the \textit{key of formation}.
The main result of our paper is that the regularized key of formation is an upper bound on the key cost of a quantum state. The core protocol underlying this result is privacy dilution,  which converts states containing ideal privacy into ones with diluted privacy. Next, we show that the key cost is bounded from below by the regularized relative entropy of entanglement, which implies the irreversibility of the privacy creation-distillation process for a specific class of states.
We further focus on mixed-state analogues of pure quantum states in the domain of privacy, and we prove that a number of entanglement measures are equal to each other for these states, similar to the case of pure entangled states. The privacy cost and distillable key in the single-shot regime exhibit a yield-cost relation, and basic consequences for quantum  devices are also provided. Importantly, our results presented here will remain valid even if entangled states with zero distillable key were shown to exist.
\end{abstract}

\maketitle 
\section{Introduction}
Quantum mechanics is a rather surprising physical theory, as it is reversible in principle, contrary to our common experience. After an arbitrary quantum operation is applied to a system, if full access to its environment is available, the system's state can be set back to its initial form by a reversal operation.
However, in practice, reversibility is usually not possible. The system, after an operation, typically becomes entangled with an inaccessible environment in an irreversible way. 
This irreversibility can be quantified in different ways in the resource-theoretic framework~\cite{Chitambar-Gour}, which was first introduced and studied for the case of the resource theory of entanglement~\cite{BBPSSW1996}.

The origin of irreversibility is even more fundamental in the case of the resource of quantum secret key. In this setting, an eavesdropper or hacker will never give back a system that has leaked, even if it might be possible in practice. 
This fact was first noticed and studied in the classical scenario of secret key agreement (SKA)~\cite{Maurer1993,Renner2003}. There, the adversary has only a classical memory $Z$ and eavesdrops on two honest parties, who possess classical random variables $X$ and $Y$, respectively, and transform them by local (classical) operations and public communication.
Inspired by an apparent correspondence between entanglement theory and SKA established in~\cite{GW},
it was proved that the cost of creating a distribution sometimes exceeds its distillable secure content in the SKA scenario~\cite{Renner2003} (see also~\cite{Acn2004,W05,Horodecki2005,ChitambarHsieshWinter,Chitambar2015} for related work).  Key distillation for the quantum generalization of SKA has been studied in~\cite{ChristandlEkert,Ozols2014,Chitambar2018} via a mapping of tripartite distributions into pure tripartite states, as proposed in~\cite{Collins2002}~\footnote{The mapping of~\cite{Collins2002} is the following embedding of a classical distribution $P(x,y,z)$ into a Hilbert space: $\{P(x,y,z)\}_{x,y,z} \mapsto |\Psi\rangle_{XYZ} \coloneqq  \Sigma_{x,y,z} \sqrt{P(x,y,z)}|x,y,z\rangle $.}.

Building a quantum internet infrastructure is one of the main visions of the quantum information community~\cite{WehnerQInternet}.
It is then important to understand the cost of the network's constituents, i.e., quantum states and channels with secure key.
Indeed, since pure entanglement is not a precondition of quantum security~\cite{pptkey,keyhuge}, we assert here that a secure key is a proper resource for quantifying the information-theoretic expense of the quantum-secured internet. Motivated by this, we initiate the study of the cost of a secret key and understanding the problem of irreversibility in the quantum cryptographic scenario. 

In the most basic case of a point-to-point connection in the network, the two honest parties share a bipartite quantum state and process it by local operations and classical communication (LOCC). Assuming the worst case, the quantum eavesdropper has access to a purifying system of a purification~\cite{Nielsen-Chuang} of the state of the honest parties and obtains any classical communication and any system traced out by them during LOCC processing.
The main question we pose in this setting is as follows: ``How much private key is needed for the creation of an arbitrarily good approximation of $n$ copies of a bipartite quantum state by LOCC for sufficiently large $n$?"
 
With this in mind, we depart from the above approach of exploring the embedding of SKA into the quantum scenario \footnote{We consider, in fact, processing of an arbitrary pure state $|\psi\>_{ABE}$ (such that $\mathrm{Tr}_E|\psi\>\!\<\psi|_{ABE} = \rho_{AB}$) by (quantum) local operations and public communication. The pure state $|\psi\>_{ABE}$ in our approach may not originate from any classical distribution via the Collins--Popescu mapping.}. 
Instead, we develop the resource theory of key secure against a quantum adversary. It is important to note that the latter theory differs from the resource theory of entanglement. This is because there exist states that contain no distillable entanglement $(E_D =0)$~\cite{bound} but contain a strictly positive rate of distillable key $K_D$ secure against a quantum adversary~\cite{pptkey,keyhuge}. 

While entanglement theory has been thoroughly studied~\cite{Horodecki2009} since its invention~\cite{BBPSSW1996,BDSW96}, the resource theory of private key secure against a quantum adversary has yet to be developed fully from the resource-theoretic perspective. On the one hand, lower and upper bounds on the distillable key have been developed thoroughly, starting from the seminal protocols BB84 and E91~\cite{BEN84,Ekert91} and their follow-ups~\cite{Old_review,Devetak_2005,Acn2007,Pawowski2011,Lo2012,ArnonFriedman2018}  in a number of cryptographic scenarios (cf.~\cite{LoPan_review,Das2021}). Also, the aforementioned distillable key $K_D$ has been shown to be an entanglement measure and studied in~\cite{pptkey,keyhuge,Christandl-phd}. On the other hand, to our knowledge, the privacy cost has hitherto not been introduced nor studied in the fully quantum setup.

In this paper, we close this gap by defining and characterizing the key cost $K_C$, a fundamental quantity in the resource theory of privacy. The key cost is an upper bound on the distillable key secure against a quantum eavesdropper,  and it indicates how much key one needs to invest when creating a quantum state.  In order to characterize the key cost, we introduce and study another quantity, the {\it key of formation} $K_F$. Informally, $K_F$ is equal to the minimum average amount of key content of a state, where the average is taken over all decompositions of the state into a finite mixture of pure states that are Schmidt-twisted
\cite{HorLeakage,karol-PhD} (see~\eqref{eq:kf1} for a formal definition).  
One of our main results, relating the key cost and the key of formation, is encapsulated in the following inequality
\begin{equation}
    K_C(\rho) \leq K_F^{\infty}(\rho),
    \label{eq:main_char}
\end{equation}
where $K_F^{\infty}(\rho) \coloneqq  \lim_{n\rightarrow \infty}\frac{1}{n}K_F(\rho^{\otimes n})$ (see Theorem~\ref{thm:regularization}).

We note here that key cost and key of formation are analogous
 to the entanglement cost $E_C$ and entanglement of formation $E_F$ in entanglement theory~\cite{BBPSSW1996}, respectively. Furthermore, 
the aforementioned generalized private states form a class of states in the privacy domain that corresponds to the class of pure states in entanglement theory. Also,
the inequality in~\eqref{eq:main_char} is partially analogous to the relation $E_F^{\infty}(\rho)=E_C(\rho)$ from entanglement theory~\cite{Hayden_2001,yamasaki2024entanglement}. Although there are analogies between entanglement and private key~\cite{GW,Collins2002}, we should note that there are important distinctions, some of which have led to profound insights into the nature of quantum information~\cite{HHHLO08,HHHLO08a,science2008smith}. Additionally, the technicalities of dealing with privacy are more subtle and involved than when dealing with entanglement~\cite{pptkey,keyhuge}. For our specific problem considered here, i.e., to obtain~\eqref{eq:main_char}, it was necessary for us to develop techniques beyond what is currently known in entanglement theory. Indeed, our main technical contribution here, for establishing~\eqref{eq:main_char}, involves designing a protocol that dilutes privacy.

The applicability of the introduced quantities goes beyond the resource theory of (device-dependent) private key. Naturally, $K_C$ is a novel entanglement measure. $K_F$ and $K_C$ are also applicable in a scenario in which security is independent of the inner workings of the devices (device-independent security~\cite{DIkeyReview}). Indeed, by the technique of~\cite{CFH21}, the {\it reduced} versions of $K_C$ and $K_F$ are novel measures of Bell non-locality~\cite{Brunner2014}. More importantly, the single-shot version of the reduced $K_C$, that is, the reduced single-shot key cost $K_C^{\varepsilon \downarrow}$, can be treated as the definition of the key cost of a device --- a quantity that has not been considered so far. 

In large part, the quantum internet, when built, will be used to transfer classical data securely. Hence, the right choice of a quantum state with unit cost appears to be a state that has one bit of ideally secure key.
Such states have been introduced already in~\cite{pptkey,keyhuge} and are called private bits. More generally, we can consider $d$-dimensional private states, called private dits. These states, denoted by $\gamma(\Phi^+)$, can be understood as a maximally entangled state
\begin{equation}
\Phi^+_{A_{K}B_{K}}\coloneqq \frac{1}{d_k}\sum_{i=0}^{d_k-1}\ket{ii}\!\bra{jj}_{A_{K}B_{K}},
\end{equation}
where $d_k=\min\{|A_{K}|,|B_{K}|\}$ ($|A_K|$ denotes the dimension of the Hilbert space of $A_K$), which has become ``twisted'' with some arbitrary shared state $\rho_{A_{S}B_{S}}$ by means of a controlled-unitary 
\begin{equation}
\tau_{A_{K}B_{K}A_{S}B_{S}}\coloneqq \sum_{i=0}^{d_k-1} \op{ii}_{A_KB_K}\otimes U_i^{A_{S}B_{S}},    
\end{equation}
leading to
\begin{equation}
    \gamma_{d_k}(\Phi^+) \coloneqq \tau (\Phi^+_{A_{K}B_{K}}\otimes \rho_{A_{S}B_{S}})\tau^{\dagger}.
\end{equation}

Since it has been shown (see~\cite{pptkey,keyhuge}), that a state contains at least $\log_2 d$ bits of secure key accessible by local (incomplete) von Neumann measurements if and only if it is of the above form, it is clear that the private states are resourceful states in the resource theory of private key. However, it is not clear which states are useless~\cite{Chitambar-Gour} (or ``free'' in the language of resource theory), from which no key can be obtained. So far, it is known that separable, i.e., disentangled states, are useless~\cite{precondition}, while it remains open if entangled but key-undistillable states exist. This question is one of quantum information theory's most difficult open problems (see Problem~$24$ on the IQOQI Vienna list~\cite{ViennaOpen}). In this paper, we develop the resource theory of private key up to the current state of the art, i.e., under the assumption that there {\it do not exist} entangled but key-undistillable states.

\section{Summary of results}
\subsection{Irreversibility of privacy}
Informally, the {\it key cost} $K_C(\rho_{AB})$ of a quantum state $\rho_{AB}$  is the minimal ratio of the number $\log_2 d_n$ of key bits, in the form of a private state $\gamma_{d_n}(\Phi^+)$, needed to create $\rho^{\otimes n}$ approximately, to the number $n$ of copies, in the asymptotic limit of large~$n$: it is equal to the minimal value of $\lim_{n\rightarrow \infty} \frac{\log_2 d_n}{n}$ such that the approximation error tends to zero. 
Based on this definition, we also quantify the cost of  devices in terms of secret key~\cite{Brunner2014} (see Section~\ref{sec:behaviors}).

We now make several observations that are direct consequences of the definitions of the key cost and key of formation and facts known already in the literature. First, it is straightforward to see that $K_F \leq E_F$ because the set of pure states is a subset of the set of generalized private states and $K_C\leq E_C$ because a maximally entangled state is also a private state. Furthermore, if each of the local quantum systems $A$ and $B$ is either a qubit or qutrit, i.e., $|A| \in \{2,3\}$ and $|B| \in \{2,3\}$, then $K_F=E_F$. 
This is because, in this case, the shield 
systems $A_{S}B_{S}$ have trivial dimensions equal to one so that the generalized private states coincide with pure states. However, it is not clear if $K_C=E_C$ in this low-dimensional case. On the other hand in larger dimensions, we know that $K_C\neq E_C$ generally, as is the case for the so-called flower state~\cite{smallkey}. This state is a $(2 d_s)^2$-dimensional private state for which $K_C\leq 1$ by definition while $E_C \geq {\log_2 d_s}/{2}$~\cite{locking}. More importantly, by a standard approach, we prove that 
\begin{equation}
K_C(\rho) \geq E_R^{\infty}(\rho)   , 
\end{equation}
(see Theorem~\ref{thm:irrev}), where the regularized relative entropy of entanglement is defined as~\cite{RelEnt}
\begin{align}
E_R^{\infty}(\rho) & \coloneqq \lim_{n\rightarrow \infty}\frac{1}{n} E_R(\rho^{\otimes n}),  \\
E_R(\rho)  & \coloneqq  \inf_{\sigma \in \SEP} \mathrm{Tr}[\rho (\log_2 \rho - \log_2 \sigma)],
\end{align}
with $\SEP$ denoting the set of all separable states of a given bipartite system. Furthermore, there are states $\widehat{\rho}$ for which an entanglement measure called {\it squashed entanglement} $E_{\operatorname{sq}}$ is strictly less than $E_R^{\infty}$~\cite{CSW12,BCHW}. Since both $E_R^{\infty} $ and $E_{\operatorname{sq}}$ are upper bounds on $K_D$~\cite{pptkey,Christandl-phd,keyhuge,Mark_squashed},  this leads to a first example of irreversibility in the quantum cryptographic scenario
\begin{equation}
K_D(\widehat{\rho}) \leq E_{\operatorname{sq}}(\widehat{\rho}) \ll E_R^{\infty}(\widehat{\rho}) \leq K_C(\widehat{\rho}).    
\end{equation}
Hence, the resource of private key joins the family of resources for which irreversibility occurs~\cite{Chitambar-Gour}.

\subsection{The key cost}
To quantify the investment in creating a state, we need to control how much private key a given private state possesses. A cogent choice is the class of irreducible private states, i.e., those for which $K_D(\gamma(\Phi_{d_k}^+))= \log_2 d_k$, such that there is no more key in this state than that contained in its key systems. This class has been characterized in~\cite{IR-pbits} in the following way: if there are no entangled but key undistillable states, all irreducible private states are {\it strictly} irreducible (SIR)~\cite{karol-PhD,FerraraChristandl}, that is, taking the form
 \begin{equation}
     \gamma_{d_k,d_s}(\Phi_+)\coloneqq \frac{1}{d_k} \sum_{i,j=0}^{d_k-1}|ii\>\!\<jj|_{{A_KB_K}}\otimes U_i \rho_{{A_SB_S}} U_j^{\dagger},
     \label{eq:strict_ir_main_text}
 \end{equation}
where $U_i\rho_{{A_SB_S}}U_i^{\dagger}$ is a separable state for each $i$ and~$\rho_{{A_SB_S}}$ is an arbitrary $d_s\times d_s$ state. Since we study the key cost~$K_C$ up to the current state of the art, such that no entangled but key undistillable states are known to exist, we focus on the above class of strictly irreducible states in the definition of the key cost. 
Under this assumption, the definition of the key cost in the asymptotic regime is as follows:
\begin{definition}
The asymptotic key cost $K_{C}(\rho)$ and one-shot key cost $K_{C}^{\varepsilon}(\rho)$ of a state $\rho_{AB}$ are defined as 
\begin{align}
    K_{C}(\rho) & \coloneqq  \sup_{\varepsilon \in (0,1)} \limsup_{n\rightarrow\infty} \frac{1}{n} K_{C}^{\varepsilon}(\rho^{\otimes n}) , \\
    K_{C}^{\varepsilon}(\rho) & \coloneqq \inf_{\substack{\mc{L} \in \LOCC,\\ \gamma_{d_k,d_s} \in \IR}}
    \left\{
    \begin{array}{c}
        \log_2 d_k : \\
        \frac{1}{2} \norm{\mc{L}(\gamma_{d_k,d_s})- \rho}_1\leq \varepsilon
    \end{array}
    \right\},
    \label{eq:def-one-shot-key-cost}
\end{align}
where the infimum in~\eqref{eq:def-one-shot-key-cost} is taken over every LOCC channel $\mathcal{L}$ and every strictly irreducible private state $\gamma_{d_k,d_s}$ with an arbitrarily large, finite shield dimension $d_s\geq 1$. 
\end{definition}

The choice of LOCC in the above definition is justified by the fact that, in the private key distillation scenario, it is known that the following two approaches are equivalent: (i) the distillation of dits of privacy in the form 
$\frac{1}{d_k}\sum_{i=0}^{d_k-1}\op{ii}\otimes \rho_E$ from a tripartite state $\psi_{ABE}^{\otimes n}$ via local operations and public communication (LOPC) and (ii) the distillation of private states $\gamma_{d_k,d_s}(\Phi^+)$ for some~$d_s$, from $\rho_{AB}^{\otimes n}$ using LOCC~\cite{keyhuge,pptkey,karol-PhD}.
While most of the following results hold for the asymptotic definition of the key cost, we prove a yield-cost relation in the one-shot scenario that bounds the one-shot distillable key from above by the one-shot key cost and a small correction factor
 \begin{align}
 K_D^{\varepsilon_1}(\rho)\leq K_C^{\varepsilon_2}(\rho) + \log_2\! \left(\frac{1}{1-(\varepsilon_1+\varepsilon_2)}\right),
 \end{align}
 for $\varepsilon_1+\varepsilon_2 <1$ and $\varepsilon_1,\varepsilon_2\in[0,1]$, where $K_D^{\varepsilon_1}$ is the single-shot version of the distillable key $K_D$
(see Theorem~\ref{thm:second_low}). See~\cite{W21,RyujiRegulaWilde} for similar prior results.

\subsection{The key of formation}
Private states can be naturally generalized to a class of states that correspond to pure bipartite states. This generalization is achieved when one twists a pure bipartite state $\psi$ in such a way that $\tau$ controls the Schmidt bases of $\psi$, thereby obtaining generalized private states~\cite{karol-PhD,HorLeakage}, which need not be pure and are denoted by $\gamma(\psi)$. 
Below we focus on a subclass of the generalized private states, those which are {\it strictly irreducible}. They have a structure that assures irreducibility by construction; that is, they satisfy a condition analogous to the one from~\eqref{eq:strict_ir_main_text} characterising strictly irreducible private states.
We use these states to define the key of formation of a bipartite state $\rho_{AB}$ as an analog of the entanglement of formation
    \begin{equation}
        K_F(\rho_{AB})\coloneqq  \inf_{\{p_k,\gamma(\psi_k)\}_{k=1}^{k_{\operatorname{\max}}}} \sum_{k=1}^{k_{\operatorname{\max}}} p_k S_{A_{K}}(\gamma(\psi_k)),
\label{eq:kf1}
    \end{equation}
    where
    \begin{equation}
    S_{X}(\rho_{XY}) \coloneqq  -\operatorname{Tr}[\rho_X\log\rho_X]     
    \end{equation}
    is the von Neumann entropy of the reduced state $\rho_X \equiv \mathrm{Tr}_Y[\rho_{XY}]$.
    Here $\sum_{k=1}^{k_{\operatorname{\max}}} p_k \gamma(\psi_k)=\rho_{AB}$, and for each $k$, there exist unitary transformations $W_A^{(k)}\colon A\rightarrow A_{K}A_{S}$ for Alice and $W_B^{(k)}\colon B\rightarrow B_{K}B_{S}$ for Bob, such that the key can be obtained by measuring the systems $A_{K}B_{K}$ directly with local von Neumann measurements in the Schmidt bases of $\psi_k$. The shield systems $A_{S}B_{S}$ protect the key from the eavesdropper. Moreover, each state $\gamma(\psi_k)$ has the property that all of its key content is accessible in the $A_{K}B_{K}$ systems via local von Neumann measurements. In particular,  conditioned on the outcomes of this measurement, the systems $A_SB_S$ remain separable; i.e., the state is strictly irreducible.
    In the definition above, $k_{\operatorname{\max}}$ is arbitrarily large but finite. Based on~\cite{Yang2005}, we prove that taking the infimum over $k_{\operatorname{\max}}\leq |A|^2|B|^2+1$ suffices. The relation of the introduced quantities $K_C$, $K_F$, $K_F^\infty$ to the well known entanglement measures such as distillable key, entanglement cost, entanglement of formation, etc., is represented in Figure \ref{fig:entanglement-measures-diagram}.

\subsection{The case of reversibility}
The main result~\eqref{eq:main_char} along with the fact that $K_F(\rho)=\inf_{\{(p_k,\gamma(\psi_k))\}}\sum_k p_k K_D(\gamma(\psi_k))$, combined  with other bounds, implies our second main result. Namely,
\begin{align}
    K_D(\gamma(\psi)) & =E_R(\gamma(\psi))=\nonumber E_R^{\infty}(\gamma(\psi))=K_F^\infty(\gamma(\psi))=\\
       K_F(\gamma(\psi))&=K_C(\gamma(\psi))  = S_{A}(\psi)=S_{A_{K}}(\gamma(\psi))
    \end{align}
    (see Theorem~\ref{thm:key_cost_for_private}). We note that these equalities are appealing, in analogy with what happens for pure states $\psi$ in entanglement theory, for which the distillable entanglement, relative entropy of entanglement, and entanglement cost all coincide and are equal to the von Neumann entropy of the marginal of $\psi$. The difference is that system $A_K$ is the full local subsystem for $\gamma(\psi)$ only when $\gamma(\psi)$ is itself a pure state, which is generically not the case.

    More importantly, the analogy with entanglement theory is only partial. Indeed, the squashed entanglement~\cite{CW04} does not lie between the distillable key and the key cost (see detailed argument given in Remark~\ref{rem:squashed}). It is opposed to the fact that every proper entanglement monotone lies between the distillable entanglement and entanglement cost~\cite{H3}.

\subsection{Dilution protocol}
As mentioned, the core of the upper bound on 
$K_C$ via~\eqref{eq:main_char} is a dilution protocol (DP).
This protocol transforms a private state via LOCC into an approximation of a large number of copies of the desired generalized private state. Our idea behind this protocol is ultimately different from that of~\cite{LoPopescu}, where dilution of entanglement has been proposed using pure entanglement and quantum teleportation~\cite{BBPSSW1996}. A key difference arises because private states have, in general, a low amount of pure entanglement. We partly follow the idea of coherence dilution~\cite{YangWinter}, because generalized private states have a basis that is naturally distinguished by the fact that coherence in that basis assures secrecy of the key.

Our main statement regarding this dilution protocol is encapsulated in the following theorem (Theorem~\ref{thm:transformation}):

\begin{theorem}
\label{thm:pdprot_main_text}
For every generalized private state
$\gamma(\psi)$ with $|\psi\>=\sum_{a\in {\cal A}}\lambda_{a}|e_a\>|f_a\>$, and sufficiently large $n$, the state
$\gamma(\psi)^{\otimes n}$  can be created by LOCC from
$\gamma_{d_n}(\Phi^+)\otimes\Phi^+_{d'_n}$ with error less than $2\varepsilon$ in trace distance, and the number of key bits is equal to
\begin{equation}
\log_2 d_n = \lceil n(S_{A_K}(\psi) +\eta)\rceil,    
\end{equation}
 with $\log_2 d'_n= \lceil 4\delta n\rceil$, where $\eta,\varepsilon,\delta \rightarrow 0$ as $n\rightarrow \infty$.
\end{theorem}

The idea of the proof is as follows. Both parties are allowed at the beginning of the protocol to share an arbitrary private state of their choice. We thus fix their initial private state to be $\gamma_{d_n}(\Phi_+)$ with $d_n\coloneqq  |{\cal A}|^{\lceil n(S_{A_K}(\psi) + \eta) \rceil}$, that has a shield system in the same state and rotated by the same twisting unitary as the target state we want to create,
i.e.,
\begin{equation}
\gamma(\psi)^{\otimes n}=\sum_{\bf{s},\bf{s'}=0}^{|{\cal A}|^{n-1}} \sqrt{\lambda_{\bf{s}}\lambda_{\bf{s'}}} |e_{\bf s}f_{\bf s}\>\!\<e_{\bf s'} f_{\bf s'}| \otimes U_{\bf s}\widetilde{\rho}_{A_SB_S} U_{\bf s'}^{\dagger}.    
\end{equation}
Here
\begin{align}
{\bf s} & =({\bf s}[1],{\bf s}[2],\hdots,{\bf s}[n])\in {\cal A}^n, \\
U_{\bf s} & = U_{{\bf s}[1]}\otimes U_{{\bf s}[2]}\otimes \cdots  \otimes U_{{\bf s}[n]}\\
\widetilde{\rho}_{A_SB_S}& =\rho_{A_SB_S}^{\otimes {n}}.
\end{align}
Alice creates a pure state $|\phi\>_{A''}^{\otimes n}$ with the same amplitudes as the state $|\psi\>^{\otimes n}$ and compresses it, so that it is stored in $\lceil n(S_{A_K}(\psi)+\eta)\rceil$ qubits. She then uses the controlled-XOR operation from $A''$ to $A_K$, measures $A_K$ in the computational basis, and sends the result ${x}$ to Bob. This operation is a coherently applied counterpart of one-time-pad encryption, an analog of quantum teleportation~\cite{Collins2002}. This part of the protocol is similar to what is used in~\cite{YangWinter}, which is only the first part of our dilution protocol here.

The next part of the protocol proceeds as follows.
After Bob corrects his key part $B_K$, shifting his system by $x$, the resulting state
(after decompression on both sides) has the correct amplitudes, but the form of the shield systems is not yet correct. Up to a normalization factor, the overall state is as follows
\begin{align}
\sum_{{\bf s},{\bf s}'}{\sqrt{\lambda_{\bf s}\lambda_{{\bf s}'}}} |{\bf s}{\bf s}\>\!\<{\bf s}'{\bf s}'|\otimes U_{\beta_x({\bf s})}\rho_{A_SB_S}^{\otimes n}U_{\beta_x({\bf s}')}.
\end{align}
Here $\beta_x({\bf s})$ is a certain automorphism on the so-called $\delta$-typical sequences of length $n$, to which $\bf s$ also belongs (for details, see Section~\ref{sec:DPdescription}). The challenging task is then to (in a sense) invert $\beta_x$. To achieve this task, we design two algorithms:  Permutation Algorithm (PA) followed by Phase Error  Correction (PEC) and Inverting~PA (IPA). The idea of PA is to permute the shield systems $A_SB_S$ in order to invert $\beta_x$ in most positions, using the key and the system holding the pattern of errors as a control.
This results in a string ${\bf s}_{\operatorname{cor}}$, which agrees with ${\bf s}$ in most positions. 
The remaining small amount of systems, for which there is still a phase error, i.e., ${\bf s}_{\operatorname{cor}}[i]\neq {\bf s}[i]$,
are teleported to one side in the PEC procedure, corrected locally, and teleported back. This step consumes a nonzero but negligible amount of maximally entangled states. Finally, the Inverting PA procedure ensures that PA and PEC should not leave systems containing valuable information to Eve upon tracing them out. For details of the proof, see Sections~\ref{sec:alg}--\ref{sec:RP}.

\begin{figure}
    \centering
    \includegraphics[width=0.5\linewidth]{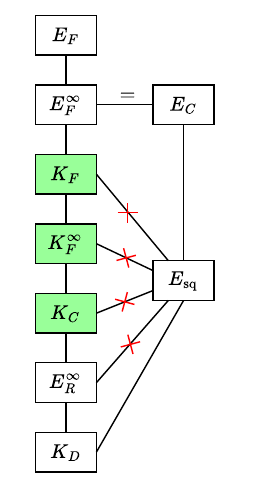}
    \caption{The diagram depicts  relations between well known entanglement measures and the quantities $K_C$, $K_F$, and $K_F^{\infty}$ introduced here (green boxes). A connection between two quantities corresponds to the fact that the higher one is no smaller that the lower one. Red cross means that connected quantities are incomparable. Note however, that the above diagram reduces for generalized private states, as quantities $K_F, K_F^{\infty}, K_C, E_R^{\infty}, K_D$ become equal for such states (see Theorem~\ref{thm:key_cost_for_private}).}
    \label{fig:entanglement-measures-diagram}
\end{figure}

\section{Summary of contents}

We introduce all the necessary preliminaries in Section~\ref{sec:preliminaries_strong_typ}. In Section~\ref{sec:irreversibility_of_privacy_creation}, we prove that the regularised relative entropy of entanglement is a lower bound on the key cost, which enables showing an example of the irreversibility of privacy creation-distillation process. In Section~\ref{sec:key_of_formation}, we define the key of formation and prove some of its basic properties, including monotonicity under certain elementary LOCC operations. Section~\ref{sec:dilution} contains a construction of a Privacy Dilution protocol, which dilutes twisted maximally entangled states to generalized private states. The rate at which privacy dilution can be performed using this protocol is given in the main theorem of this section. In Section~\ref{sec:regularized_key_of_formation_is_cost}, the Privacy Dilution protocol is used to prove that the regularized key of formation is an upper bound on the key cost. Section~\ref{sec:key_cost_of_generalized_private_states} contains proofs of several lemmas, including properties of generalized private states. The main theorem of this section states that the key of formation, the regularized key of formation, the key cost, the distillable key, and the relative and regularized relative entropy of entanglement are equal for strictly irreducible generalized private states. Section~\ref{sec:yield_cost_relation} shows the single-shot yield-cost relation, i.e., that the single-shot key cost is an upper bound (up to a factor) on the distillable key. Section~\ref{sec:behaviors} contains the definition of the key cost of a quantum device. It also introduces the device-independent key cost of a private state. The last section (Section~\ref{sec:subadditivity_of_key_of_formation}) contains the proof of subadditivity of the key of formation.

\section{Preliminaries}\setcurrentname{Preliminaries}
\label{sec:preliminaries_strong_typ}

Let us begin by providing some background on quantum information here and refer the reader to~\cite{BengtssonZyczkowski-book,Wat15} for more details.
The dimension of any quantum system $A$ is defined as the dimension of its Hilbert space $\mathcal{H}_A$. Let $|A|\coloneqq \dim(\mathcal{H}_A)$. The identity operator acting on the Hilbert space $\mathcal{H}_A$ is denoted as $\id_A$. The Hilbert space of a bipartite quantum system $AB$ is $\mathcal{H}_{AB}\coloneqq \mathcal{H}_A\otimes\mathcal{H}_B$, and its dimension is $|AB|=|A||B|$. Let $\mathcal{B}_+(\mathcal{H}_A)$ denote the set of all positive semi-definite operators acting on~$\mathcal{H}_A$. Quantum states are represented by density operators $\rho_A: \mathcal{H}_A\to \mathcal{H}_A$ and they satisfy following three criteria: (i) $\rho=\rho^\dag$, (ii) $\rho\geq 0$, and (iii) $\Tr[\rho]=1$. A pure state is a density operator of rank one, and a state that is not pure is called a mixed state. With slight abuse of nomenclature, we refer to both the pure state $\psi_A=\op{\psi}_A$ and its corresponding state vector $\ket{\psi}_A$, where $\ket{\psi}_A\in\mathcal{H}_A$, as pure states. Let $\mathcal{D}(\mathcal{H}_{A})$ denote the set of all quantum states of the system $A$. Let $\SEP(A\! :\! B)$ denote the set of all separable states of a bipartite quantum system $AB$, i.e., $\rho_{AB}\in\SEP(A\! :\! B)$ if and only if we can express $\rho_{AB}$ as
\begin{equation}
    \rho_{AB}=\sum_{x}p_x\sigma^x_A\otimes\omega^x_B,
\end{equation}
where $\sum_{x}p_x=1$ and for each $x$ we have $0\leq p_x\leq 1$, $\sigma^x_A\in\mathcal{D}(\mathcal{H}_A)$, and $\omega^x_B\in\mathcal{D}(\mathcal{H}_B)$. A bipartite quantum state $\rho_{AB}$ is called entangled if it is not separable, i.e., $\rho_{AB}$ is entangled if $\rho_{AB}\notin\SEP(A\! :\! B)$. 

A quantum channel $\mathcal{N}_{A\to B}: \mathcal{D}(\mathcal{H}_A)\to\mathcal{D}(\mathcal{H}_{B})$ is a completely positive, trace-preserving map. We use $\id_{A\to B}:\mathcal{D}(\mathcal{H}_A)\to\mathcal{D}(\mathcal{H}_{B})$ to denote an identity channel when $|A|=|B|$. The Choi state of a quantum channel $\mathcal{N}_{A\to B}$ is the output state $(\id_{R}\otimes\mathcal{N}_{A\to B})(\Phi^+_{RA})$ of the channel $\id\otimes\mathcal{N}$ when the input state is a maximally entangled state 
\begin{equation}
\Phi^+_{RA}\coloneqq  \frac{1}{d}\sum_{i,j=0}^{d-1}\ket{ii}\!\!\bra{jj}_{RA}, 
\end{equation}
where $d=\min\{|R|,|A|\}$ is the Schmidt rank of $\Phi^+_{RA}$ and $\{\ket{i}\}_{i}$ forms a set of orthonormal kets (vectors). An isometry is a linear map $U:\mathcal{H}_A\to \mathcal{H}_B$ such that $U^\dag U=\id_A$ and $UU^\dag=\Pi_B$, where $\Pi_B$ is a projection on the support of $\mathcal{H}_B$. A positive-operator valued measure (POVM) is a set $\{\Lambda^x\}_x$ of positive semi-definite operators, i.e., $\Lambda^x\geq 0$ for each $x$, such that $\sum_x\Lambda^x=\id$. A quantum instrument~$\Delta$ is a collection $\{\mathcal{M}^x\}_x$ of completely positive, trace-nonincreasing maps such that the sum map $\sum_x\mathcal{M}^x$ is a quantum channel
\begin{equation}
    \Delta_{A'\to AX}(\cdot)=\sum_x\mathcal{M}^x_{A'\to A}(\cdot)\otimes\op{x}_X,
\end{equation}
where quantum system $X$ represents a classical register, because $\{\op{x}\}_x$ forms a set of orthogonal pure states, which are perfectly discriminable.

The fidelity of $\rho_A,\sigma_A\in\mathcal{B}_+(\mathcal{H_A})$ is defined as~\cite{U76}
\begin{equation}
    F(\rho,\sigma)\coloneqq \norm{\sqrt{\rho_A}\sqrt{\sigma_A}}_1^2 ,
\end{equation}
where $\norm{\zeta}_1\coloneqq \Tr\sqrt{\zeta^\dag\zeta}$ is the trace norm for an arbitrary bounded operator $\zeta$. 

An LOCC channel $\mathcal{L}_{A'B'\to AB}$ is a bipartite quantum channel that can be expressed in the separable form as $\sum_{x}\mathcal{E}^x_{A'\to A}\otimes\mathcal{F}^x_{B'\to B}$ where each $\mathcal{E}^x$ and $\mathcal{F}^x$ are completely positive, trace-nonincreasing maps such that the sum map $\mathcal{L}$ is a quantum channel~\cite{CLM+14,Wat15}. Here pairs $A', A$ and $B',B$ are held by two spatially separated parties, say Alice and Bob. However, note that not every separable channel is an LOCC channel~\cite{BDFMRSSW99}. Using only an LOCC channel, Alice and Bob can only create separable states. In what follows, $\operatorname{LOCC}$ will denote also the set of all local (quantum) operations and classical communication (LOCC) channels.

Let $\mathscr{B}$ be the set of all bipartite quantum states $\rho_{AB}$ for all $|A|,|B|\geq 2$. A function $\mathbf{E}:\mathscr{B}\to \mathbb{R}$ is called an \textit{entanglement monotone} if it is monotonic (nonincreasing) under the action of an arbitrary LOCC channel, i.e., $\mathbf{E}(\rho_{A'B'})\geq \mathbf{E}(\mathcal{L}_{A'B'\to AB}(\rho_{A'B'}))$ where pairs of quantum systems $A',A$ and $B',B$ are held at spatially separated labs, say Alice and Bob, respectively. A direct consequence of this definition is that entanglement monotones are
\begin{enumerate}
    \item Invariant under the action of local isometries,
    \begin{equation}
       \mathbf{E}(\rho_{A'B'})= \mathbf{E}((U\otimes V)\rho_{A'B'}(U^\dag\otimes V^\dag)),
    \end{equation}
    for all local isometries $U_{A'\to AE}$ and $V_{B'\to BF}$, and
    \item Invariant under appending of quantum states via local operations,
    \begin{equation}\label{eq:apstate}
        \mathbf{E}(\rho_{AB})=\mathbf{E}(\sigma_{A'}\otimes\rho_{AB}\otimes\omega_{B'}),
    \end{equation}
    for all quantum states $\sigma_{A'}$ and $\omega_{B'}$ appended locally. 
\end{enumerate}

We also need fundamental results from classical information theory (see, e.g.,~\cite{EK11,Cover2006}). Consider a random variable $X$ with probability distribution $\{p_X(x)\}_{x \in \mathcal{X}}$, where $\mathcal{X}$ is a finite alphabet. Let $x^n$ be a sequence drawn from $\mathcal{X}$. The empirical probability mass function of $x^n$, also called its type, is given by
\begin{equation}
    \pi(x|x^n)\coloneqq \frac{\abs{\{i:x_i=x\}}}{n} \ \text{for}\ x\in\mathcal{X}. 
\end{equation}
For $\delta\in(0,1)$, we define the strongly $\delta$-typical  set as
\begin{align}
    T_n^{\delta}(X)\coloneqq 
 \{x^n: \abs{\pi(x|x^n)-p(x)}\leq \delta p(x) \ \forall x\in\mathcal{X} \}.
 \label{eq:def-str-typ}
\end{align}
The size $\abs{T_n^{\delta}(X)}$ of the strongly $\delta$-typical set  is approximately $2^{nH(X)}$, where
\begin{equation}
H(X)\coloneqq -\sum_{x\in\mathcal{X}}p_X(x)\log_2 p_X(x)    
\end{equation}
is the Shannon entropy of $X$~\cite{Sha48}; to be more precise, for all sufficiently large $n$
\begin{equation}
    (1-\delta)2^{n[H(X)-\eta(\delta)]}\leq \abs{T^\delta_n(X)}\leq  2^{n[H(X)+\eta(\delta)]},
    \label{eq:def-str-typ-size}
\end{equation}
where $\eta(\delta)\coloneqq \delta H(X)$, so that $\eta(\delta)\to 0$ as $\delta\to 0$. The probability that a randomly selected  sequence $x^n$ belongs to the strongly $\delta$-typical set $T_n^{\delta}(X)$ converges to one in the asymptotic limit $n\to \infty$; i.e.,
\begin{equation}
\lim_{n\to\infty} {\rm Pr}(X^n\in T^{\delta}_n(X))=1.    
\end{equation}

We will often use shorthand notations or suppress system labels and identity operations for simplicity, whenever the meaning is clear from the context.

\subsection{Entropies and generalized divergence}

The quantum entropy, also called von Neumann entropy, of a quantum system $A$ in a state $\rho_A\in\mathcal{D}(\mathcal{H}_A)$ is defined as
\begin{equation}
    S(A)_{\rho}\coloneqq S(\rho_A)=-\Tr[\rho_A\log_2\rho_A].
\end{equation}
For the sake of readability, we may use the notation $S_A(\rho)$ instead of the above.
 It is known that the quantum entropy satisfies the inequalities $0\leq S(\rho_A)\leq \log_2|A|$. The conditional quantum entropy, also called the conditional von Neumann entropy, $S(A|B)_{\rho}$ of a bipartite quantum state $\rho_{AB}$ conditioned on quantum system $B$ is given by
\begin{equation}
    S(A|B)_{\rho}\coloneqq S(AB)_{\rho}-S(B)_{\rho}.
\end{equation}
It is known that the conditional quantum entropy can be negative for an entangled state $\rho_{AB}$~\cite{CA97}. 

The quantum mutual information $I(A;B)_{\rho}$ of a quantum state $\rho_{AB}$ is defined as
\begin{equation}
I(A;B)_{\rho}\coloneqq S(A)_{\rho}+S(B)_{\rho}-S(AB)_{\rho}.    
\end{equation}
The quantum conditional mutual information $I(A;B|C)_{\rho}$ is the mutual information between $A$ and $B$ conditioned on $C$
\begin{equation}
    I(A;B|C)_{\rho}\coloneqq S(A|C)_{\rho}+S(B|C)_{\rho}-S(AB|C)_{\rho}.
\end{equation}
Both the quantities $I(A;B)_{\rho}$ and $I(A;B|C)_{\sigma}$ for all quantum states $\rho_{AB}$ and $\sigma_{ABC}$, respectively, are always non-negative~\cite{LR73}. The squashed entanglement $E_{\operatorname{sq}}(\rho_{AB})$ of a bipartite quantum state $\rho_{AB}$ is defined as~\cite{CW04}
\begin{equation}
    E_{\operatorname{sq}}(\rho_{AB})\coloneqq \inf_{\substack{\rho_{ABE}:\\ \Tr_E[\rho_{ABE}]=\rho_{AB}}} \frac{1}{2}I(A;B|E)_{\rho},
\label{eq:sq-ent-def}
\end{equation}
where the infimum is taken over every state $\rho_{ABE}$ that is an extension of $\rho_{AB}$.

A function $\mathbf{D}:\mathcal{D}(\mathcal{H}_A)\times\mathcal{D}(\mathcal{H}_A)\to \mathbb{R}$ is called a generalized divergence~\cite{PV10} if it satisfies the data-processing inequality for all quantum states $\rho,\sigma\in\mathcal{D}(\mathcal{H}_A)$ and an arbitrary quantum channel $\mathcal{N}_{A\to B}$
\begin{equation}
    \mathbf{D}(\rho\Vert\sigma)\geq \mathbf{D}(\mathcal{N}(\rho)\Vert\mathcal{N}(\sigma)).
\end{equation}
It directly follows from the above definition that any generalized divergence remains invariant under the following two operations~\cite{WWY14}: the action of an isometry~$U$ and appending of quantum states, i.e.,
\begin{align}
     \mathbf{D}(\rho\Vert\sigma) &= \mathbf{D}(U\rho U^\dag\Vert U\sigma U^\dag),\\
      \mathbf{D}(\rho\Vert\sigma) &= \mathbf{D}(\rho\otimes\omega \Vert \sigma\otimes\omega),
\end{align}
for an arbitrary isometric operator $U$ and an arbitrary quantum state $\omega$. Examples include the quantum relative entropy~\cite{Ume62}, defined for quantum states $\rho,\sigma$ as
\begin{equation}
D(\rho\Vert\sigma)\coloneqq \Tr[\rho (\log_2 \rho - \log_2 \sigma)],
\end{equation}
when $\supp(\rho)\subseteq \supp(\sigma)$ and otherwise $D(\rho\Vert\sigma)\coloneqq +\infty$. Another example is the sandwiched R\'enyi relative entropy~\cite{MDSFT13,WWY14}, which is denoted as $\wt{D}_\alpha(\rho\Vert\sigma)$ and defined for states $\rho,\sigma$, and  $ \alpha\in (0,1)\cup(1,\infty)$ as
\begin{equation}\label{eq:def_sre}
\wt{D}_\alpha(\rho\Vert \sigma)\coloneqq  \frac{1}{\alpha-1}\log_2 \Tr\!\left[\left(\sigma^{ (1-\alpha) / 2\alpha }\rho\sigma^{ (1-\alpha) / 2\alpha }\right)^\alpha \right] ,
\end{equation}
but it is set to $+\infty$ for $\alpha\in(1,\infty)$ if $\supp(\rho)\nsubseteq \supp(\sigma)$. The sandwiched R\'enyi relative entropy $\wt{D}_\alpha(\rho\Vert\sigma)$ of two states is nonincreasing under the action of a quantum channel $\mathcal{N}$ for all $\alpha\in[\frac{1}{2},1)\cup(1,\infty)$~\cite{FL13,W18opt}, i.e.,
\begin{equation}
\wt{D}_\alpha(\rho\Vert \sigma)\geq \wt{D}_\alpha(\mathcal{N}(\rho)\Vert \mathcal{N}(\sigma)).
\end{equation}
It is also monotone in~$\alpha$~\cite{MDSFT13}; that is, 
\begin{equation}
    \wt{D}_\alpha(\rho\Vert\sigma)\geq \wt{D}_\beta(\rho\Vert\sigma)
\end{equation}
for all $\alpha,\beta\in(0,1)\cup(1,\infty)$ such that $\alpha\geq \beta $. 

In the limit $\alpha\to 1$ the sandwiched R\'enyi relative entropy converges to the quantum relative entropy, and in the limit  $\alpha\to \infty$, it converges to the max-relative entropy~\cite{MDSFT13}, which is defined as~\cite{D09,Dat09}
\begin{equation}\label{eq:max-rel}
D_{\max}(\rho\Vert\sigma)\coloneqq \inf\{\lambda\in\mathbb{R}:\ \rho \leq 2^\lambda\sigma\},
\end{equation}
and if $\supp(\rho)\nsubseteq\supp(\sigma)$ then $D_{\max}(\rho\Vert\sigma)\coloneqq + \infty$. It is also known that $I(A;B)_{\rho}=D(\rho_{AB}\Vert \rho_A\otimes\rho_B)$, where $\rho_A$ and $\rho_B$ are reduced states of $\rho_{AB}\in\mathcal{D}(\mathcal{H}_{AB})$.

Another generalized divergence is the $\varepsilon$-hypothesis-testing divergence~\cite{BD10,WR12}, defined as
\begin{equation}
D^\varepsilon_h(\rho\Vert\sigma)\coloneqq -\log_2\inf_{\Lambda: 0\leq\Lambda\leq \id}\{\Tr[\Lambda\sigma]:\  \Tr[\Lambda\rho]\geq 1-\varepsilon\},
\label{def:d_h}
\end{equation}
for $\varepsilon\in[0,1]$ and density operators $\rho$ and $\sigma$. The $\varepsilon$-hypothesis testing relative entropy for $\varepsilon=0$ reduces to the min-relative entropy, i.e.,
\begin{equation}\label{eq:min-h}
    D_{\min}(\rho\Vert\sigma)\coloneqq D^{\varepsilon=0}_h(\rho\Vert\sigma) = -\log_2\operatorname{Tr}[\Pi_\rho \sigma],
\end{equation}
where $\Pi_\rho$ denotes the projection onto the support of $\rho$ (see \cite[Appendix~A-3]{wang_2019} for a proof of this equality).

Some other well-known generalized divergences include the trace distance $\norm{\rho-\sigma}_{1}$ and the negative of fidelity $-F(\rho,\sigma)$.

\subsection{Private states and generalized private states}

A tripartite key state $\gamma_{ABE}$ contains $\log_2 d_k$ bits of secret key, shared between $A$ and $B$, against an eavesdropper possessing $E$, if there exists a state $\sigma_E$ such that \begin{align}
   & (\mathcal{M}_A\otimes\mathcal{M}_B)(\gamma_{ABE}) \nonumber\\
    &\quad\quad = \frac{1}{d_k}\sum_{i=0}^{d_k-1}\op{i}_A\otimes\op{i}_B\otimes\sigma_E
\end{align} 
for a projective measurement channel $\mathcal{M}(\cdot)=\sum_{i}\op{i}(\cdot)\op{i}$, where $\{\ket{i}\}_{i=0}^{d_k-1}$ forms an orthonormal basis. The quantum systems $A$ and $B$ are maximally classically correlated, and the key value is uniformly random and independent of $E$.

A (perfect) bipartite private state $\gamma(\Phi^+)$ containing at least $\log_2d_k$ bits of secret key has the following form~\cite{keyhuge,pptkey,karol-PhD}
\begin{equation}
    \gamma_{A_KA_SB_KB_S}(\Phi^+)=\tau(\Phi^+_{A_KB_K}\otimes\rho_{A_SB_S})\tau^{\dag}, 
\end{equation}
where $\Phi^+_{A_KB_K}$ is a maximally entangled state of Schmidt rank $d_k$ and the state $\rho_{A_SB_S}$ of the shield systems gets twisted by $\tau_{A_KB_KA_SB_S}$, which is a controlled-unitary operation given as
\begin{equation}
    \tau_{A_KB_KA_SB_S}=\sum_{i=0}^{d_k-1}\op{ii}_{A_KB_K}\otimes U_{A_SB_S}^i ,
\end{equation}
where $U_{A_SB_S}^i$ is an arbitrary unitary operator  for each $i$. It is understood that $A_K, A_S$ are held by Alice and $B_S, B_S$ are held by Bob, both being safe against a quantum adversary.

A special class of private states for which $K_D(\gamma(\Phi^+))= \log_2 d_k$ (i.e., there is no more key in this state than that contained in its key systems) are called irreducible private states. This class has been characterized in~\cite{IR-pbits} in the following way: if there are no entangled but key undistillable states, all irreducible private states are strictly irreducible (SIR)~\cite{karol-PhD,FerraraChristandl}, that is, taking the form
 \begin{equation}
     \gamma_{d_k,d_s}(\Phi_+)\coloneqq \frac{1}{d_k} \sum_{i,j=0}^{d_k-1}|ii\>\!\<jj|_{\operatorname{key}}\otimes U_i \rho_{\operatorname{shield}} U_j^{\dagger},
     \label{eq:strict_ir}
 \end{equation}
 where $U_i\rho_{\operatorname{shield}}U_i^{\dagger}$ is a separable state for each $i$ and $\rho_{\operatorname{shield}}$ is an arbitrary state defined over the Hilbert space of dimension $d_s\times d_s$.

 Private states $\gamma(\Phi^+)$ can be naturally generalized to a class of states in which $\Phi^+_{A_KB_K}$ is replaced with an arbitrary pure bipartite state $\psi_{A_SB_S}$, i.e.,\begin{equation}
     \gamma_{A_KA_SB_KB_S}(\psi)=\tau(\psi_{A_KB_K}\otimes\rho_{A_SB_S})\tau^{\dag},
 \end{equation}
 where $\tau_{A_KA_SB_KB_S}$ is a controlled-unitary operation (twisting unitary) defined earlier that twists the state $\psi_{A_KB_K}$ of the key system with the state $\rho_{A_SB_S}$ of the shield system. These states $\gamma(\psi)$ have been called Schmidt-twisted pure states in~\cite{HorLeakage,karol-PhD}. They are, however, pure if and only if $\rho_{A_SB_S}$ is pure; else, they are mixed. Therefore, we will refer to them here as {\it generalized private states}. These states are of the following form
\begin{align}
   \gamma_{A_KA_SB_KB_S}(\psi)\coloneqq  \widetilde{\tau}( \psi_{A_KB_K}\otimes \rho_{A_SB_S}) \widetilde{\tau}^{\dagger},
   \label{eq:Stpb}
\end{align}
where
\begin{align}
 \psi_{A_KB_K} & =\op{\psi}_{A_KB_K},\\
 |\psi\>_{A_KB_K} & = \sum_{i=0}^{d_k-1}\sqrt{\mu_i} |e_i f_i\>_{A_KB_K},
\end{align}
is a pure bipartite state of Schmidt-rank $d_k$ and a twisting unitary 
\begin{equation}
\widetilde{\tau}_{A_KA_SB_KB_S}=\sum_{i=0}^{d_k-1}\op{e_if_i}_{A_KB_K}\otimes U^{i}_{A_SB_S}
\end{equation}
has control in the orthonormal bases $\{|e_i\>\}_i$ and $\{|f_i\>\}_i$, while the state $\rho$ is arbitrary. In~\cite{HorLeakage}, a class of generalized private states $\gamma(\psi)$ called {\it irreducible generalized private states} (called there Schmidt-twisted pure states) were introduced. We relax this class here, by allowing all states that are locally unitarily equivalent to the generalized private states to fall into the same class of generalized private states. 

\begin{definition}
    A generalized private state $\gamma(\psi)_{AB}$ is called irreducible if there exists a unitary transformation
    \begin{align}
    W_{AB}\coloneqq W_{A\rightarrow A_KA_S}\otimes W_{B\rightarrow B_KB_S},  
    \end{align}
    such that the following two conditions are satisfied: 
    \begin{multline}
    W_{AB} \gamma(\psi)_{AB} W_{AB}^{\dagger} = \\
    \sum_{i,j} \sqrt{\mu_i\mu_j}|e_i f_i\>\!\<e_jf_j|_{A_KB_K}\otimes U_i\rho_{A_SB_S} U_j^{\dag},
    \end{multline}
    and
    \begin{equation}
    S_{A_K}(W_{AB} \gamma(\psi) W_{AB}^{\dagger}) = K_D(\gamma(\psi)).
    \end{equation}
    \label{def:gen_p_states}
\end{definition}
We further define strictly irreducible generalized private states.
\begin{definition}
    The generalized private state, equivalent up to local unitary transformations to a state
    \begin{multline}
    \gamma_{A_KB_KA_SB_S}(\psi) = \\
    \sum_{i,j} \sqrt{\mu_i\mu_j}|e_i f_i\>\!\<e_jf_j|_{A_KB_K}\otimes U_i\rho_{A_SB_S} U_j^{\dag}
    \end{multline}
    is called strictly irreducible if $\rho_i = U_i \rho_{A_SB_S} U_i^{\dagger}\in \SEP(A_S\! :\! B_S)$ for each $i$. We denote the set of all strictly irreducible generalized private states as GSIR.
    \label{def:gen_sir_states}
\end{definition}

\section{Irreversibility of privacy creation and distillation}\setcurrentname{Irreversibility of privacy creation and distillation}
\label{sec:irreversibility_of_privacy_creation}

Before we state the main result of this section, we introduce our definition of the asymptotic key cost $K_{C}(\rho)$ and one-shot key cost $K_{C}^{\varepsilon}(\rho)$ of a bipartite quantum state $\rho_{AB}$.

\begin{definition}
The asymptotic key cost $K_{C}(\rho)$ and one-shot key cost $K_{C}^{\varepsilon}(\rho)$ of a state $\rho_{AB}$ are defined as 
\begin{align}
K_{C}(\rho) & \coloneqq  \sup_{\varepsilon \in (0,1)} \limsup_{n\rightarrow\infty} \frac{1}{n} K_{C}^{\varepsilon}(\rho^{\otimes n}) , \\
K_{C}^{\varepsilon}(\rho) & \coloneqq \inf_{\substack{\mc{L} \in \LOCC,\\ \gamma_{d_k,d_s} \in \IR}} \left\{
\begin{array}[c]{c}
    \log_2 d_k :\\
    \frac{1}{2}\norm{\mc{L}(\gamma_{d_k,d_s})- \rho}_1\leq \varepsilon
\end{array}
\right\},
\end{align}
where the infimum is taken over every LOCC channel~$\mathcal{L}$ and every strictly irreducible private state $\gamma_{d_k,d_s}$ with arbitrarily large, finite shield dimension $d_s\geq 1$. 
\end{definition}

To exhibit irreversibility in the privacy distillation-creation process, we use standard techniques to prove that the regularized relative entropy of entanglement bounds the key cost from below.

\begin{theorem}
\label{thm:irrev}
For every bipartite state $\rho_{AB}$, the following bound holds
	\begin{equation}
	K_{C}(\rho) \geq E^{\infty}_R(\rho),
	\end{equation}
	where  
    $E_R^{\infty}(\rho) \coloneqq \lim_{n\rightarrow \infty}\frac{1}{n} E_R(\rho^{\otimes n})$ and 
    $E_R(\rho)  \coloneqq  \inf_{\sigma \in \SEP} D(\rho\|\sigma)$.
\end{theorem}

\begin{proof}
We follow here some standard arguments used in quantum resource theories. Let us begin by giving the basic idea of the proof.
We first note that the relative entropy of entanglement of an irreducible private state~$\gamma_k$
is bounded from above by $k$, i.e., $E_R(\gamma_k) \leq k$ (see Theorem~2 in~\cite{IR-pbits}).
We then note that under the action of an LOCC channel $\mc{L}$ that approximately generates the target state $\rho^{\otimes n}$, the relative entropy does not increase, and for a separable state $\sigma$,
$\mc{L}(\sigma) =\sigma'\in \SEP$.  Furthermore, asymptotic continuity 
of the relative entropy of entanglement corresponds to the following inequality
\begin{equation}
    |E_R(\mc{L}(\gamma_k)) - E_R(\rho^{\otimes n})| \leq f(\varepsilon,d,n),
\end{equation}
where
\begin{align} 
f(\varepsilon,d,n) & \coloneqq \varepsilon \log_2 d^n + g(\varepsilon),\\
    g(\varepsilon) & \coloneqq (\varepsilon+1) \log_2(\varepsilon+1) - \varepsilon \log_2 \varepsilon.
\end{align}
Regularization of the considered quantity by $n$ finishes
the proof.

We state these steps formally now. Consider that
\begin{align}
\log_2 d_k & \geq E_R(\gamma_{d_k}) \\ 
& \geq E_R(\mc{L}(\gamma_{d_k})) \\
& \geq E_R(\rho^{\otimes n})-f(\varepsilon,d,n),
\end{align}
where we have applied \cite[Theorem~2]{IR-pbits} for the first inequality, LOCC monotonicity of the relative entropy of entanglement for the second inequality, and asymptotic continuity of the relative entropy of entanglement \cite[Corollary~8]{W16} for the third.
Since the above inequalities hold for each $\mc{L}\in \operatorname{LOCC}$ and $\gamma_{d_k}\in \operatorname{SIR}$ such that $\frac{1}{2}\left\|\mc{L}(\gamma_{d_k}) - \rho^{\otimes n}\right\|_1 \leq \varepsilon$ and $K_C^{\varepsilon}$ is defined to be the infimum of $\log_2 d_k$ over all such $\mc{L}$ and $ \gamma_{d_k}$, then we conclude that
\begin{align}
   K_C^{\varepsilon}(\rho^{\otimes n}) \geq 
    E_R(\rho^{\otimes n}) -f(\varepsilon,d,n) .
    \end{align}
After dividing by $n$ and applying the definition of key cost and the regularized relative entropy of entanglement, we conclude that
\begin{align}
   K_C(\rho) = \sup_{\varepsilon \in (0,1) }\limsup_{n\to\infty} \frac{1}{n}K_C^{\varepsilon}(\rho^{\otimes n}) & \geq E^{\infty}_R(\rho).
\end{align}
This concludes the proof.
\end{proof}

As an important consequence, it follows that  the process of privacy creation is irreversible. 
Indeed in~\cite{CSW12} (see also~\cite{BCHW}) it was shown that there are states $\widehat{\rho}$ satisfying $E_{\operatorname{sq}} \ll E_R^{\infty}$, where 
$E_{\operatorname{sq}}$ is the squashed entanglement measure defined in~\eqref{eq:sq-ent-def}. We then have 
\begin{equation}
K_D(\widehat{\rho})\leq E_{\operatorname{sq}}(\widehat{\rho})\ll  E_R^{\infty}(\widehat{\rho})\leq K_{C}(\widehat{\rho}),
\end{equation}
which implies that $K_D(\widehat{\rho})\ll K_{C}(\widehat{\rho})$.
The above  inequality  is the first example of irreversibility in the case of privacy. 

\begin{remark}
We note here that a simplified version of the definition of $K_C$ is to consider only those states that take the form of the so-called Bell private states~\cite{FerraraChristandl} whenever they are strictly irreducible.
A Bell private state takes the form
\begin{equation}
    \sum_{i} p_i \op{\Phi_i}_{A_KB_K}\otimes\rho^{(i)}_{A_SB_S}
\end{equation}
where $\rho_{A_SB_S}^{(i)}$ are arbitrary, yet mutually orthogonal states on $A_SB_S$, and the states $\op{\Phi_i}$ are Bell states of local dimension $d_k$ generated
from $|\Phi\> =\frac{1}{\sqrt{d_k}}\sum_{i=0}^{d_k-1} |ii\>$ by the local rotations
\begin{equation}
    X^{m}{ Z}^j = \left(\sum_{l'=0}^{d_k-1} |m\oplus l'\>\!\<l'|\right)\left(\sum_{l=0}^{d_k -1} e^{2\pi i j l/d_k}|l\>\!\<l|\right).
\end{equation}
This simplification comes from the fact (see~\cite[Lemma 1]{FerraraChristandl}) that a one-way LOCC channel can reversibly transform any private state to a Bell private state. 
\end{remark}

\begin{remark} \label{rem:broaden-Kc-definition}
The existence of entangled but zero distillable key states would result in a redefinition of the key cost to a novel quantity, say $K_C'$, which is smaller
than our key cost $K_C$, since it would be based on taking the infimum over a larger class of private states, those with conditional states that are not
necessarily separable.
\end{remark}

\section{Key of formation}\setcurrentname{Key of formation}
\label{sec:key_of_formation}

Here, we generalize the notion of entanglement of formation to the case
of the key of formation. We will need
a notion corresponding to a pure state
in the domain of private states. For this purpose,
we consider irreducible generalized private states $\gamma_{A_KA_SB_KB_S}(\psi)$ (see Eq.~\eqref{eq:Stpb}). We denote such a state by $\gamma(\psi)$, where $\psi$ is the pure state under consideration.

We note that subsystem $A_KB_K$, i.e., the key part of a private state, is not the same for different private states as they can differ by some product $W_{A_KA_S}\otimes W_{B_KB_S}$ of local unitary
transformations,  and the key part of one private state can be of a different dimension than the one for the other private state.
For this reason, by the {\it key} $A_{K}$ (similarly $B_{K}$), we will mean the one from which after a von Neumann measurement in some product basis $\{|e_i\>\}_i$ ($\{|f_i\>\}_i$) the whole key of the generalized private state given in Eq.~\eqref{eq:Stpb}
is obtained. 
The key of formation is then defined as follows:
 \begin{definition} 
 \label{def:formation}
 The key of formation of a bipartite state~$\rho$ is defined as
\begin{align}
    K_F(\rho)\coloneqq \inf_{\scriptsize
    \left\{\begin{array}{c}
        \sum_{k=1}^K p_k\gamma(\psi_k)=\rho,\\
        \gamma(\psi_k) \in \operatorname{GSIR}
    \end{array}\right\}} \sum_{k=1}^K p_kS_{A_{K}}[\gamma(\psi_k)],
    \label{eq:kf}
\end{align}
where $\gamma(\psi_k)$ is a strictly irreducible generalized private state and $S(\cdot)$ denotes the von Neumann entropy. Additionally, the infimum is taken over finite, yet arbitrarily large  $K$.
\end{definition}
The definition provided above directly parallels the definition of the entanglement of formation, but it is applicable to the private key resource theory developed here. Intuitively, $K_F$ quantifies the minimal average amount of (i) directly accessible, (ii) ideally secure correlations, where the minimum runs over the set of states having this property. 
Here, by ``directly accessible,'' we mean the correlations that can be generated by two local, in general incomplete, von Neumann measurements. Furthermore, we consider ideally secure correlations, i.e., decoupled from the purifying system of a purification of $\rho$ held by the eavesdropper, yet possibly biased. They generalize notion of secret key, which has the outcomes of the mentioned von Neumann measurements necessarily equally likely.
Generalized private states satisfy the above two properties (i)-(ii); thus, they are a natural choice in the definition.
However, we need to narrow this class of states to irreducible ones only, since, without this requirement, the proposed definition would be trivially equal to zero. Indeed, all the secret correlations could be considered as a part of the shield subsystem, which is not taken into account in the definition. Finally, we decided to consider only strictly irreducible generalized private states because this class is precisely characterized in~\cite{HorLeakage}. 

Let us note that the defined key of formation has an operational interpretation. Namely, it quantifies the minimal average amount of ideally secure, directly accessible correlations required to generate a state by sampling from its ensemble. 

As a first step, we show that in the above definition, the infimum can be taken over ensembles with at most $(|A||B|)^2+1$ elements. This technical result is crucial in proving
the main result in Theorem~\ref{thm:regularization}.

\begin{lemma}
\label{lem:finite}
For a state $\rho_{AB}$, the infimum in Definition~\ref{def:formation} of $K_F(\rho_{AB})$ can 
be taken on ensembles with a number of elements at most $(|A|\times |B|)^2+1$.
\end{lemma}

\begin{proof}
We closely follow the idea of Proposition~3 in~\cite{Yang2005}. Consider a decomposition of $\rho_{AB}$ into a finite number of states, as $\rho_{AB} =  \sum_{k=1}^{K}p_k\gamma(\psi_k)$, with $K>D+1$ and $D\coloneqq (|A|\times |B|)^2$. We aim to find a decomposition $\rho_{AB} = \sum_{k=1}^{D+1} q_k \widetilde{\gamma}(\psi_k)$, where $\{q_k\}_k$ is a probability distribution. One can construct a polytope $C(S)$, which is the convex hull of the set
\begin{equation}
S\coloneqq \{(\gamma(\psi_k),S_{A_K}(\gamma(\psi_k))\}_{k=1}^K.     
\end{equation}
The set $C(S)$ is convex by definition and a proper subset of $\mathbb{R}^{D}$. This holds because private states are spanned by Hermitian operators and have unit trace, and one more parameter corresponds to the value of $S_{A_K}(\gamma(\psi_k))$. Now, by definition
\begin{equation}
y \coloneqq \left(\sum_{k=1}^K p_k \gamma(\psi_k) ,\sum_{k=1}^{K} p_kS_{A_K}(\gamma(\psi_k))\right)\in C(S). 
\end{equation}
It is then important to note that the extremal points of $C(S)$ are contained within the set $S$. Hence by Caratheodory's theorem, there exist at most $D+1$ extremal points $\{(\widetilde{\gamma}(\psi_j),S_{A_K}(\widetilde{\gamma}(\psi_j))\}_{j=1}^{D}\subseteq S$ in $C(S)$ such that
$y= \sum_{j=1}^{D+1}q_j(\widetilde{\gamma}(\psi_j),S_{A_K}(\widetilde{\gamma}_j))$, which implies the statement of the lemma. 
\end{proof}

Note that, while we show above that the infimum in the definition of $K_F$ can be taken over ensembles of bounded size, it remains open to determine if the infimum is attained. In what follows, we will use that $K_F$, constructed by the convex-roof method, is a convex function. We prove it below for the sake of the completeness of the presentation.

\begin{lemma}
    $K_F$ is a convex function.
    \label{lem:kf_is_conv}
\end{lemma}
\begin{proof}
    Let $\sum_k^K p_k \rho_k$ be a convex mixture of quantum states and let $\eta_1, \ldots, \eta_K>0$. Since $K_F$ is by the definition an infimum over ensembles into generalized strictly irreducible private states, then for each $\rho_k$ there exist an ensemble $\{(q_{kl}, \gamma_{kl})\}_l$ such that $\sum_{l} q_{kl} \gamma_{kl} = \rho_k$ and $\sum_{l} q_{kl} S_{A_K}[\gamma_{kl}] \leq K_F(\rho_k) + \eta_k$. Now observe that $\sum_{k,l} p_k q_{kl} \gamma_{kl}$ is a valid decomposition into generalized strictly irreducible private states of the state $\sum_k^K p_k \rho_k$. Again, since $K_F$ is an infimum we obtain
    \begin{align}
        &K_F\!\left(\sum_k p_k \rho_k\right) \leq \sum_{k,l} p_k q_{kl} S_{A_K}[\gamma_{kl}] \nonumber\\
        &\leq \sum_{k} p_k (K_F(\rho_k) + \eta_k) = \sum_k p_k K_F(\rho_k) + \sum_k p_k\eta_k.
    \end{align}
    Since it holds true for every $\eta_1, \ldots, \eta_K>0$,  it follows that
    \begin{equation}
        K_F\!\left(\sum_k p_k \rho_k\right) \leq \sum_k p_k K_F(\rho_k),
    \end{equation}
    concluding the proof.
\end{proof}

\subsection{Monotonicity of the key of formation under some elementary operations}\setcurrentname{Monotonicity of the key of formation under some elementary operations}
\label{sec:monotonicity}

 In this section, we prove that $K_F$ is monotonic
 under certain LOCC operations. We note first that it is easy to observe that the key of formation does not increase after local measurements on pure states.

\begin{observation}
\label{pbs:KF-LOCC-mono}
$K_F$ does not increase on pure states under LOCC. 
\end{observation}

\begin{proof}
We note that, for a pure state, $K_F=K_D=E_F$ by definition. Since $E_F$ is an LOCC monotone, we have 
\begin{equation}
    E_F(\psi_{AB})  \geq \sum_k p_k E_F(\sigma_{AB}^k),
\end{equation}
where $\psi_{AB}$ is a pure state, $k$ is an outcome of the LOCC operation, $p_k$ is the probability of outcome $k$, and $\sigma_{AB}^k$ is a state given that outcome $k$ was obtained. Combining the above statements, we obtain
\begin{align}
    K_F(\psi_{AB}) & = E_F(\psi_{AB}) \label{eq:equal_for_pure}\\
    & \geq 
    \sum_k p_k E_F(\sigma_{AB}^k) \\
    & \geq
    \sum_k p_k K_F(\sigma_{AB}^k).
\end{align}
Here the equality comes from the fact that $|\psi\>$ is pure; hence it cannot be expressed as a mixture of a non-pure members ensemble such as non-pure generalized private states. The first inequality follows from the monotonicity of~$E_F$ under LOCC. The second inequality follows from the definition, i.e., $K_F\leq E_F$ for all states. 
\end{proof}

We note, however, that in general, $E_F$ can be much larger than $K_F$ for a private
state. To give an example, consider the state provided in~\cite{smallkey} and recalled explicitly below

\begin{equation}
\gamma_{2,d_s} = 
\frac{1}{2}\left[\begin{array}{cccc}
\sigma & 0 & 0 & \frac{1}{d_s}U^T \\
0 & 0 & 0 & 0 \\
0 & 0 & 0 & 0 \\
\frac{1}{d_s}U^* & 0 & 0 & \sigma 
\end{array}\right],
\label{eq:X_form}
\end{equation}
where
\begin{align}
    U & =\sum_{i,j=0}^{d_s-1}u_{ij}|ii\>\!\<jj|,\\
    \sigma & = \frac{1}{d_s}\sum_{i=0}^{d_s-1} |ii\>\!\<ii|,
\end{align}
with the matrix $\sum_{ij} u_{ij} |i\>\!\<j|$ a unitary transformation. 
In~\cite{locking}, it was shown that the above state satisfies
\begin{equation}
    E_C(\gamma_{2,d_s}) \geq \frac{\log_2 d_s}{2},
    \label{eq:E_Cgreater}
\end{equation}
while (as this state is a strictly irreducible private state), from Corollary~\ref{prop:pbitbounds}, its asymptotic key cost equals $\log_2 d_k =1$. Hence, whenever $d_s > 4$, the
entanglement cost exceeds the key cost of this state.
Now, since $E_F\geq E_F^{\infty}=E_C$, we have
that for the above state $E_F(\gamma_{2,d_s})>K_F(\gamma_{2,d_s})$ for $d_s \geq 4$.
Due to this strict inequality holding, we cannot prove monotonicity of $K_F$ for mixed states by the same approach given in the proof of Observation~\ref{pbs:KF-LOCC-mono}.

We do not provide a general proof of monotonicity under LOCC operations. However, here we present partial results. In~\cite{Horodecki2009}, it is stated that any convex-roof type measure, to which $K_F$ belongs,
is monotonic under LOCC if and only if it is monotonic under local measurements on average. However for the sake of completeness, we show here explicitly the proofs of monotonicity of $K_F$ under all other elementary operations, including adding a local auxiliary state, local partial trace, and random unitary channels.
\begin{observation}
    If $K_F$ is non-increasing under the action of some operation $\Lambda\in \operatorname{LOCC}$ on strictly irreducible generalized private states, then it is non-increasing under the action of $\Lambda$.
\end{observation}
    
\begin{proof}
Let $\rho$ be an arbitrary bipartite state. Since $K_F$ is equal to the infimum over ensembles of $\rho$ into strictly irreducible generalized private states, then for all $\eta > 0 $,  there exists an ensemble $\{(p_i,\gamma_i)\}_i$ such that
\begin{equation}
    \eta + K_F(\rho) \geq \sum_i p_i K_F(\gamma_i).
\end{equation}
Let $\Lambda\in \operatorname{LOCC}$. Suppose that $K_F(\gamma) \geq K_F(\Lambda(\gamma))$ for each strictly irreducible generalized private state $\gamma$. Then the following chain of inequalities holds
    \begin{align}
        \eta + K_F(\rho) & \geq \sum_i p_i K_F(\gamma_i) \\
        & \geq \sum_i p_i K_F(\Lambda(\gamma_i)) \\
        & \geq 
         K_F\!\left(\sum_i p_i \Lambda(\gamma_i)\right) \\
         & = K_F\!\left(\Lambda\!\left(\sum_i p_i \gamma_i\right)\right) \\
         & = K_F(\Lambda(\rho)).
    \end{align}
    For the second inequality, we applied the assumption that $K_F$ does not increase under the action of LOCC on generalized private states.
    Since  the inequality $\eta + K_F(\rho) \geq K_F(\Lambda(\rho))$ holds for all $\eta > 0$, then we conclude that $K_F(\rho) \geq K_F(\Lambda(\rho))$.
\end{proof}

\begin{observation}\label{obs:invariance-unitary}
    $K_F$ is invariant under local unitary transformations.
\end{observation}

\begin{proof}
    Let us fix $\kappa>0$, and let $\{(p_i,\gamma(\psi_i))\}_i$ be an arbitrary ensemble decomposition of $\rho_{AB}$ in terms of generalized private states, such that 
    \begin{equation}
    \sum_i p_i S_{A_K}(\gamma(\psi_i)) = K_F(\rho)+\kappa.
    \label{eq:nearlyKF}
    \end{equation}
    Let us note that Definition~\ref{def:gen_sir_states} for strict irreducibility of each $\gamma(\psi_i)$ implies that for each $i$ there exists a unitary transformation of the form
    \begin{align}
    W_{AB}^{(i)}\coloneqq W^{(i)}_{A\rightarrow A_KA_S}\otimes W^{(i)}_{B\rightarrow B_KB_S},  
    \end{align}
    such that 
    \begin{multline}
    W_{AB}^{(i)} \gamma(\psi_i)_{AB} {W_{AB}^{(i)}}^{\dagger} =\\
    \sum_{l,m} \sqrt{\mu_l^{(i)}\mu_m^{(i)}}|e_l^{(i)} f_l^{(i)}\>\!\<e_m^{(i)}f_m^{(i)}|_{A_KB_K}\\
    \otimes U_l^{(i)}\rho_{A_SB_S}^{(i)} {U_m^{(i)}}^{\dag},
    \label{eq:local-unitary-strict-irr}
    \end{multline}
    where $|\psi_i\> = \sum_l \mu_l^{(i)}|e_l^{(i)}\>|f_l^{(i)}\>$, the operator $U_l^{(i)}$ is a unitary transformation for each $l$ and $i$, and $\rho_{A_SB_S}^{(i)}$ is a quantum state for each $i$.
    Moreover
    \begin{equation}
    S_{A_k}(W_{AB}^{(i)} \gamma(\psi_i)
    {W_{AB}^{(i)}}^{\dagger}) =K_D(\gamma(\psi_i)).
    \label{eq:equalsKey}
    \end{equation}
    
    We will now show that the 
    ensemble
    \begin{equation}
    \left\{(p_i,  (U_A\otimes U_B)\gamma(\psi_i) (U^{\dagger}_A\otimes U_B^{\dagger}))\right\}_i    
    \end{equation}
    is a valid
    decomposition of $\rho'\coloneqq  (U_A\otimes U_B)\rho_{AB} (U_A^{\dagger}\otimes U_B^{\dagger})$
    into strictly irreducible generalized private states,
    achieving the rate in~\eqref{eq:nearlyKF}.
    This will show that $K_F(\rho)\geq K_F(\rho')$.
    To show this, let us fix $i$ arbitrarily and note that, by linearity,
    $\left\{(p_i,  (U_A\otimes U_B)\gamma(\psi_i) (U^{\dagger}_A\otimes U_B^{\dagger}))\right\}_i$ is 
    an ensemble decomposition of $\rho'$. We define
    $V_i \coloneqq  W_{AB}^{(i)} \circ (U_A^{\dagger}\otimes U_B^{\dagger})$ and observe that
    \begin{multline}
    V_i \rho' V_i^{\dagger} = 
    \sum_{l,m} \sqrt{\mu^{(i)}_l\mu^{(i)}_m}|e_l^{(i)} f_l^{(i)}\>\!\<e_m^{(i)}f_m^{(i)}|_{A_KB_K}\\
    \otimes U^{(i)}_l\rho^{(i)}_{A_SB_S} (U^{(i)}_m)^{\dag}.
    \end{multline}
    Moreover
    \begin{multline}
    S_{A_K}(V_i \rho' V_i^{\dagger}) 
    =K_D(\gamma(\psi_i))\\
    =K_D((U_{A}\otimes U_{B})\gamma(\psi_i) (U_{A}^\dagger\otimes U_{B}^\dagger))=K_D(\rho').
    \end{multline}
    The first equality follows from~\eqref{eq:equalsKey} and the second from the fact that distillable key is invariant under local unitary transformations. 
    Hence for every $\kappa$ and ensemble decomposition of $\rho$
    achieving $K_F(\rho)+\kappa$, we constructed
    an ensemble decomposition  of $\rho'$ achieving the same value.
    Since $\kappa>0$ is arbitrary, we conclude that
    $K_F(\rho)\geq K_F(\rho')$. Finally, since $\rho=(U_A^{\dagger}\otimes U_B^{\dagger})(\rho')(U_A\otimes U_B)$, repeating the above argument with the roles of $\rho$ and $\rho'$ exchanged, we obtain $K_F(\rho')\geq K_F(\rho)$, which implies the desired equality.
    \color{black}
\end{proof}

\begin{observation}
    $K_F$ does not increase under addition of local ancilla, i.e. $K_F(\rho\otimes \rho') \leq K_F(\rho))$, where $\rho'$ is a local state of either Alice or Bob.
\end{observation}

\begin{proof}
    Let $\{(p_i, \gamma(\psi_i))\}_i$ be an arbitrary ensemble decomposition of $\rho$ into strictly irreducible generalized private states. Then the ensemble 
    \begin{equation}
       \{(p_i, \gamma(\psi_i)\otimes \rho')\}_i
    \end{equation}
    is a valid ensemble decomposition of $\rho\otimes \rho'$ (not necessarily optimal). Indeed, local ancilla can be treated as a part of a shield system where each unitary operator that acts on the shield is extended by identity, written explicitly as follows:
    \begin{multline}
    \widetilde{\gamma}(\psi_i) = \gamma(\psi_i)\otimes \rho' =\\
    \sum_{l,m}\sqrt{\lambda^{(i)}_l\lambda^{(i)}_m}|e_l^{(i)}f_l^{(i)}\>\!\<e_m^{(i)}f_m^{(i)}|\otimes\\ \left(U^{(i)}_l\otimes \id\right) \left(\sigma^{(i)}\otimes \rho'\right) \left((U^{(i)}_m)^{\dagger} \otimes \id\right).
    \end{multline}
    From the above formula, it is clear that 
    \begin{align}
    K_F(\rho) & = \sum_i p_i K_F(\gamma(\psi_i)) \\
    & = \sum_i p_i K_F(\widetilde{\gamma}(\psi_i)) \\
    & \geq K_F(\rho\otimes \rho').            
    \end{align}
    We have used the fact that extending the shield of the strictly irreducible generalized private state $\gamma(\psi_i)$  by a local state results in a strictly irreducible generalized private state $\widetilde{\gamma}(\psi_i)$.
\end{proof}

\begin{observation}
    $K_F$ is non-increasing under the action of a local random unitary channel, i.e.
    \begin{equation*}
        K_F\!\left(\sum_x p_x (U_x\otimes \id) \gamma_{a\widetilde{A}B}(U_x\otimes \id)^{\dag} \right) \leq K_F(\gamma_{a\widetilde{A}B}).
    \end{equation*}
    where the channel $\sum_x U_x(.)U_x^\dagger$ acts on subsystem $a$.
\end{observation}
\begin{proof}
    \begin{align}
        &K_F\!\left(\sum_x p_x\left(U_x \otimes \id \right) \gamma_{a\widetilde{A}B} \left(U_x \otimes \id \right)^\dag\right) \notag  \\
        & \leq\sum_x p_x K_F( \left(U_x \otimes \id \right) \gamma_{a\widetilde{A}B} \left(U_x \otimes \id \right)^\dag) \\
        &\leq  \sum_x p_x K_F(\gamma_{a\widetilde{A}B}) = K_F(\gamma_{a\widetilde{A}B}).
    \end{align}
    The first inequality holds by Lemma \ref{lem:kf_is_conv}, i.e., convexity of $K_F$, and the second by Observation~\ref{obs:invariance-unitary}.
\end{proof}
As a consequence of the above observation we notice that $K_F$ is also non-increasing under the action of local von Neumann measurements and local partial traces. Indeed, the action of a local von Neumann measurement can written as
\begin{equation*}
    \gamma_{a\widetilde{A}B} \mapsto \frac{1}{d_a} \sum_\ell\left(V_\ell \otimes \id \right) \gamma_{a\widetilde{A}B} \left(V_\ell \otimes \id \right)^\dag,
\end{equation*}
where $V_\ell = \sum_j \exp{2\pi i \ell j /d}|j\>\!\<j|$ and local partial trace can be performed as a composition of three operations for which $K_F$ is non-increasing: local von Neumann measurement, quantum Fourier transform, and local von Neumann measurement.

We note here that, from the above, it follows that~$K_F$ is monotonic under arbitrary local quantum operations. Indeed, we showed that $K_F$ is monotonic under basic operations from which any quantum operation can be composed (adding auxiliary system, unitary and tracing out). What remains a challenge is to show that $K_F$ does not increase under classical communication. This fact would imply monotonicity under LOCC in the sense given in the preliminary section, i.e., $K_F(\rho)\geq K_F(\Lambda(\rho))$ for $\Lambda \in \mathrm{LOCC}$. What is less clear is if $K_F$ is collectively monotonic~\cite{Horodecki2009,Eisert2022}, that is, if $K_F(\rho)\geq \sum_i p_i K_F(\rho_i)$ where $\rho_i$ is a post-measurement state (after certain LOCC) and $p_i$ is the corresponding probability of its occurrence.
 Let us stress here that the approach used in entanglement theory does not apply. In particular, it is not clear to us how to follow the general approach of~\cite{Vidal_2000} for establishing $K_F$ to be an entanglement measure  because the amount of key in the private state does not seem to be expressible as a direct entropic function of the reduced state of the $AB$ systems (as opposed to subsystems $A_KB_{K}$, which can be different for each private state). We also note that
$K_F$ is not a so-called ``mixed convex roof,'' as introduced in~\cite{Yang_2009}. Indeed, in the latter reference, the mixed convex roof is defined as the infimum over ensembles of {\it all} mixed states that apply to a given state. Instead, $K_F$ is a function of splitting only into special (in general mixed) states, the generalized private states.

Another important property that we would wish from $K_F$ is asymptotic continuity. However, all known methods from entanglement theory do not immediately apply to show it. To give one example, the method of~\cite{W16} for proving asymptotic continuity cannot be directly used. This is because a measurement on a purifying system of $\rho$ that splits it into an ensemble of generalized private states does not necessarily generate an ensemble of generalized private states when applied to the purifying system of a state $\rho'$ that is close to $\rho$ in trace distance. Indeed, this is contrary to the fact used in the proof of continuity of $E_F$~\cite{W16}, according to which any rank-one POVM applied to the purifying system results in an ensemble with pure members, no matter to which state ($\rho$ or $\rho'$) it is applied. 

\section{Privacy dilution}\setcurrentname{Privacy dilution}
\label{sec:dilution}

This section is devoted to the description of how the privacy of a private state $\gamma(\Phi^+)$ can be diluted into an $\varepsilon$-approximation of $n$ copies $\gamma(\psi)^{\otimes n}$ of a generalized private state, with $|\psi\>=\sum_{a\in {\cal A}} \lambda_a |e_a\>|f_a\>$. The protocol is such that the error $\varepsilon$ in dilution approaches zero with increasing $n$.

Before providing details of the protocol, we describe two essential subroutines:  Permutation Algorithm (PA) and
Phase Error Correction (PEC). These procedures take care of the phase errors that occur during the execution of the first six steps of the Dilution Protocol. These two
procedures use some auxiliary systems. We need these systems
to return at the end of the protocol to their original state. Otherwise, they may leak some information about the key. We therefore also describe the inverse of the PA protocol (IPA), which takes care of this issue.

Since this section is rather technical, below we provide a table of notations used, with the goal of helping the reader in understanding various proofs and constructions.
\begin{table}[h!]
\centering
\begin{tabular}{ | p{1.5cm} | p{4.5cm} | }
 \hline
 ${\cal A}^n$ & a set of all $n$ element sequences of symbols from alphabet ${\cal A}$\\
 \hline
 ${\bf s}, \widehat{\bf s}$  & a sequence of symbols \\ 
 \hline
 ${\bf s}[i]$ & $i$-th element of a sequence~${\bf s}$ \\  
 \hline
 $\perp_i$ & a symbol $\perp$ at $i$-th position \\
 \hline
 $|{\bf s}|$ & length of a sequence ${\bf s}$ \\
 \hline
$|a({\bf s})|$ & number of occurrences of symbol $a$ in  sequence ${\bf s}$\\
 \hline
 $X[1\ldots n]$ & a quantum system consisting of $n$ subsystems \\
 \hline
$X[i]$ & $i$-th subsystem of a composite quantum system \\
 \hline
$X[i:k]$ & subsystems of system $X$ of indices from $i$ to $k$ \\
 \hline
$\left|\perp\right\rangle_i$ & a quantum state $\left|\perp\right\rangle$ at $i$-th subsystem of a composite quantum system\\
 \hline
 $X_{{\bf s}_k}$ & $U_{{\bf s}_k}\rho^{\otimes |{\bf s}_k|} U_{{\bf s}_k}^{\dag}$ \\
 \hline
\end{tabular}
\caption{Notations used in this section.}
\end{table}

\subsection{Permutation Algorithm}\setcurrentname{Permutation Algorithm}
\label{sec:alg}

The Permutation Algorithm makes use of two permutations $\pi_{\mathrm{key}}({\bf s},\hat{\bf s})$ and $\pi_{\text{shield}}({\bf s},\hat{\bf s})$ that allow Alice and Bob to rearrange their subsystems into a desired order. 
The idea behind these permutations is as follows. By constructing a private state, given that the key takes a value~$\bf s$, the shield should be rotated by the twisting unitary $U_{\bf s}$. More precisely, the shield system should be in a state $U_{\bf s}\rho^{\otimes n}U_{\bf s}^{\dag}$, where
\begin{align}
    {\bf s} & = ({\bf s}[1],\ldots ,{\bf s}[n]) \in {\cal A}^n, \\
    U_{\bf s} & = U_{{\bf s}[1]}\otimes U_{{\bf s}[2]}\otimes \cdots  \otimes U_{{\bf s}[n]}. 
\end{align}
However, when a phase error occurs (as will take place in our dilution protocol described in the next section), then the unitary $U_{\widehat{\bf s}}$ occurs instead.
If $\bf s$ and $\widehat{\bf s}$ were of the same type (with the same number of occurrences of symbols in it), then it would be sufficient to
permute the systems of the shield, i.e., transform $\widehat{\bf s}$ into $\bf s$ by a permutation, thereby correcting the error. However, $\widehat{\bf s}$ may be of a different type. In what follows, both sequences under consideration are from the same $\delta$-typical set $T_n^\delta$ \cite[Section~2.4]{EK11}; that is, each symbol from the alphabet ${\cal A}$ has an empirical frequency that satisfies
\begin{equation}
\abs{\lambda_a -\frac{|a({\bf s})|}{n}}\leq \delta\lambda_a, 
\end{equation}
where $|a({\bf s})|$ is the number of occurrences of $a$ in string ${\bf s}$ and $\lambda_a$ is the probability that symbol $a$ occurs. 
Then for all symbols in ${\cal A}$ appearing in ${\bf s}$ such that their number of occurrences in $\widehat{{\bf s}}$ is larger, there is no problem; we can assign some entries of $\widehat{{\bf s}}$ to correspond with the same values as in $\bf s$ in a proper place. However, for all symbols in $\bf s$ that occur a larger number of times in ${\bf s}$ than in $\widehat{\bf s}$, there is a problem that cannot be resolved by a local permutation. Later on, we will see that $\widehat{\bf s} =\beta_x({\bf s})$ where  $\beta_x: T_n^\delta \to T_n^\delta$ (defined in Lemma~\ref{lem:beta-map}) is a bijection.

The role of the permutation $\pi_{\text{shield}}({\bf s},\hat{\bf s})$ is to correct most of the errors by using the fact that $\bf s$ and $\widehat{\bf s}$ are close to being related to each other by a permutation. The permutation $\pi_{\text{shield}}({\bf s},\hat{\bf s})$ also swaps aside to systems $T_{A_S}T_{B_S}$ all the shielding systems for which the phase error cannot be corrected by a local permutation. The number of such errors is bounded from above by $L_{\max}\coloneqq  \lceil 2\delta n\rceil$. These systems will be treated later by teleportation, local correction on one side, and teleportation back with a small amount of pure form of key, i.e., the singlet state.

In spite of the fact that the permutation $\pi_{\mathrm{key}}({\bf s},\hat{\bf s})$ is constructed in the same way as $\pi_{\mathrm{shield}}({\bf s},\hat{\bf s})$, its role is different. This operation, from a copy of the system $\widehat{\bf s}$ and the auxiliary system $S_{\operatorname{out}}$ in the state $\left|\perp\right\>^{\otimes (n+L_{\max})}$, creates a system $S_{\operatorname{out}}$ composed of two subsystems in the states ${\bf s}_{\operatorname{cor}}$ and ${\bf s}_{\operatorname{err}}$, respectively (and leaves $\widehat{\bf s}$ in the state $\left|\perp\right\>^{\otimes n}$). The states ${\bf s}_{\operatorname{cor}}$ and ${\bf s}_{\operatorname{err}}$ indicate how to correct the error
on systems $B_ST_{A_S}T_{B_S}$ when they are in Bob's possession after teleportation of $T_{A_S}$ by Alice. The state~${\bf s}_{\operatorname{err}}$ indicates which unitary transformations should be inverted as they are incorrect. The state ${\bf s}_{\operatorname{cor}}$ indicates that if the symbol $\perp$ occurs at position $p_k$, then the unitary~$U_a$ should be applied if the symbol $a$ is in position~$p_k$ in ${\bf s}$. When applied to auxiliary systems, the permutation $\pi_{\mathrm{key}}$ creates a state 
that is further a control system for the Phase Error Correction algorithm.
The register holding ${\bf s}_{\operatorname{err}}$ indicates how to undo the errors, while the
registers holding ${\bf s}$ and ${\bf s}_{\operatorname{cor}}$ possess information about the correct unitary transformation that needs to be applied on system $T_{A_S}T_{B_S}$.
Since the action of the above mentioned permutations is described on a pair $({\bf s},\widehat{\bf s})$ of sequences, then they are applied in a control-target manner (with control on the value of ${\bf s}$ and $\widehat{\bf s}$ ).

At the end of this section, we describe the Permutation Algorithm in more detail. For now, we first claim the existence of the desired permutations, which we argue by proving a number of combinatorial lemmas.

\begin{lemma}\label{lem:permutation-typical-seq}
    Suppose that ${\bf x}, \widehat{\bf x}\in T_n^\delta$. Then there exists an increasing sequence of indices $\{i_k\}_k$ selected from the set $\{1,\ldots, n\}$ such that ${\bf x}$ can be written as
    \begin{equation}
        {\bf x} = {\bf x}_1 {\bf x}[i_1] {\bf x}_2 {\bf x}[i_2] \cdots {\bf x}_{f({\bf x}, \widehat{\bf x})} {\bf x}[i_{f({\bf x}, \widehat{\bf x})}] {\bf x}_{f({\bf x}, \widehat{\bf x})+1},
    \end{equation}
    where ${\bf x}_k$ are blocks (not necessarily nonempty) of contiguous elements from the string ${\bf x}$, 
    \begin{equation}
        f({\bf x}, \widehat{\bf x}) \coloneqq  \sum_{a\in {\cal A}} \big||a({\bf x})| - |a(\widehat{\bf x})|\big| \leq \lceil 2\delta n\rceil ,
    \end{equation}
    and there exists a permutation $\pi$ such that
    \begin{equation}
        \pi(\widehat{\bf x}) = {\bf x}_1 \widehat{\bf x}[j_1] {\bf x}_2 \widehat{\bf x}[j_2] \cdots {\bf x}_{f({\bf x}, \widehat{\bf x})} \widehat{\bf x}[j_{f({\bf x}, \widehat{\bf x})}] {\bf x}_{f({\bf x}, \widehat{\bf x})+1},
    \end{equation}
    where $\{j_k\}_k$ is an increasing sequence of indices selected from the set $\{1,\ldots, n\}$.
\end{lemma}

\begin{proof}
    From strong typicality, for all $a \in {\cal A}$, we have that
    \begin{equation}
    \big||a({\bf x})| - |a(\widehat{\bf x})|\big|\leq 2\delta\lambda_a n.
    \end{equation}
    Therefore $0 \leq f({\bf x}, \widehat{\bf x}) \leq L_{\mathrm{max}}$, where
    \begin{equation}
        L_{\mathrm{max}} \coloneqq  \lceil 2\delta n\rceil.
    \end{equation}
    Since $f({\bf x}, \widehat{\bf x})$ counts the difference between the occurrences of each symbol in strings ${\bf x}$ and $\widehat{\bf x}$, then ${\bf x}$ and $\widehat{\bf x}$ have at least $n-f({\bf x}, \widehat{\bf x})$ in common. So we can group $n-f({\bf x}, \widehat{\bf x})$ elements from $\widehat{\bf x}$ into at most $f({\bf x}, \widehat{\bf x}) + 1$ blocks (not necessarily non-empty) $\widehat{\bf x}_1, \widehat{\bf x}_2, \ldots, \widehat{\bf x}_{f({\bf x}, \widehat{\bf x})}, \widehat{\bf x}_{f({\bf x}, \widehat{\bf x})+1}$ which are equal to blocks ${\bf x}_1, {\bf x}_2 , \ldots, {\bf x}_{f({\bf x}, \widehat{\bf x})}, {\bf x}_{f({\bf x}, \widehat{\bf x})+1}$ of contiguous elements from ${\bf x}$ such that
    \begin{equation}
        {\bf x} = {\bf x}_1 {\bf x}[i_1] {\bf x}_2 {\bf x}[i_2] \ldots {\bf x}_{f({\bf x}, \widehat{\bf x})} {\bf x}[i_{f({\bf x}, \widehat{\bf x})}] {\bf x}_{f({\bf x}, \widehat{\bf x})+1},
    \end{equation} 
    where $i_k = i_{k-1} + |{\bf x}_k|+1$ and $i_1 = |{\bf x}_1|+1$. The remaining $f({\bf x}, \widehat{\bf x})$ elements from the string $\widehat{\bf x}$ can be denoted as $\widehat{\bf x}[j_1]\ldots \widehat{\bf x}[j_{f({\bf x}, \widehat{\bf x})}]$, where $j_k$ are indices of elements in the sequence $\widehat{\bf x}$. Observe that the elements $\widehat{\bf x}[j_1]\ldots \widehat{\bf x}[j_{f({\bf x}, \widehat{\bf x})}]$ can be ordered in such a way that $j_k < j_m$ for $k<m$. Now it is clear that there exists a permutation $\pi$ that rearranges elements of $\widehat{\bf x}$ into the desired order.
\end{proof}

\begin{lemma}\label{lem:permutation-pi-key-typical-seq}
    Suppose that ${\bf x}, \widehat{\bf x}\in T_n^\delta$. Then there exists a permutation $\pi$ such that 
    \begin{multline}
        \pi(\widehat{\bf x} \underbrace{\perp \ldots \perp}_{n+L_{\mathrm{max}}}) = \underbrace{\perp \ldots \perp}_{n}\\
        {\bf x}_1 \perp {\bf x}_2 \perp \ldots {\bf x}_{f({\bf x}, \widehat{\bf x})} \perp {\bf x}_{f({\bf x})+1} \\
        \widehat{\bf x}[j_1] \ldots \widehat{\bf x}[j_{f({\bf x}, \widehat{\bf x})}] \underbrace{\perp \ldots \perp}_{L_{\mathrm{max}} - f({\bf x}, \widehat{\bf x})}.
    \end{multline}
\end{lemma}

\begin{proof}
    From Lemma~\ref{lem:permutation-typical-seq}, we know that there exists a permutation $\pi$ that  rearranges the string $\widehat{\bf x} \underbrace{\perp \ldots \perp}_{n+L_{\mathrm{max}}}$ into
    \begin{equation}
        {\bf x}_1 \widehat{\bf x}[j_1] {\bf x}_2 \widehat{\bf x}[j_2] \ldots {\bf x}_{f({\bf x}, \widehat{\bf x})} \widehat{\bf x}[j_{f({\bf x}, \widehat{\bf x})}] {\bf x}_{f({\bf x}, \widehat{\bf x})+1} \underbrace{\perp \ldots \perp}_{n+L_{\mathrm{max}}}.
    \end{equation}
    Now it remains to compose it with $n + f({\bf x}, \widehat{\bf x})$ transpositions to obtain the desired permutation.
\end{proof}

\begin{lemma}\label{lem:quantum-pi-key}
    Suppose that ${\bf s}, \widehat{\bf s}\in T_n^\delta$. Then there exists a unitary operation corresponding to $\pi_{\mathrm{key}}({\bf s}, \widehat{\bf s})$ that maps systems $\widehat{S}$ and $S_{\operatorname{out}}$, which are initially in the state $|\widehat{{\bf s}}\>_{\widehat{S}}\left|\perp \right\rangle_{S_{\operatorname{out}}}^{\otimes (n+L_{\mathrm{max}})}$, to systems in the state $|{\bf s}_{\operatorname{all}}\>$
    \begin{multline}
    |{\bf s}_{\operatorname{all}}\>= \left|\perp\right\rangle_{\widehat{S}}^{\otimes n}\otimes 
    |{\bf s}_1,\perp_{i_1},{\bf s}_2,\perp_{i_2},\hdots,\\
    {\bf s}_{f({\bf s},\widehat{\bf s})},\perp_{i_{f({\bf s},\widehat{\bf s})}},{\bf s}_{f({\bf s},\widehat{\bf s})+1},
    \widehat{\bf s}[j_1], \widehat{\bf s}[j_2],  \ldots, \\
    \widehat{\bf s}[j_{f({\bf s},\widehat{\bf s})}],\perp_{(f({\bf s},\widehat{\bf s})+1)}, \ldots ,\perp_{L_{\max}}\>_{S_{\operatorname{out}}} \\ 
    \equiv |\underbrace{\perp,  \ldots \perp}_{\text{$n$ times}}\>_{\widehat{S}}|{\bf s}_{\operatorname{cor}},{\bf s}_{\operatorname{err}}\>_{S_{\operatorname{out}}}.
    \end{multline}
    Here for abbreviation, by $\perp_{i_k}$ we mean the state $\left|\perp\right\>$ at position $i_k$.
\end{lemma}

\begin{proof}
    The desired operation is a quantum realization (i.e., a permutation of quantum systems) of a permutation from Lemma~\ref{lem:permutation-pi-key-typical-seq}. Since every permutation is a composition of transpositions, it is clear that such a quantum operation can be constructed as a composition of SWAP operations.
\end{proof}

\begin{lemma}\label{lem:quantum-pi-shield}
    Suppose that ${\bf s}, \widehat{\bf s}\in T_n^\delta$, the permutation $\pi_{\mathrm{key}}({\bf s}, \widehat{\bf s})$  corresponds to a unitary operation from Lemma~\ref{lem:quantum-pi-key}, and system $A_S$ is in the state
    \begin{equation*}
        U_{\widehat{\bf s}} \rho^{\otimes n} U_{\widehat{\bf s}}^{\dag}.
    \end{equation*}
    Then there exists a unitary operation $\pi_{\mathrm{shield}}({\bf s}, \widehat{\bf s})$ that maps systems $A_S[1 \ldots n],\widetilde{A}_S[1 \ldots n]$ and $T_{A_S}[1 \ldots L_{\max}]$ where $\widetilde{A}_S[1 \ldots n]$ is initially in state $\left|\perp\right\rangle ^{\otimes n}$ and $T_{A_S}[1 \ldots L_{\max}]$ is initially in state $\left|\perp\right\rangle^{\otimes L_{\mathrm{max}}}$ into systems in the states
    \begin{align}
        &A_S \to \left|\perp \right\rangle^{\otimes n},\\
        &\widetilde{A}_S[1 \ldots n] \to \\
        &U_{{\bf s}_1}\rho^{\otimes |{\bf s}_1|} U_{{\bf s}_1}^{\dag}\otimes \left|\perp\right\rangle\!\left\langle \perp\right |_{i_1} \otimes U_{{\bf s}_2}\rho^{\otimes |{\bf s}_2|} U_{{\bf s}_2}^{\dag}\otimes \left|\perp\right\rangle\!\left\langle \perp\right |_{i_2} \otimes\nonumber\\
        &  \ldots \otimes \left|\perp\right\rangle\!\left\langle \perp\right |_{i_{f({\bf s},\widehat{\bf s})}} \otimes U_{{\bf s}_{f({\bf s},\widehat{\bf s}) + 1}}\rho^{\otimes |{\bf s}_{f({\bf s},\widehat{\bf s}) + 1}|} U_{{\bf s}_{f({\bf s},\widehat{\bf s}) + 1}}^{\dag}, \\
        &T_{A_S} \to A_S[j_1]\ldots A_S[j_{f({\bf s},\widehat{\bf s})}]\underbrace{\perp \ldots \perp}_{L_{\mathrm{max}}- f({\bf s},\widehat{\bf s})},
    \end{align}
    where indices $i_k, j_k$ and blocks ${\bf s}_k$ are the same as in a given quantum operation $\pi_{\mathrm{key}}({\bf s}, \widehat{\bf s})$.
\end{lemma}

\begin{proof}
    This operation is almost the same as the given $\pi_{\mathrm{key}}({\bf s}, \widehat{\bf s})$, but it acts on different subsystems.
\end{proof}

We now aim to define the action of PA for Alice, which uses the  permutations defined above. See Algorithm~\ref{alg:PPA} below. Bob's version is almost the same, but with the interchange of symbols $A\leftrightarrow B$. The procedure first creates auxiliary systems $\widehat{S}$, $\widehat{S}_1$, $T_{A_S}$, and ${\widetilde A}_S[1 \ldots n]$ and then performs permutations $\tau_{\mathrm{key}}^{PA}$ and $\tau_{\text{shield}}^{PA}$, which are based on   $\pi_{\mathrm{key}}({\bf s},\widehat{\bf s})$ and $\pi_{\text{shield}}({\bf s},\widehat{\bf s})$ as  defined in Lemmas~\ref{lem:quantum-pi-key} and~\ref{lem:quantum-pi-shield}, respectively.

\begin{algorithm}[H]
\caption{Permutation Algorithm}
\begin{algorithmic}[1]
\Procedure{PA}{$A_K[1 \ldots n], A_S[1 \ldots n]$}\\
From $|{\bf s}\>_{A_K}$, locally create the state $|{\bf s}\>_{A_K}|\beta_x({\bf s})\>_{\widehat{S}}|\beta_x({\bf s})\>_{\widehat{S}_1}\equiv |{\bf s}\>_{A}\left|\widehat{\bf s}\right\>_{\widehat{S}} \left|\widehat{\bf s}\right\>_{\widehat{S}_1}$. \\
create a register $T_{A_S}$ of size $L_{\max}$ in the state $\left|\perp\right\>^{\otimes L_{\max}}$. \\
create a register ${\widetilde A}_S[1 \ldots n]$ in the state $\left|\perp\right\>^{\otimes n}$.\\
create a register $S_{\operatorname{out}}$ in the state $\left|\perp\right\>^{\otimes (n+L_{\max})}$\\
Alice performs the following unitary on  $A_K{\widehat S}{\widehat{S}_1}S_{\operatorname{out}}$:
$\tau^{PA}_{\mathrm{key}}\coloneqq \sum_{{\bf s},\widehat{\bf s}\in T^{\delta}_n}|{\bf s},\widehat{\bf s}\>\!\<{\bf s},\widehat{\bf s}|_{A_K\widehat{S}}\otimes \pi_{\mathrm{key}}({\bf s},\widehat{\bf s})_{\widehat{S}_1S_{\operatorname{out}}}$ \\
Alice performs the following unitary on  $A_K\widehat{S}A_S\widetilde{A}_ST_{A_S}$:
$\tau_{\text{shield}}^{PA}\coloneqq \sum_{{\bf s},\widehat{\bf s}\in T^{\delta}_n}|{\bf s},\widehat{\bf s}\>\!\<{\bf s},\widehat{\bf s}|_{A_K\widehat{S}}\otimes \pi_{\text{shield}}({\bf s},\widehat{\bf s})_{A_S\widetilde{A}_ST_{A_S}}+$
$\sum_{({\bf s},\widehat{\bf s})\notin T^{\delta}_n\times T^{\delta}_n}|{\bf s},\widehat{\bf s}\>\!\<{\bf s},\widehat{\bf s}|_{A_K\widehat{S}}\otimes \id_{A_S\widetilde{A}_ST_{A_S}}$

\EndProcedure  
\end{algorithmic} \label{alg:PPA}
\end{algorithm}

To summarize this subsection, we provide an example for the reader who is more interested in the idea of the protocol rather than in a strict mathematical proof.
Let $n=9$, $\delta = \frac{7}{9}$, ${\cal A} = \{a,b,c,d\}$ and $\lambda_x = \frac{1}{4}$ for each $x\in {\cal A}$. For this example, we choose the following two strings from ${\cal A}^n$
\begin{align}
    {\bf s} & = (b,c,c,b,d,b,a,a,c) , \\
    \widehat{\bf s} & = (c,b,b,c,c,d,a,d,c).
\end{align}
It is easy to check that these two strings belong to the $\delta$-typical set $T_n^\delta$. Given $n$ and $\delta$ as chosen above, we find that $L_{\mathrm{max}} = 14$. The states of the systems $\widehat{S}_1$, $S_{\mathrm{out}}$, $A_S$, $ \widetilde{A}_S$, and  $T_{A_S}$ before and after the application of $\tau_{\mathrm{key}}^{PA}$ and $\tau_{\mathrm{shield}}^{PA}$ are depicted in Figures~\ref{fig:before-PA} and \ref{fig:after-PA}, respectively.
\begin{figure}[h]
	\centering
	\includegraphics[trim={2cm 2 3 2},width=7cm]{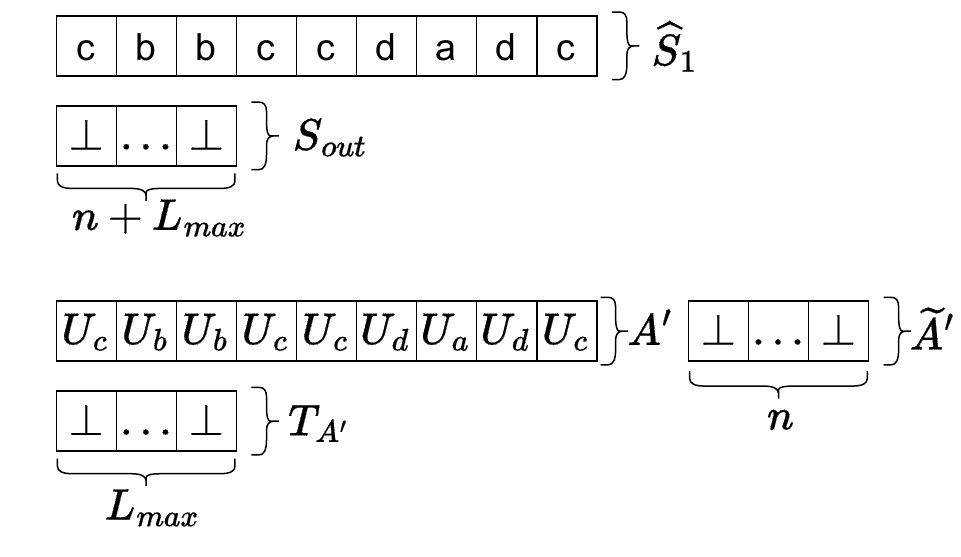}
	\caption{Systems $\widehat{S}_1, S_{\mathrm{out}}, A_S, \widetilde{A}_S, T_{A_S}$ before the application of~$\tau_{\mathrm{key}}^{PA}$ and~$\tau_{\mathrm{shield}}^{PA}$.}
	\label{fig:before-PA}
\end{figure}

\begin{figure}[h]
	\centering
	\includegraphics[trim={2cm 2 3 2},width=7cm]{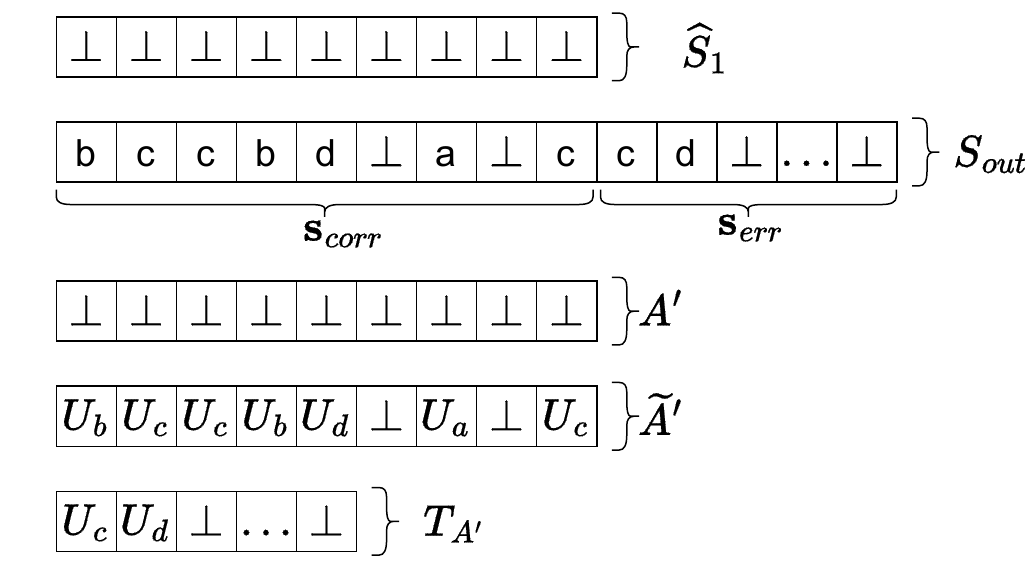}
	\caption{Systems $\widehat{S}_1, S_{\mathrm{out}}, A_S, \widetilde{A}_S, T_{A_S}$ after the application of $\tau_{\mathrm{key}}^{PA}$ and $\tau_{\mathrm{shield}}^{PA}$.}
	\label{fig:after-PA}
\end{figure}

\subsection{Phase error correction algorithm}\setcurrentname{Phase error correction algorithm}
\label{sec:PEC}

We now explain how the output of the PA (applied both by Alice and by Bob), taken as the input 
of the Phase Error Correction (PEC) algorithm, is used to correct coherently the errors that remain after the PA procedure. Most of the errors have been corrected by PA.
However, the errors that came out
due to a type mismatch of ${\bf s}$ and 
$\widehat{{\bf s}}\equiv \beta_x({\bf s})$,
collected in the system $T_{A_S}T_{B_S}$ need to be handled separately. By notation, $|x\>_j$ is equal to the $j$-th subsystem of the state $|x\>$. Further, $S_{\operatorname{out}}^{B}$ denotes system $S_{\operatorname{out}}$ obtained by applying $\tau_{\mathrm{key}}^{PA}$ by Bob during the execution of PA given in Algorithm~\ref{alg:PPA}. The 
aforementioned errors can be corrected in two steps:
\begin{enumerate}
    \item untwisting incorrect unitaries,
    \item applying correct unitaries.
\end{enumerate}

The first step is easy to perform because it suffices for Bob to apply the following unitary operation
\begin{multline}
    \tau_{1}^{\mathrm{PEC}}\coloneqq \bigotimes_{i=1}^{L_{\mathrm{max}}}\Big[\sum_{a\neq \perp}|a\>\!\<a|_{{\bf s}_{\operatorname{err}}[i]}\otimes U^{\dagger}_{a,T_{A_S}[i]T_{B_S}[i]}+\\
    \left|\perp\right\rangle\!\left\langle \perp\right |_{{\bf s}_{\operatorname{err}}[i]}\otimes \id_{a,T_{A_S}[i]T_{B_S}[i]}\Big].
\end{multline}
Performing the second step requires auxiliary system $C$ initially in state $|1\>_C$ that will be used as a counter. To make use of this counter, we define the following unitary operation
\begin{equation}
    U_{XC}^{\mathrm{shift}} \coloneqq  \sum_{a\in {\cal A}}|a\>\!\<a|_X\otimes \id_C + \left|\perp\right\rangle\!\left\langle \perp\right |_X\otimes U_C^{\oplus 1},
    \label{eq:shift}
\end{equation}
where $U^{\oplus 1}$ is a unitary operation that transforms $|t\>$ into $|(t\oplus 1) \mod L_{\mathrm{max}}\>$. To perform the second step, Bob has to do the following procedure. For each $i \in \{ 1, \ldots, n\}$, Bob applies 
\begin{multline}
     \tau_{2}^{\mathrm{PEC}}(i) \coloneqq  \\
    \sum_{b\in {\cal A}}\sum_{c=1}^{L_{\mathrm{max}}}\sum_{a\in {\cal A}\cup\{\perp\}}|a\>\!\<a|_{{\bf s}[i]}\otimes |b\>\!\<b|_{{\bf s}_{\mathrm{cor}}[i]}\otimes |c\>\!\<c|_C\otimes \id_{T_{A_S}T_{B_S}} \\
     + |a\>\!\<a|_{{\bf s}[i]}\otimes \left|\perp\>\!\<\perp\right|_{{\bf s}_{\mathrm{cor}}[i]}\otimes |c\>\!\<c|_C\otimes U_{T_{A_S}T_{B_S}}^{a,c},
\end{multline}
where 
\begin{equation}
    U_{T_{A_S}T_{B_S}}^{a,c} \coloneqq  U_{a,T_{A_S}[c]T_{B_S}[c]}\otimes \id_{T_{A_S}[\neq c]T_{B_S}[\neq c]}
\end{equation}
and $U_{a}$ is a unitary operation from the definition of private states and corresponding to the symbol $a$ (see~\eqref{eq:strict_ir}). Next he applies $U_{{\bf s}_{\mathrm{cor}}[i]C}^{\mathrm{shift}}$, as defined in~\eqref{eq:shift}.
We conclude this section by presenting the PEC algorithm in a compact format:

\begin{algorithm}[H]
\caption{Phase Error Correction Algorithm}
\begin{algorithmic}[1]
\Procedure{PEC}{$B_KS^{B}_{\operatorname{out}}T_{A_S}T_{B_S}$}\\
Alice teleports the system $T_{A_S}$ to Bob.\\
Bob creates system $C$ in state $|1\>$.\\
Bob performs
\State{$\tau_{1}^{\mathrm{PEC}}\coloneqq \otimes_{i=1}^{L_{\mathrm{max}}}\big[\sum_{a\neq \perp}|a\>\!\<a|_{{\bf s}_{\operatorname{err}}[i]}\otimes U^{\dagger}_{a,T_{A'_i}T_{B'_i}}+\left|\perp\right\rangle\!\left\langle \perp\right |_{{\bf s}_{\operatorname{err}}[i]}\otimes \id_{a,T_{A'_i}T_{B'_i}}\big]$}
\For{$\mathrm{i}=1$ to $n$}
    \State{$\tau_{2}^{\mathrm{PEC}}(i),
        $}
    \State $U_{{\bf s}_{\mathrm{cor}}[i]C}^{\mathrm{shift}} 
    $
\EndFor
\State \textbf{endfor}\\
Bob teleports the system $T_{A_S}$ back to Alice.\\
Bob resets his counter system $C$ to state $|1\>$ by applying $\left( U_{{\bf s}_{\mathrm{cor}}[i]C}^{\mathrm{shift}}\right)^{\dag}$ $n$ times.
\EndProcedure  
\end{algorithmic}\label{alg:PEC}
\end{algorithm}

To continue the example given in the previous section, we present the action of $\tau_1^{\mathrm{PEC}}$ in Figure~\ref{fig:tau_1_PEC}, $\tau_2^{\mathrm{PEC}}(6)$ in Figure~\ref{fig:tau_2_PEC_6}, and $U_{{\bf s}_{\mathrm{cor}}C}^{\mathrm{shift}}$ in Figure~\ref{fig:U_shift}.

\begin{figure}[h]
	\centering
	\includegraphics[trim={2cm 2 3 2},width=7cm]{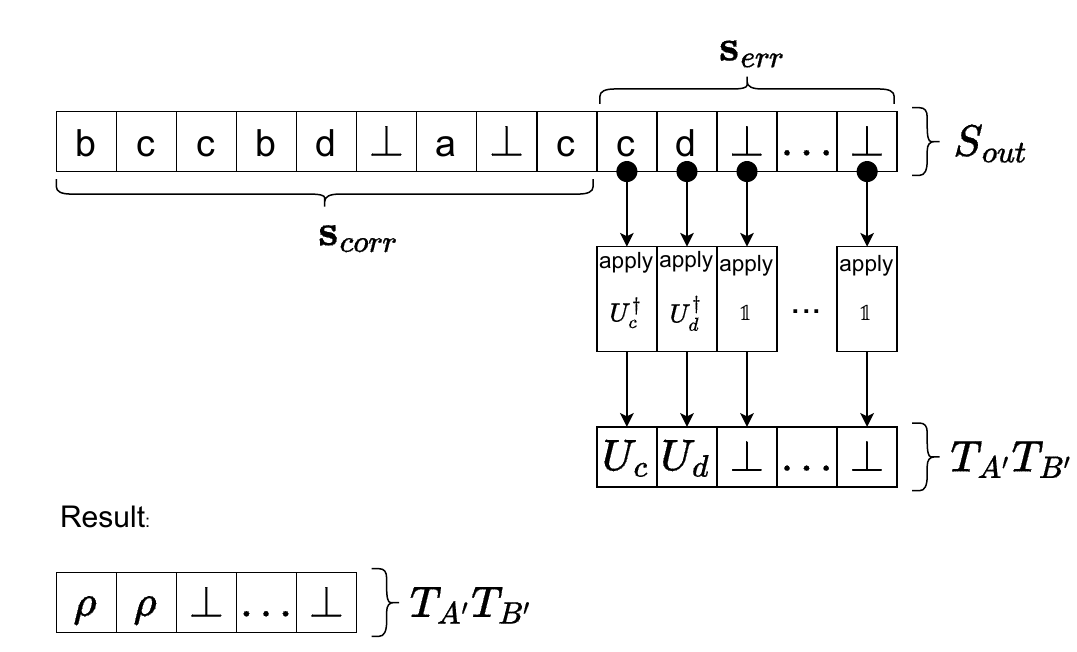}
	\caption{$\tau_1^{\mathrm{PEC}}$ operation. Control systems are marked with a black dot. If there is a symbol $a\in {\cal A}$ on ${\bf s}_{\mathrm{err}}[k]$, then it applies $U^\dag_a$ on the system $T_{A_S}T_{B_S}[k]$. On the other hand, if there is a symbol $\perp$ on ${\bf s}_{\mathrm{err}}[k]$, then it applies $\id$. After the action of $\tau_1^{\mathrm{PEC}}$, states stored in the system $T_{A_S}T_{B_S}$ become untwisted.}
	\label{fig:tau_1_PEC}
\end{figure}

\begin{figure}[h]
	\centering
	\includegraphics[trim={2cm 2 3 2},width=7cm]{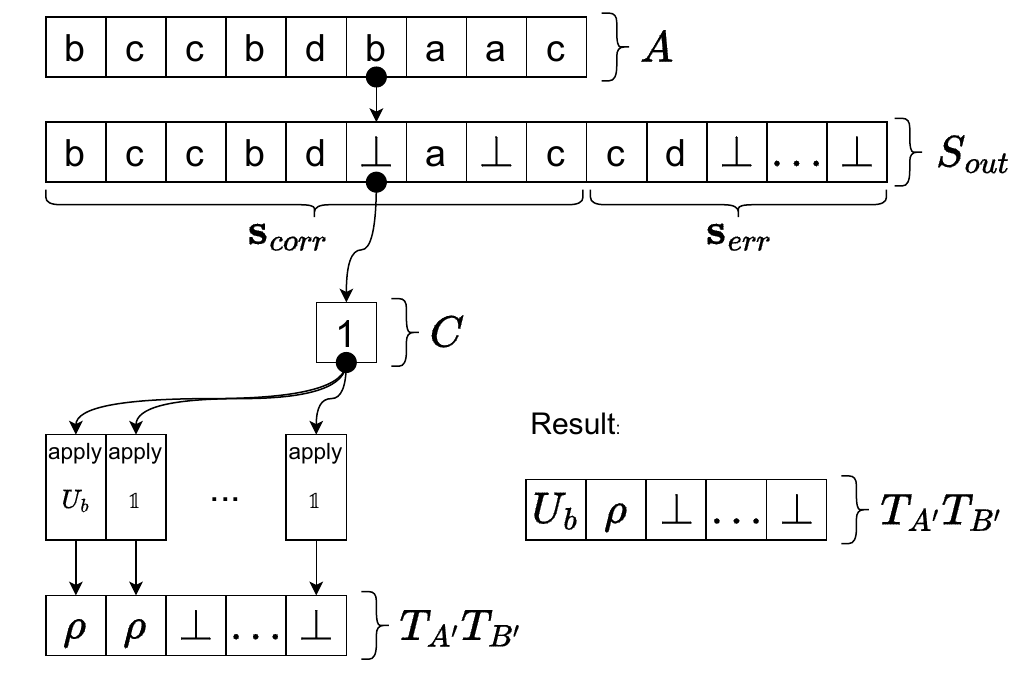}
	\caption{The action of $\tau_2^{\mathrm{PEC}}(6)$ operation. Notice that this is only one step from the loop which iterates from $i=1$ to $i=9$ and for each $i$, the operation $\tau_2^{\mathrm{PEC}}(i)$ is performed. $\tau_2^{\mathrm{PEC}}(6)$ operation has a control on $A_K[6]$, ${\bf s}_{\mathrm{cor}}[6]$, $C$ and target on $T_{A_S}T_{B_S}$. If it `sees' $\perp$ on ${\bf s}_{\mathrm{cor}}[6]$, then it applies $U_b$ on $T_{A_S}T_{B_S}[1]$ where $b$, $1$ are states of system ${\bf s}_{\mathrm{cor}}[6]$ and $C$ respectively. Otherwise, it applies $\id$. As a result of this operation, the state $T_{A_S}T_{B_S}[1]$ is twisted with a correct unitary operation. Notice also that after this operation state of the system~$C$ will be shifted, because ${\bf s}_{\mathrm{cor}}[6]$ is in a $\perp$ state (see Figure~\ref{fig:U_shift} or Step~8 of Algorithm~\ref{alg:PEC}).}
	\label{fig:tau_2_PEC_6}
\end{figure}

\begin{figure}[h]
	\centering
	\includegraphics[trim={2cm 2 3 2},width=6.2cm]{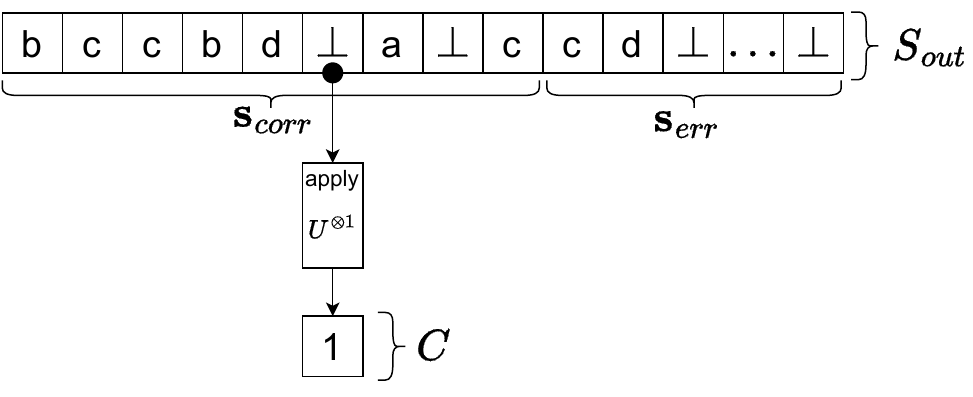}
	\caption{The action of $U_{{\bf s}_{\mathrm{cor}}[6]C}^{\mathrm{shift}}$ operation. It has control on ${\bf s}_{\mathrm{cor}}[6]$ system and target on $C$ system. If there is a $\perp$ symbol on ${\bf s}_{\mathrm{cor}}[6]$, then the state of $C$ is incremented by $1$.}
	\label{fig:U_shift}
\end{figure}

\subsection{Inverting PA Protocol}\setcurrentname{Inverting PA Protocol}
\label{sec:RP}

In the preceding two sections, we have seen that the composition of PA and PEC is, in fact, a composition of twistings, i.e., controlled-unitary operations. The first one has a control on $|{\bf s}\>_{A_K}$ and creates $|\widehat{\bf s}\>$. The next one $\tau_{\mathrm{key}}^{PA}$ has a control on $|{\bf s}\>_{A_K}|\widehat{\bf s}\>_{\widehat{S}}$ and permutes $|\widehat{\bf s}\>_{\widehat{S}_1}|{\bf s}_{\operatorname{out}}\>$. Further $\tau_{\text{shield}}^{PA}$ is applied with a control on $|{\bf s}\>_{A_K}|\widehat{\bf s}\>_{\widehat{S}}$ and with a target on $A_S[1\ldots n]\widetilde{A}_S[1\ldots n]T_{A_S}[1\ldots L_{\max}]$. 
Next, PEC also performs twistings, specified by $\tau_{1}^{\mathrm{PEC}}$ and $\tau_{2}^{\mathrm{PEC}}$ given in lines $5$ and $8$ of the definition of PEC. Note that the 
twisting operations mentioned above involve auxiliary systems initiated in a tensor power of  $|0\>$ or $\left|\perp\right\>$. These auxiliary systems were changed via further operations.

In this section, we show that all the auxiliary systems can be uncomputed to their original state and therefore traced out with no leakage of information about the key of systems $A_KB_K$ of the generalized private state to Eve. It is the final building block of the privacy dilution protocol. Note here that if we trace out the auxiliary systems to Eve without resetting them to the initial state, she will learn about the key part of the state to be produced. The output of this procedure would no longer be the $\widetilde{\gamma}(\rho_{\delta,n})$ state, which is undesirable.

In what follows, we describe the Inverting Permutation Algorithm, which outputs the desired state only on the systems $A_KA_SB_KB_S$. 
However, at first, we will describe a procedure of rearranging subsystems of systems $A_S,\widetilde{A}_S,T_{A_S}, B_S,\widetilde{B}_S,T_{B_S}$ to obtain the correct state on $A_S,B_S$ and $\left|\perp\right\rangle^{\otimes (n+L_{\mathrm{max}})}$ on the remaining auxiliary systems. To do this, we have to define another unitary operation, namely 
\begin{multline}
    \tau_1^{\mathrm{IPA}}(i)\coloneqq 
    \sum_{c=1}^{L_{\mathrm{max}}}\sum_{a\in{\cal A}}|a\>\!\<a|_{{\bf s}_{\mathrm{cor}}[i]}\otimes |c\>\!\<c|_C \otimes \id_{T_{A_S}\widetilde{A}_S} + \\
     \left|\perp\>\!\<\perp\right|_{{\bf s}_{\mathrm{cor}}[i]}\otimes |c\>\!\<c|_C \otimes W_{T_{A_S}\widetilde{A}_S}^{c,i},
\end{multline}
where
\begin{equation}
    W_{T_{A_S}\widetilde{A}_S}^{c,i} \coloneqq  \operatorname{SWAP}[\widetilde{A}_S[i],T_{A_S}[c]]\otimes \id_{\widetilde{A}_S[\neq i]T_{A_S}[\neq c]}.
\end{equation}
 To perform this operation, Alice creates an auxiliary system $C$ that will serve as a counter (Bob has already created it during the PEC algorithm). Now, for each $i\in \{1,\ldots, n\}$, both of them have to perform the same steps on appropriate systems, as follows:
\begin{enumerate}
    \item apply $\tau_1^{\mathrm{IPA}}(i)$,
    \item apply $U_{{\bf s}_{\mathrm{cor}}[i]C}^{\mathrm{shift}}$.
\end{enumerate}
After these steps, both of them have to swap systems~$\widetilde{A}_S$ and $A_S$ to obtain a corrected shield in system $A_S$ and the state~$\left|\perp \right\rangle^{\otimes (n+L_{\mathrm{max}})}$ in systems $\widetilde{A}_S$ and $T_{A_S}$. It only remains to reset the counter by performing $\left( U_{{\bf s}_{\mathrm{cor}}[i]C}^{\mathrm{shift}} \right)^{\dag}$ $n$ times.

The description of IPA is as follows. 
\begin{algorithm}[H]
\caption{Inverting PA}
\begin{algorithmic}[1]
\Procedure{IPA}{$A_KA_S$$\widehat{S}_1$$S_{\operatorname{out}}$$T_{A_S}$,$\beta_x$}\\
Alice creates system $C$ in state $|1\>$\\
Alice does:
\For{$\mathrm{i}=1$ to $n$}
    \State{$\tau_1^{\mathrm{IPA}}(i)$
    ,}
    \State{$U_{{\bf s}_{\mathrm{cor}}[i]C}^{\mathrm{shift}}$}
\EndFor
\State \textbf{endfor}\\
Alice swaps systems $A_S$ and $\widetilde{A}_S$\\
Alice performs the following unitary on  $B_K{\widehat S}{\widehat{S}_1}S_{\operatorname{out}}$:
${\tau^{PA}_{\mathrm{key}}}^{\dagger}\equiv\sum_{{\bf s},\widehat{\bf s}\in T^{\delta}_n}|{\bf s},\widehat{\bf s}\>\!\<{\bf s},\widehat{\bf s}|_{B_K\widehat{S}}\otimes \pi_{\mathrm{key}}^{-1}({\bf s},\widehat{\bf s})_{\widehat{S}_1S_{\operatorname{out}}}$ \\
Alice undoes the creation of $|\widehat{s}\>_{\widehat{S}}|\widehat{s}\>_{\widehat{S}_1}$ done in Step~1 of PA by applying $U_{\beta_x}^{\dagger}$ twice\\
Alice applies CNOTs to reset the state $|{\bf s}\>$ to $| 0\>^{\otimes n}$  in systems $\widehat{S}, \widehat{S}_1$\\
Alice traces out the auxiliary systems $\widehat{S}\widehat{S}_1S_{\operatorname{out}}\widetilde{A}_ST_{A_S}$.
\EndProcedure  
\end{algorithmic}\label{alg:IPPA}
\end{algorithm}
Lastly, we depict the action of $\tau_1^{\mathrm{IPA}}(8)$ in Figure~\ref{fig:tau_1_IPA_8}.
\begin{figure}[h]
	\centering
	\includegraphics[trim={2cm 2 3 2},width=7cm]{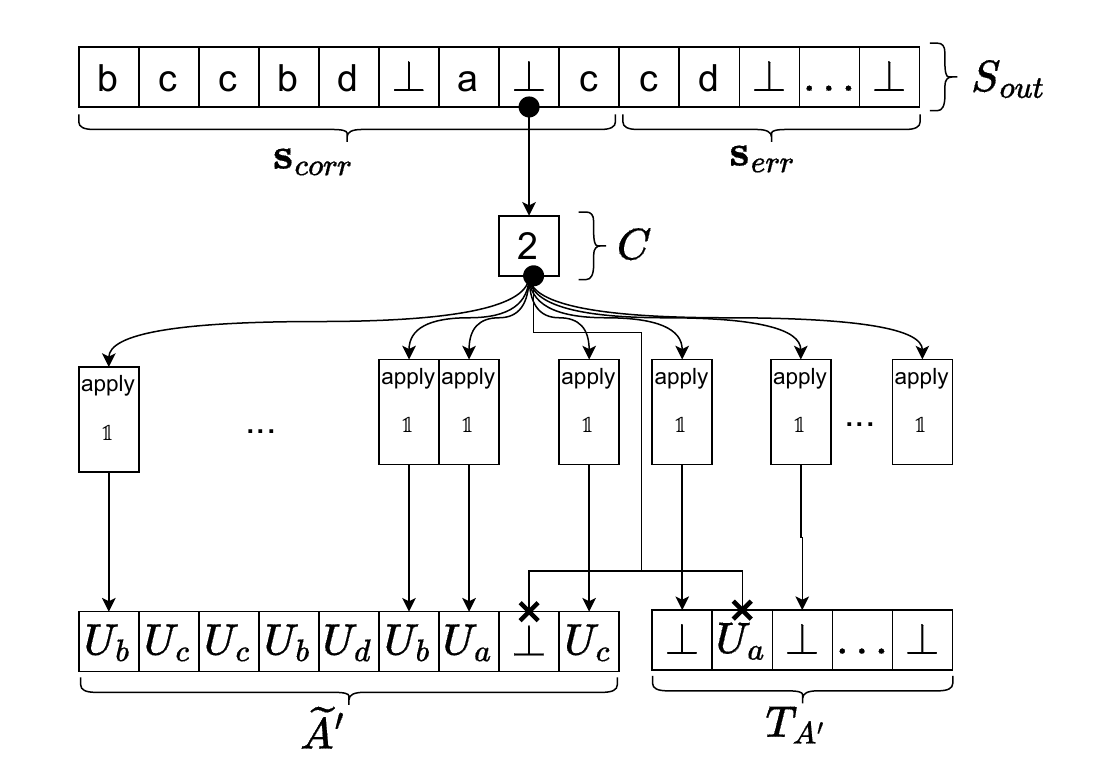}
	\caption{$\tau_1^{\mathrm{IPA}}(8)$ operation. Notice that this is only one step from the loop, which iterates from $i=1$ to $i=9$. The state of the counter system $C$ is $2$ since it has been shifted when $i=6$ (see Figure~\ref{fig:U_shift} and Step~6 of Algorithm~\ref{alg:IPPA}). This operation swaps $\widetilde{A}_S[8]$ (since $i=8$) and $T_{A_S}[2]$ (since state of the counter system $C$ is $2$) and applies $\id$ to systems $\widetilde{A}_S[\neq 8]$ and $T_{A_S}[\neq 2]$. The swap operation occurs, because system ${\bf s}_{\mathrm{cor}}[i]$ is in a state $\perp$. If this system was in a state other than~$\perp$, then the operation $\tau_1^{\mathrm{IPA}}(8)$ would apply the identity $\id$ to all systems.}
	\label{fig:tau_1_IPA_8}
\end{figure}

\subsection{Privacy dilution protocol}\setcurrentname{Privacy dilution protocol}
\label{sec:DPdescription}

We now introduce the privacy dilution protocol, which transforms a private state into a generalized private state.
It uses the fact that the secure one-time pad protocol is a classical correspondent of quantum teleportation.

\begin{theorem}
\label{thm:transformation} 
For any generalized private state
$\gamma(\psi)$ with $|\psi\>=\sum_{a\in {\cal A}}\lambda_a|e_a\>|f_a\>$ and sufficiently large number $n \in \mathbb{N}$, the state
$\gamma(\psi)^{\otimes n}$  can be created by LOCC from a strictly irreducible private state
$\gamma_{d_n}(\Phi^+)\otimes|\Phi^+\>\!\<\Phi^+|_{d'_n}$ with error less than $2\varepsilon$ in trace distance, key rate $\log_2 d_n = \lceil n(S(A_K)_{\psi} +\eta)\rceil$, and $\log_2 d'_n= \lceil 4\delta n\rceil$, where $\eta,\varepsilon,\delta \rightarrow 0$ as $n\rightarrow \infty$.
\end{theorem}

Before proving Theorem~\ref{thm:transformation} in Section~\ref{sec:DPproof}, we 
describe the Privacy Dilution protocol, which assures the validity of the claim made here. To this end, we first define a private state $\gamma_{d_n}(\Phi^+)$, followed by the steps of the Dilution Protocol. 

Let $|\psi\>=\sum_{a\in{\cal A}}\sqrt{\lambda_a}|e_a\>|f_a\>$.
Further denote
\begin{align}
 \psi^{\otimes n} & = \sum_{\bf s} \sqrt{\lambda_{\bf s}} |e_{\bf s} f_{\bf s}\>,   \\
\gamma(\psi)^{\otimes n} & =\sum_{{\bf s,s'}=0}^{|{\cal A}|^{n}-1} \sqrt{\lambda_{\bf s}\lambda_{\bf s'}} |e_{\bf s}f_{\bf s}\>\!\<e_{\bf s'} f_{\bf s'}| \otimes U_{\bf s}\widetilde{\rho}_{A_SB_S} U_{\bf s'}^{\dagger} ,
\end{align}
where more precisely,
\begin{align}
 \widetilde{\rho}_{A_SB_S} & =\rho^{\otimes n}_{A_SB_S},   \\
 {\bf s} & = ({\bf s}[1],\ldots ,{\bf s}[n]) \in {\cal A}^n,\\
U_{\bf s} & = U_{{\bf s}[1]}\otimes U_{{\bf s}[2]}\otimes \cdots  \otimes U_{{\bf s}[n]}, \\
\lambda_{\bf s} & =\lambda_{{\bf s}[1]}\times \lambda_{{\bf s}[2]}\times \cdots \times \lambda_{{\bf s}[n]}.
\end{align}
It will be convenient to work with an approximate version of this state. It is sufficient because $\gamma(\psi)^{\otimes n}$ is isomorphic to $\widetilde{\gamma}(\psi^{\otimes n})$, where in the latter state the shielding system is $\widetilde{\rho}_{A_SB_S}$ and the twisting is constructed by $\{U_{\bf s}\}_{\bf s}$.
We further note that we will be able to construct only an approximate version of $\widetilde{\gamma}_{A_SB_S}$,
namely, one that twists only the typical
subspace, i.e. $\operatorname{span}(\{|e_{\bf s}f_{\bf s}\>\}_{{\bf s}\in T_n^{\delta}})$, of the space $\operatorname{span}(\{|e_{\bf s}f_{\bf s}\>\}_{{\bf s}\in {\cal A}^n})$.
To this end, note that, for each $\varepsilon >0$ and $\delta>0$, there exists sufficiently large $n$ such that
\begin{equation}
    \left\|  \psi^{\otimes n} - \frac1N\sum_{{\bf s},{\bf s}'\in T^{\delta}_{n}}|\psi_{\bf s}\>\!\<\psi_{{\bf s}'}| \right\|_1 \leq \varepsilon,
\end{equation}
where $|\psi_{\bf s}\> \coloneqq \sqrt{\lambda_{\bf s}}|e_{\bf s}\>|f_{\bf s}\>$, $T^{\delta}_n$ is the set of $\delta$-typical sequences of $n$ symbols from the alphabet ${\cal A}$, and $N=\sum_{s\in T^{\delta}_n} \lambda_{\bf{s}}$, so that $N\geq 1-\varepsilon$, due to the law of large numbers (see Section~\ref{sec:preliminaries_strong_typ}). These are sequences with a number of each symbol close to its expected value by $\delta$ in modulus.
Let
\begin{equation}
    \rho_{\delta,n} \coloneqq \frac{1}{N}\sum_{{\bf s},{\bf s}'\in T^{\delta}_{n}}|\psi_{\bf s}\>\!\<\psi_{{\bf s}'}|.
\end{equation}
According to the above, by the data-processing inequality for the trace distance and the triangle inequality, for every $\varepsilon,\delta>0$ and sufficiently large $n$,
\begin{equation}
    \left\|  \widetilde{\gamma}(\psi^{\otimes n})-\widetilde{\gamma}(\rho_{\delta,n})\right\|_1  \leq 3 \varepsilon,
    \label{eq:close_as_desired}
\end{equation}
where $\widetilde{\gamma}(\rho_{\delta,n})$ is defined as
\begin{multline}
    \widetilde{\gamma}(\rho_{\delta,n})\coloneqq  \\
    (1-\varepsilon) \sum_{{\bf s},{\bf s}' \in T_{n}^{\delta}}\frac{\sqrt{\lambda_{\bf s}\lambda_{\bf s'}}}{N}|e_{\bf s}f_{\bf s}\>\!\<e_{{\bf s}'} f_{{\bf s}'}|\otimes U_{\bf s} \widetilde{\rho}_{A_SB_S} U_{{\bf s}'}^{\dagger}\\
    +\varepsilon |\emptyset\emptyset\emptyset\emptyset\>\!\<\emptyset\emptyset\emptyset\emptyset|_{A_KB_KA_SB_S},
    \label{eq:correct_output}
\end{multline}
with $\emptyset \notin \{{\cal A}\cup \{\perp\}\}$.

We now construct a private state from which we will be able to produce the state $\widetilde{\gamma}(\rho_{\delta,n})$, which is sufficient for us due to the above closeness relation.
We design it using an optimal compression encoding~\cite{Cover2006}. This encoding, call it $C$, maps one-to-one $|{\cal A}|$-ry strings ${\bf s}$ of symbols from ${\cal A}$ with entropy
\begin{equation}
    H(\{\lambda_{\bf s}\}_{\bf s})=H\!\left(\left\{\prod_{i=1}^{n}\lambda_{({\bf s})[i]}\right\}_{\bf s}\right)=S(A_K)_{\psi}
\end{equation}
into codewords
 $c({\bf s})$ of length at most $L(n,\eta)\equiv\lceil n( S(A_K)_{\psi} +\eta)\rceil$ for some $\eta>0$, which converges to $0$ as $n\rightarrow \infty$ (here, $H(\cdot)$ denotes the Shannon entropy of a distribution). These codewords are orthogonal when treated as pure states in a computational basis. A simple example of such an encoding is mapping $\delta_1$-strongly typical sequences to all sequences (in lexicographic order) of length $\lceil nS(A_K)_{\psi}(1+\delta_1)\rceil$ and all not $\delta_1$-strongly typical sequences (see definition in~\eqref{eq:def-str-typ} and its properties) to an error symbol $\emptyset$ (which happens with negligible probability $\leq \varepsilon$ for sufficiently large $n$). In that case we identify $\eta$ with $\delta_1 S(A_K)_{\psi}$.  However, any other compression scheme with rate $L(n,\eta)$, which does not map a strongly typical sequence to an error state, works in this case.

We will consider now only the $\delta$-typical sequences $\bf s$,
as these are the only ones 
appearing in $\rho_{\delta,n}$ that get twisted in $\widetilde{\gamma}(\rho_{\delta,n})$.
Let the encoding $C$ assign ${\bf s} \mapsto c({\bf s})$, where $c({\bf s})$ is a codeword that corresponds to ${\bf s}$ via the chosen encoding of string ${\bf s}$, which is of size $L(n,\eta)$ at most. Further, we order each codeword $c({\bf s})$ lexicographically, denoted by $r(c({\bf s}))$, so that they correspond further one-to-one to $|{\cal A}|$-ary strings of length $L(n,\eta)$ through the mapping $R$. Because in the case of the strongly typical sequences, the Shannon compression scheme is lossless, $C$~is a bijection. Therefore there exists a bijection ${\bf s}(r)\coloneqq C^{-1}\circ R^{-1}(r)$ that maps $r$ back to ${\bf s}$.

We are ready to construct $\gamma_{d_n}(\Phi^+)$ from which we will create $\widetilde{\gamma}(\rho_{\delta,n})$, along with the assistance of a sublinear amount of entanglement. The construction is as follows
\begin{align}
 \gamma_{d_n}(\Phi^+) & =\sum_{r,r'} \frac{1}{d_n}|rr\>\!\<r'r'|\otimes U_{{\bf s}(r)}\rho_{A_SB_S}^{\otimes n}U_{{\bf s}(r')}^{\dagger}  
 \label{eq:resource_private_state}
\end{align}
where
\begin{align}
U_{{\bf s}(r)} & = U_{{\bf s}(r)[1]}\otimes \cdots  \otimes U_{{\bf s}(r)[n]}  ,  \\
{\bf s}(r) & \equiv ({\bf s}(r)[1],\ldots ,{\bf s}(r)[n]), \\
\rho_{A_SB_S}^{\otimes n} & \equiv \widetilde{\rho}_{A_SB_S},
\end{align}
as given in the definition
of the state $\widetilde{\gamma}(\rho_{\delta,n})$.

Due to the above choice of a private state, the proposed protocol only modifies the key part of the involved state, while the shield is prepared in advance in an almost correct form, that is up to the permutation of the local subsystems.

Consider $|\phi\> = \sum_{a \in {\cal A}} \sqrt{\lambda_a} |e_a\>$. We first note that by an isometry between entangled and coherent superposition states, for the $\varepsilon, \delta >0$ and $n$ chosen as above, the following inequality holds~\cite{YangWinter}
\begin{equation}
    \left\|  \phi^{\otimes n} - \frac{1}{N}\sum_{{\bf s},{\bf s}'\in T^{\delta}_{n}}|\phi_{\bf s}\>\!\<\phi_{{\bf s}'}| \right\|_1 \leq \varepsilon,
\end{equation}
where $\phi_{\bf s} = \sqrt{\lambda_{\bf s}}|e_{\bf s}\>$.
Next, we will use the typical state 
$\sigma_{\delta,n}\coloneqq \frac{1}{N}\sum_{{\bf s},{\bf s}'\in T^{\delta}_{n}}|\phi_{\bf s}\>\!\<\phi_{{\bf s}'}|$.

The Dilution Protocol consists of the steps we define in what follows. In each step $2,\ldots ,7$, when a defined action encounters a system in the $\emptyset$ symbol, it is assumed to be defined as an identity operation. However, the situation when the system is in the $\emptyset$ state happens with negligible probability $\leq O(\varepsilon)$. We note that steps 1-6 resemble the coherence dilution protocol from~\cite{YangWinter}. We give below the steps of the protocol.
\begin{enumerate}
    \item Alice creates $\phi^{\otimes n}\coloneqq \sum_{{\bf s }=0}^{|{\cal A}|^n-1} \sqrt{\lambda_{\bf s}} |e_{\bf s}\>_{A''}$,
    and transforms it into $\widetilde{\sigma}_{\delta,n}\coloneqq (1-\varepsilon)\sigma_{\delta,n}+\varepsilon |\emptyset\>\!\<\emptyset|$ by replacing an atypical $\bf s$ with an error state $|\emptyset\>\!\<\emptyset|$, which however occurs with probability less than~$O(\varepsilon)$. 
    
    \item Alice applies a    unitary transformation that maps $|e_{\bf s}\>_{A''}$ into  $|{\bf s}\>_{A''}$. Then she performs coherently a classical encoding $C$ based compression algorithm~\cite{Cover2006} which
    transforms $\widetilde{\sigma}_{\delta,n}$ into 
    \begin{equation}
    \sigma'_{\delta,n}=(1-\varepsilon)\frac{1}{N}\sum_{c,c'=0}^{d_n-1} \sqrt{\lambda_{{\bf s}(c)}\lambda_{{\bf s}(c')}} |c\>\!\<c'|+\varepsilon |\emptyset\>\!\<\emptyset| ,     
    \end{equation}
    where $c$ is a codeword of the encoding $C$,
    the classical encoding maps ${\bf s} \mapsto c$,  and ${\bf s}(c)=C^{-1}(c)$.
    On symbol $\emptyset$, the encoding is defined as the identity operation.
    
    \item Alice performs a two-outcome POVM on the system $A''$ that determines if the system is in the state~$|\emptyset\>$ or not. If the system $A''$ is in the state $\emptyset$, Alice publicly communicates this to Bob, and they both replace the system $A'B'$ with the $|\emptyset\emptyset\emptyset\emptyset\>_{A_KB_KA_SB_S}$ state, trace out system $A''$, and stop the Privacy Dilution in a failure state, which happens with probability $\leq O(\varepsilon)$.
    
    Otherwise,
    Alice performs the controlled-XOR operation with $A''$  as the control  and the subsystem~$A_K$  of $\gamma_{d_n}(\Phi^+)_{A_KA_SB_KB_S}$ as the target. This operation transforms the state $|c\>_{A_K}$ into
    $|c\oplus r\>$, with $\oplus$ being a digit-wise addition modulo $|{\cal A}|$ operation.
    
     \item Alice measures the system $A_K$ and communicates the result  of this measurement $x\equiv c\oplus r$ publicly to Bob. 
     
    \item Bob performs the unitary transformation $|x\ominus c\>\rightarrow |c\>$ on system $B_K$.
    
    \item Alice and Bob coherently apply the inverse of the compression algorithm $C$, which transforms  $|c\>$ back into $|{\bf s}\>$ and $|c'\>$ into $|{\bf s}'\>$, respectively.
    
    \item Alice and Bob perform the algorithms to correct phase errors, i.e., apply IPA$\circ$ PEC$\circ$ PA.
    Alice traces out system $A_K$, and Bob traces out the corresponding system holding the state of the message $|x\>\!\<x|$.
    
    \item Alice changes coherently a basis state $|{\bf s}\>_{A''}$ into 
    $|e_{\bf s}\>_{A''}$, by applying 
    \begin{equation}
        W_{A''}=\sum_{{\bf s}\in T^{\delta}_n}|e_{\bf s}\>\!\<{\bf s}|+\sum_{{\bf s} \notin T^{\delta}_n}|{\bf s}\>\!\<{\bf s}|+ |\emptyset\>\!\<\emptyset|,
    \end{equation}
     while Bob changes $|{\bf s}'\>_{B_K}$ to $|f_{{\bf s}'}\>_{B_K}$ analogously. 
\end{enumerate}
In the next section, we prove that
the resulting state (taking into account an error event) of the Dilution Protocol is equal to $\widetilde{\gamma}_{\delta,n}$ given in~\eqref{eq:correct_output}.

\subsection{Correctness of the Privacy Dilution Protocol}\setcurrentname{Correctness of the Privacy Dilution Protocol}
\label{sec:DPproof}
We now give a proof of Theorem~\ref{thm:transformation}.

\begin{proof}[Proof of Theorem~\ref{thm:transformation}]
We comment on the steps of the Privacy Dilution Protocol to show its correctness,  thereby proving Theorem~\ref{thm:transformation}.

For the sake of readability, we omit in the notation the fact discussed in point $2$ that the sequence ${\bf s}$ may be atypical, resulting in state $|\emptyset \emptyset \emptyset \emptyset\>_{A_KB_KA_SB_S}$. We will assume further that each sequence $\bf s$ is typical, i.e., belongs to the set $T_{n}^{\delta}$.

After the second step, that is, compressing the ancillary state by mapping ${\bf s}\mapsto c$ on system $A''$, the total state of $A''A_KB_KA_SB_S$ systems is
\begin{multline}
    \sum_{c,c',r,r'}\frac{1}{d_n}\frac{\sqrt{\lambda_{{\bf s}(c)}\lambda_{{\bf s}(c')}}}{N}|c\>\!\<c'| \otimes |rr\>\!\<r'r'| \\
    \otimes U_{{\bf s}(r)} \widetilde{\rho}_{A_SB_S} U_{{\bf s}(r')}^{\dagger},
\end{multline}
where $d_n = |{\cal A}|^{\lceil L(n,\eta)\rceil}$.
After the third step, which applies the controlled-XOR operation, the total state of the system is as follows
\begin{multline}
    \sum_{c,c',r,r'}\frac{1}{d_n} \frac{\sqrt{\lambda_{{\bf s}(c)}\lambda_{{\bf s}(c')}}}{N}|c\>\!\<c'|_{A''}\otimes  \\
     |r\oplus c,r\>\!\<r'\oplus c',r'|_{AB} \otimes U_{{\bf s}(r)} \widetilde{\rho}_{A_SB_S} U_{{\bf s}(r')}^{\dagger}.
\end{multline}
After measurement of system $A_K$ in step $4$, the state collapses
with probability $\frac{1}{d_n}$ 
to state $|x\>\!\<x|$ with $x = r\oplus c=r'\oplus c'$ on system $A$. Since this implies $r= x \ominus c$  and $r'=x\ominus c'$, (where $\ominus$ is the inverse of the modulo-XOR operation $\oplus$, i.e., it adds an opposite element in a group of addition modulo $|{\cal A}|$), then after tracing out $A_K$, the total state on systems $A''B_KA_SB_S$ is
\begin{multline}
    \sum_{c,c'} \frac{\sqrt{\lambda_{{\bf s}(c)}\lambda_{{\bf s}(c')}}}{N}|c\>\!\<c'|_{A''} \otimes |x\ominus c\>\!\<x\ominus c'|_{B_K} \\
    \otimes U_{{\bf s}(x\ominus c)} \widetilde{\rho}_{A_SB_S} U_{{\bf s}(x\ominus c')}^{\dagger}.
\end{multline}
Further in step $5$, Bob, knowing $x$, transforms $|x\ominus c\>\!\<x\ominus c'| \equiv |r\>\!\<r'|\rightarrow |c\>\!\<c'|$ which can be done unitarily by applying $W\equiv \sum_c|c\>\!\<x\ominus c|$ since $\{|x\ominus c\rangle \}_c$ forms a basis, as does $\{|c\rangle \}_c$. Hence, the total state is as follows
\begin{multline}
    \sum_{c,c'} \frac{\sqrt{\lambda_{{\bf s}(c)}\lambda_{{\bf s}(c')}}}{N}|c\>\!\<c'|_{A''} \otimes |c \>\!\<c'|_{B_K} \\
    \otimes U_{{\bf s}(x\ominus c)} \widetilde{\rho}_{A_SB_S} U_{{\bf s}(x\ominus c')}^{\dagger}.
\end{multline}
In step $6$, both Alice and Bob apply the decompression algorithm, which maps
$|c\>$ back to $|{\bf s}(c)\>$ on Alice's side, and $|c'\> \mapsto |{\bf s}(c')\>$ on Bob's side, so the resulting state is
\begin{multline}
    \sum_{{\bf s},{\bf s}'\in T^{\delta}_n} \frac{\sqrt{\lambda_{{\bf s}}\lambda_{{\bf s}'}}}{N}|{\bf s}\>\!\<{\bf s}'|_{A''} \otimes |{\bf s}\>\!\<{\bf s}'|_{B_K} \\
    \otimes U_{\beta_x({\bf s})} \widetilde{\rho}_{A_SB_S} U_{\beta_x({\bf s}')}^{\dagger},
    \label{eq:prefinal}
\end{multline}
where we identify ${\bf s}(c)$ with ${\bf s}$ and ${\bf s}(c')$ with ${\bf s}'$, which is possible due to the bijection between ${\bf s}(c)$ and $c$ (and between ${\bf s}(c')$ and $c'$). We also note that the index $c$, as we show below, is no longer needed in the description of the state. We introduce a function of $\bf s$ and $\bf s'$ denoted as $\beta_x$ instead. To describe this function, we introduce and apply Lemma~\ref{lem:beta-map}.

The map $\beta_x$ indicates the pattern of the ``phase errors.'' Indeed, the indices of the key part ${\bf s}$ and $\beta_x({\bf s})$ of identifying $U_{\beta_x({\bf s})}$ do not match in general as they should in the state $\widetilde{\gamma}_{\delta,n}$. We would like to
perform the mapping $\beta_x^{-1}$; however, we need to do it using a negligible amount of private key.
In what follows, we will do it
by an LOCC operation accompanied by $O(n\delta)$ copies of the singlet state.
The latter assistance of entanglement will not alter the overall private key cost
of the transformation since we can eventually take the limit $\delta \rightarrow 0$.

We now describe the phase-error correction algorithm invoked in Step~$7 $ of the protocol.
As a warm-up, let us observe that when ${\bf s}$ and $\widehat{{\bf s}}$ are of the same type, then for each ${\bf s}$ the mapping $\beta_x^{-1}$ is a permutation $\pi({\bf s})$.
Hence in the ideal (yet unreachable) case of $\delta=0$, it would be enough to define an error-correcting map $\mathcal{P}_{\operatorname{ideal}}$  as
\begin{equation}
    \mathcal{P}_{\operatorname{ideal}}\coloneqq  \sum_{{\bf s} \in T^{\delta}_n} |{\bf s}\>\!\<{\bf s}|_{A''}\otimes \pi({\bf s})_{A_S}
    \label{eq:ideal},
\end{equation}
and the same map for systems $B_KB_S$. However,
we need to modify the above map
so that it takes care of the fact
that the mapping $\beta_x$ can map ${\bf s}$ into a sequence $\widehat{{\bf s}}$ that is not of the same type. Fortunately, it is an automorphism, so we control how much the ${\bf s}$ and $\beta_x({\bf s})$ sequences differ: for each symbol $a \in \{0,\hdots,|{\cal A}|-1\}$ the number of occurrences of $a$ in ${\bf s}$ (written as $|a({\bf s})|$) differs from $|a(\widehat{{\bf s}})|$ at most by $\lceil 2\delta n\rceil$ (see Lemma~\ref{lem:permutation-typical-seq}). We have defined the action of partial inversion of $\beta_x$, called Permutation Algorithm (PA), by Algorithm~\ref{alg:PPA} given in Section~\ref{sec:alg}.

Let us further denote by $U_{{\bf s}_{\operatorname{cor}}}$ a tensor product $\bigotimes_{i=1}^n U_{{\bf s}_{\operatorname{cor}}[i]}$
of unitaries 
corresponding to the sequence ${\bf s}_{\operatorname{cor}}$ in a way that, if in position $i$, there is the symbol $a$ in ${\bf s}_{\operatorname{cor}}$, then $U_{{\bf s}_{\operatorname{cor}}[i]}=U_a$, and when the symbol $\perp$
is at this position, then
$U_{{\bf s}_{\operatorname{cor}}[i]}=\id$.

The mapping of the PA can be captured by a transformation
of a purification of the state from~\eqref{eq:prefinal}, i.e.,
\begin{equation}
    \sum_{{\bf s}\in T^{\delta}_n}\sqrt{\frac{\lambda_{\bf s}}{N}}|{\bf s}{\bf s}\> (U_{\beta_x({\bf s})}\otimes \id_E) |\chi^{\otimes n}\>_{A_SB_SE}
\end{equation}
(where $|\chi\>_{A_SB_SE}$ is a purification of the state $\widetilde{\rho}_{A_SB_S}$ to system $E$) into a state 
\begin{align}
    &\sum_{{\bf s} \in T^{\delta}_n}\sqrt{\frac{\lambda_{\bf s}}{N}}|{\bf s}{\bf s}\>_{A_KB_K} (U_{{\bf s}_{\operatorname{cor}}}\otimes U_{\widehat{{\bf s}}[q_1]}\otimes\cdots \otimes U_{\widehat{{\bf s}}[q_{f({\bf s})}]}\otimes \nonumber\\ 
    & \id_{f({\bf s})+1}\otimes \cdots  \otimes \id_{L_{\max}}\otimes \mathrm{I}_E ) \widetilde{|\chi\>}\nonumber\\
    &\left[ {\widetilde{A}}_{S{\bf s}_{\operatorname{cor}}},T_{A_S}[1:f({\bf s})], T_{A_S}^{\perp}[f({\bf s})+1 : L_{\mathrm{max}}] \right.\nonumber\\
    &\left. {\widetilde{B}}_{S{\bf s}_{\operatorname{cor}}},T_{B_S}[1:f({\bf s})], T_{B_S}^{\perp}[f({\bf s})+1 : L_{\mathrm{max}}]E\right].
    \label{eq:beforePEC}
\end{align}
In the above for ease of reading we identify $\widetilde{|\chi\>}_Y$ with $\widetilde{|\chi\>}[Y]$. We also added a $\perp$ symbol in superscript in $T_{A_S}^{\perp}[f({\bf s})+1 : L_{\mathrm{max}}]$ to keep in mind that particular systems are in the $\left|\perp\right\rangle$ state.
By $\widetilde{|\chi\>}$ we mean the purification of the shield systems such that whenever ${\bf s}_{\operatorname{cor}}[i]= \perp$,  the state $|\chi_i\>$ (i.e., the $i$-th subsystem of $|\chi\>$) is meant to be the purification of the state $\left|\perp\perp\right\>_{A_SB_S}$ to the system $E$. If this is not the case, the $i$-th subsystem of $|\chi\>$ equals the purification of the state $\rho_{A_SB_S}$.
We have also omitted the systems $A_SB_S[1\ldots n]$ as they are in the state $\left|\perp\right\>^{\otimes n}$. Moreover, $U_{\widehat{{\bf s}}[j_k]}$ are 
unitaries whose indices do not match the index of the key at position $i_k$ in $|{\bf s}\>_{A_K}$ (reflecting the phase error). The $L_{\max}-f({\bf s})$ systems $T_{A_S}[f({\bf s})+1]\cdots T_{A_S}[L_{\max}]$ are in the state $\left|\perp\right\>$.

Next, the PEC procedure described in Algorithm~\ref{alg:PEC} is executed. It begins by teleporting to Bob the $L_{\max}$ states of the systems
\begin{equation}
A_S[j_1]\cdots A_S[j_{f({\bf s})}]\left|\perp\right\>_{j_{f({\bf s})}+1}\cdots \left|\perp\right\>_{L_{\max}},
\end{equation}
which are now stored in system $T_{A_S}$. There are $f({\bf s})$ ``phase errors''.
Namely for $k=1, \ldots ,f({\bf s})$ the $p_k$ are the positions for which ${\bf s}_{\operatorname{cor}}$ has $\left|\perp\right\>$, indicating that ${\bf s}[i_k]\neq \widehat{\bf s}[j_k]$
at this position. In other words, at position $i_k$ the key system has symbol ${\bf s}[i_k]$, while $U_{\widehat{\bf s}[j_k]}$ is performed 
on $T_{A_SB_S}[k]$ instead of $U_{{\bf s}[i_k]}$. For this reason Bob
performs controlled-unitary operations $\tau_1^{\mathrm{PEC}}$ and $\tau_2^{\mathrm{PEC}}$  to correct these errors on systems
$$T_{A_S}[1]\cdots T_{A_S}[{f({\bf s})}]
T_{B_S}[1]\cdots T_{B_S}[{f({\bf s})}],$$ leaving systems 
$$T_{A_S}[f({\bf s})+1]\cdots T_{A_S}[L_{\max}],T_{B_S}[f({\bf s})+1]\cdots T_{B_S}[L_{\max}]$$ untouched.
Indeed, this is what PEC does. The unitary $\tau_1^{\mathrm{PEC}}$ undoes the action of $U_{\widehat{s}[j_k]}$ on system $j_k$ for
each $k=1, \ldots ,f({\bf s})$, leaving these systems in state $\rho_{A_SB_S}$ each. Further
$\tau_2^{\mathrm{PEC}}$ rotates the latter states by
appropriate unitaries $U_a$ demanded by symbols ${\bf s}[i_k]=a$ for each $k$. 
Bob further teleports back the systems that he got from Alice.
After the phase errors are corrected, the state of~\eqref{eq:beforePEC} gets transformed into the following form
\begin{align}
    &\sum_{{\bf s} \in T^{\delta}_n}\sqrt{\frac{\lambda_{\bf s}}{N}}|{\bf s}{\bf s}\>_{A_KB_K} (U_{{\bf s}_{\operatorname{cor}}}\otimes U_{{\bf s}[p_1]}\otimes\cdots\otimes  U_{{\bf s}[p_{f({\bf s})}]}\otimes \nonumber\\ &\id_{f({\bf s})+1}\otimes \cdots  \otimes \id_{L_{\max}} \otimes \id_E)  \widetilde{|\chi\>}\nonumber\\
    &\left[ \widetilde{A}_{S{\bf s}_{\operatorname{cor}}},T_{A_S}[1:f({\bf s})], T_{A_S}^{\perp}[f({\bf s})+1 : L_{\mathrm{max}}] \right.\nonumber\\
    &\left. {\widetilde B}_{S{\bf s}_{\operatorname{cor}}},T_{B_S}[1:f({\bf s})], T_{B_S}^{\perp}[f({\bf s})+1 : L_{\mathrm{max}}]E\right]
\end{align}
where the only difference (apart from the different state of the auxiliary systems that we discuss shortly) is that instead of incorrect $U_{\widehat{\bf s}[i]}$ there comes $U_{{\bf s}[i]}$ which are correct ``phases.'' 
Performing PA involves borrowing the auxiliary systems $\widetilde{A}_S,T_{A_S}S_{\operatorname{out}}$, as well as the ones needed to produce $\widehat S$ and its copy ${\widehat S}_1$.
Then, there comes the Inverting PA (IPA) procedure, the task of which is to (i)
put back the corrected shielding systems as well as the accompanying correct ones from ${\widetilde A}_ST_{A_S}T_{B_S}$ to ${ A_S}{B_S}[1\ldots n]$ systems, (ii) uncompute the auxiliary systems $\widehat S$ and ${\widehat S}_1$ to their original state $|0\>$ of the initial dimension. 
The first task of IPA is rather complicated since it has to rearrange coherently the entries corrected by the PEC algorithm and conserve the order of the remaining ones. This procedure is performed in a control-target manner, where control is on systems ${\bf s}_{\mathrm{cor}}$ and $C$. The former system identifies when the SWAP operation should be applied, while the latter identifies which part of $T_{A_S}$ should be swapped. Section~\ref{sec:RP} provides an extensive description of this procedure.
On the other hand, the second task of IPA is simple since it is completed by applying the inverse of the twisting $\tau_{\mathrm{key}}^{PA}$, i.e., ${\tau_{\mathrm{key}}^{PA}}^{\dagger}$. Further, the $U_{\beta_x}^{\dagger}$ applied twice uncomputes $\widehat S$ to a state of the systems $A$, which is then uncopied by CNOT operations.
Finally, the systems ${\widehat S}^A_1{\widehat S}^B_1S_{\operatorname{out}}^{A}S_{\operatorname{out}}^B{\widetilde A}_ST_{A_S}{\widetilde B}_ST_{B_S}$ are then traced out in their original state.
The resulting state is
\begin{align}
    &\sum_{{\bf s}\in T^{\delta}_n}\sqrt{\frac{{\lambda_{\bf s}}}{N}}|{\bf s}{\bf s}\>_{A_KB_K} U_{{\bf s}}\otimes\id_E |\chi\>_{A_SB_SE}.
\end{align}
The auxiliary systems were traced out in their original state. Hence, this tracing out of auxiliary systems does not disclose any further knowledge to Eve about the above state.

We have then demonstrated explicitly, by constructing the operations of the Dilution Protocol, that they are a composition of twisting unitary operations, teleportation of a state whose dimension is independent of the value of the key, and giving back the auxiliary systems in their original state. Since the teleported system is also teleported back, we can consider their composition as an identity operation on the state. 
Since the states of auxiliary systems are returned to their original state, and the protocol transforms $A_KB_KA_SB_S[1\ldots n]$ back to the desired state, there is no security leakage when we trace out the auxiliary systems.

We note here that mapping of Step~8 of the Dilution Protocol given in Sec.~\ref{sec:DPdescription}, i.e., $|{\bf s}\>_{A_K} \rightarrow |e_{\bf s}\>_{A_K}$ on Alice's key part and $|{\bf s}\>_{B_K}\rightarrow |f_{\bf s}\>_{B_K}$ on Bob's one, can also be done by local unitary transformations. 

Finally, we observe that the system  $A_K$, which is communicated by Alice to Bob publicly, is initially correlated with the shielding system: $|x\>\!\<x|\otimes U_{\beta_x({\bf s})}\widetilde{\rho}_{A_SB_S} U_{\beta_x({\bf s})}^{\dagger}$.
However, after the protocol is done in Step~$8$, the function $\beta_x$ is inverted:
for every ${\bf s}$ on shield there is $U_{\bf s}(\widetilde{\rho}_{A_SB_S})U_{\bf s}^{\dagger}$. Hence, the shielding systems become independent from $|x\>\!\<x|_{A_K}$, and the total output state is in the desired form. Note here that the parties trace out the systems holding information about  $|x\>\!\<x|$ only at the end of the protocol since the PA procedure uses this register
to perform~$U_{\beta_x}$.
 To summarize, the resulting state of the protocol is in the desired form of the approximate twisted typical subspace of the state $|\psi\>^{\otimes n}$. Due to~\eqref{eq:close_as_desired}, we have obtained a state that approximates another, which is locally equivalent to $\gamma(\psi)^{\otimes n}$. This concludes the proof of the correctness of the Privacy Dilution protocol described above
 and in Theorem~\ref{thm:transformation}.
\end{proof}

\begin{lemma}
\label{lem:beta-map}
    The map $\beta_x :T_n^\delta \to T_n^\delta$ defined on a sequence~${\bf s}(c)$ as $\beta_x({\bf s}(c))=\widehat{\bf s}\equiv{\bf s}(x\ominus c)$ is a bijection given by the following formula
    \begin{equation}
        \beta_x \coloneqq  C^{-1} \circ R^{-1}\circ \ominus_{|{\cal A}|}\circ X_-\circ C,
        \label{eq:beta}
    \end{equation}
    where the $X_-$ map subtracts $x$ digit-wise modulo $|{\cal A}|$ and $\ominus_{|{\cal A}|}$ takes a digit-wise opposite number modulo $|{\cal A}|$, i.e., such that together they add to $0$ modulo $|{\cal A}|$.
\end{lemma}

\begin{proof}
    First, by the definition of the notation ${\bf s}(r)$, for $\widehat{r}\coloneqq x\ominus c$ the bijection $C^{-1}\circ R^{-1}$ maps $\widehat{r}$ to ${\bf s}(x\ominus c)$ (via some $c'$ not necessarily equal to c, see diagram below). On the other hand, we can bijectively reconstruct the string ${\bf s}(c)$ from $\widehat{r}$. We do so by first obtaining from $\widehat{r}$ the string ${\widetilde r}=c\ominus x$, i.e., the opposite number to $x \ominus c$, which is done by the map $\ominus_{|{\cal A}|}$. We compose it by adding $x$ mod $|{\cal A}|$, i.e., with the map $X_+$ that results in $c({\bf s})$. We finally compose the latter maps with $C^{-1}$ to obtain ${\bf s}(c)$. 
    Now, by inverting the last sequence of bijections, i.e., getting 
    $\ominus_{|{\cal A}|}\circ X_-\circ C$
    (note that $\ominus_{|{\cal A}|}$ is self-inverse), we obtain bijective mapping from ${\bf s}$ to $\widehat{r}$. Composing the latter bijection with $C^{-1}\circ R^{-1}$ we obtain the $\beta_x$ given in~\eqref{eq:beta}. The action of the $\beta_x$ bijection can be captured through the following diagram
    \begin{equation*}
        \begin{tikzcd}[every arrow/.append style={shift left}]
            {\bf s}(x\ominus c) \arrow{r}{C} \arrow{d}{\beta_x^{-1}} & c' \arrow{l}{C^{-1}}\arrow{r}{R} & \hat{r} = x\ominus c \arrow{l}{R^{-1}} \arrow[leftrightarrow]{d}{\ominus_{|{\cal A}|}} \\
            {\bf s}(c) \arrow{u}{\beta_x}\arrow{r}{C} & c \arrow{l}{C^{-1}} \arrow{r}{X_{-}} & c\ominus x \arrow{l}{X_{+}}
        \end{tikzcd}
    \end{equation*}
    This concludes the proof.
\end{proof}

\begin{remark}
We have designed a privacy dilution protocol that transforms an ideal private state into a generalized private state by LOCC. Our protocol uses at least $o(1)$ rate of privacy in the form of pure entanglement. It remains open to determine if pure entanglement is actually needed for this task.
It is worth noticing that whenever the twisting unitaries, $U_a$, are incoherent operations (composition of incoherent gates), then the DP can be completed by incoherent operations~\cite{RMP_coherence}.
\end{remark}

\section{Regularized key of formation upper bounds key cost}\setcurrentname{Regularized key of formation upper bounds key cost}
\label{sec:regularized_key_of_formation_is_cost}

We now show that the regularized key of formation is an upper bound on the key cost. 
For ease of writing, we will make use of the following definition.
\begin{definition}
A private state $\gamma(\Phi^+)$ is $\theta$-related to a generalized private state $\gamma(\phi)$ if there exists $n\in \mathbb{Z}^+$ such that  
from $\gamma(\Phi^+)$ one can create $\gamma(\phi)^{\otimes n}$ by LOCC,  up to an error $\theta$ in trace distance. For ease of writing, we sometimes omit $\theta$ whenever irrelevant or known from the context.
\end{definition}
Let us note here that for each generalized private state $\gamma(\psi)$ and $\theta >0$ there exists a private state $\theta$-related to $\gamma(\psi)$, as assured
by Theorem~\ref{thm:transformation}. We are ready to state the main result of this section.

\begin{theorem}
\label{thm:regularization}
For every bipartite state $\rho_{AB}$,
\begin{equation}
K_F^{\infty}(\rho_{AB}) \geq K_C(\rho_{AB}). 
\end{equation}
\end{theorem}

\begin{proof}
Let $\{(p_i,\gamma(\psi_i))\}_{i=1}^{k}$ be an optimal ensemble for~$K_F$.  For this part of the proof, we assume that $K_F$ is attained on a finite ensemble with $k$ elements. By Lemma~\ref{lem:finite}, we have that $k\leq (|{\cal A}|\times |{\cal B}|)^2+1$. The case when it is not attained will be considered later as a small modification of what follows.
To prove the inequality, we first note,  following~\cite{Hayden_2001}, that
by typicality arguments, for all $\varepsilon, \delta_0 >0$ and sufficiently large $n$, we have that
\begin{equation}
\begin{aligned}
\rho^{\otimes n} & \approx_{\varepsilon}\rho_n\\
& \equiv\sum_{{\bf s}\in T^{\delta_0}_n} p_{\bf s} \gamma(\psi_{{\bf s}[1]})\otimes \gamma(\psi_{{\bf s}[2]})\otimes \cdots \otimes \gamma(\psi_{{\bf s}[n]}). 
\end{aligned}
\label{eq:typical_version}
\end{equation}
The latter state is locally unitarily equivalent to
\begin{equation}
{\widetilde \rho}_n\equiv \sum_{{\bf s}\in T^{\delta_0}_n} p_{\bf s}\gamma(\psi_{\bf s}),
\label{eq:target}
\end{equation}
where the string ${\bf s}$ belongs to a strongly typical set $T^{\delta_0}_n$ with a property that it contains each state $\gamma(\psi_i)$, at least $\lfloor p_i n - \delta_0p_in \rfloor$ and at most $\lceil p_i n + \delta_0p_in \rceil$ times.
We note that $\delta_0$ can be arbitrarily small for sufficiently large $n$.

Let us now fix ${\bf s}\in T^{\delta_0}_n$. Then $\gamma(\psi_{\bf s})$ is such that each state $\gamma(\psi_i)$ in $\gamma(\psi_{\bf s})$ occurs  $l_i(n)\in \mathbb{N}$ times, which  satisfies 
\begin{equation}\lfloor p_in - \delta_0 p_i n\rfloor\leq l_i(n)  \leq \lceil np_i + \delta_0 p_i n \rceil .
\end{equation}
 We aim to create an approximate version of $\gamma(\psi_{\bf s})$.

We will do it using a sub-protocol that has access to a private state $\gamma_{d_n(i)}(\Phi^+)$  $\varepsilon_i$-related to $\gamma(\psi_i)$; this sub-protocol will create an $\varepsilon_i$-approximation of the latter state. For the sake of this proof, it suffices to set $\varepsilon_i = \varepsilon$. We also set $\delta_i>0$ arbitrarily for $i\in \{1,\ldots k\}$. As noted, the existence of $\gamma_{d_n(i)}(\Phi^+)$ is assured by Theorem~\ref{thm:transformation}. More precisely, given sufficiently large $n$ we will obtain sufficiently large $d_n(i)$ satisfying
\begin{equation}
\log d_n(i)=\lceil l_i(n)\times(S(A_K)_{\psi_i}+\eta_i + 2\delta_i)\rceil,
\end{equation} 
(where $\eta_i = \delta_iS(A_K)_{\psi_i}$) such that the Dilution Protocol  described in Section~\ref{sec:dilution} creates from $\gamma_{d_n(i)}(\Phi^+)$ a state that $\varepsilon_i$-approximates in trace norm $\gamma(\psi_i)^{\otimes l_i(n)}$. Let us denote this protocol as ${\cal P}_i$.

We note here that due to the definition of $K_F$, each $\gamma(\psi_i)$ is strictly irreducible. Hence by~\eqref{eq:resource_private_state}, $\gamma_{d_n(i)}(\Phi^+)$ is also a strictly irreducible private state. Thus, using ${\cal P}_i$ as a subroutine, we will create $\rho_{AB}$ from a strictly irreducible private state, as demanded by the definition of $K_C$.

Performing protocol ${\cal P}_i$ for each $i \in [k]$, by the triangle inequality, we obtain $\gamma_{\approx}(\psi_{\bf s})$ that is a $\theta \coloneqq  \sum_i \varepsilon_i$-approximation of the state
\begin{equation}
\bigotimes_{i=1}^k \gamma(\psi_i)^{\otimes l_i(n)}.
\end{equation}
Let us note here that the actual amount of key needed to create $\gamma_{\approx}(\psi_{\bf s})$ varies depending on the value of $\bf s$, due to fluctuations of the number of occurrences of states $\gamma(\psi_i)$, i.e., $l_i(n)$  in sequence $\bf s$.
However, by definition of $K_C$, we need 
a {\it single} private state from which 
we can create our target state given in~\eqref{eq:target}. We will therefore choose
a private state with an amount of key exceeding $K_F$ by a vanishing factor (proportional to $\delta_0 \times K_F\leq \delta_0 d_{AB}$ by a vanishing constant) which will be ready for the worst-case (the largest) amount of key needed to create $\gamma_{\approx}(\psi_{\bf s})$. Namely, we will define a private state which
is $O(\epsilon)$-related to $\gamma(\psi_{\bf s})$ for {\it every} $\delta_0$-strongly typical $\bf s$.

To achieve this task, we choose 
${\gamma}_{d_n^{*}}(\Phi^+)$ to be a private state $\theta$-related to a generalized private state of the form 
\begin{equation}\gamma^+\coloneqq \bigotimes_{i=1}^{k}\gamma(\psi_i)^{\otimes l_i(n)^+}.
\end{equation} 
The dimension $d_n^{*}$ of the key part of ${\gamma}_{d_n^{*}}(\Phi^+)$ satisfies
\begin{equation}
 \log_2  d_n^{*}\coloneqq  \sum_{i=1}^{k} \lceil l_i(n)^{+}\times (S(A_K)_{\psi_i} + \eta_i + 2\delta_i)\rceil,
\end{equation}
where
\begin{equation}
l_i(n)^+\coloneqq  \lceil n p_i  +\delta_0p_in\rceil.
\end{equation}

We focus on this case, as the above state uses the largest private key. Intuitively, it can allow for the creation of $\gamma_{\approx}(\psi_{\bf s})$ for any other $\delta_0$-typical ${\bf s} \neq {\bf s}_0$ by performing partial trace operations if less key is needed to produce $\gamma(\psi_{\bf s})$. 

To finalize the procedure of creating an approximate version of $\gamma(\psi_{\bf s})$, we need to permute locally (in the same way for Alice and Bob) subsystems of $\widetilde{\gamma}_{\approx}(\psi_{\bf s})$ by some permutation $\pi_{
{\bf s}}$, defined as
\begin{align}
{\bf s}_0\equiv(\underbrace{1\hdots 1}_{\text{$l_1(n)$ times }}, \underbrace{2\hdots 2}_{\text{$l_2(n)$ times }},\hdots, \underbrace{k\hdots k}_{\text{$l_k(n)$ times }}) \rightarrow {\bf s}.
\end{align}

We are ready to formalize the protocol ${\cal P}$, which takes as input ${\gamma}_{d_n^{*}}(\Phi^+)$ and outputs an approximation of 
${\widetilde \rho}_n = \sum_{{\bf s} \in T^{\delta_0}_n} p_s\gamma(\psi_{\bf s})$. It 
consists of the following steps:
\begin{enumerate}
    \item Draw  ${\bf s}$ at random with probability $p_{\bf s}$ in order to
     produce an approximation of $\gamma(\psi_{{\bf s}})$ from the related private state $\gamma_{d_n^{*}}(\Phi^+)$.
\item For each $i\in[k]$, apply ${\cal P}_i$  to produce a $\theta$-approximation  $\gamma^+_{\approx}(\psi_{\bf s})$ of the state $\gamma^+$.
\item If ${\bf s}\neq {\bf s}_0$, trace out some of the subsystems of $\gamma^+_{\approx}(\psi_{\bf s})$ to obtain the $\theta$-approximation $\gamma_{\approx}(\psi_{\bf s})$ of 
\begin{equation}
\bigotimes_{i=1}^k \gamma(\psi_i)^{\otimes l_i(n)}.
\end{equation}
\item If ${\bf s}\neq{\bf s_0}$, apply the permutation $\pi_{\bf s}$ to the local subsystems of $\gamma_{\approx}(\psi_{\bf s})$
to create a $\theta$-approximation of  $\gamma(\psi_{\bf s})$.
\item Erase the symbol ${\bf s}$, producing a mixed state over typical labels ${\bf s}$.
\end{enumerate} 

By construction, the above protocol ${\cal P}$
creates a $\theta$-approximation of the state 
$\widetilde{\rho}_n$, which is locally unitarily equivalent to $\rho_n$ given in~\eqref{eq:typical_version}. Let $V\equiv V_A\otimes V_B$ denote the unitary transforming $\widetilde{\rho}_n$ into $\rho_n$. Now, by the triangle inequality, the construction of ${\cal P}$, and~\eqref{eq:typical_version}, we obtain 
\begin{align}
    & \left\|  V{\cal P}({\gamma}_{d_n^{*}}(\Phi^+))V^{\dagger} -\rho^{\otimes n}\right\|  _1 \notag \\
    & \leq \left \|   V{\cal P}({\gamma}_{d_n^{*}}(\Phi^+))V^{\dagger} - {\rho}_{n}\right \|  _1  
    + \left \|  \rho^{\otimes n} - {\rho}_{n}\right \|  _1 \\
    & \leq  \theta +\varepsilon,
\end{align}
for sufficiently large $n$. We note here that
without loss of generality $\varepsilon +\theta \leq (k+1)\varepsilon\leq \varepsilon(d_{AB}^2+1)$ (see Lemma~\ref{lem:finite}). 

Since $V$ is local, ${\cal P}'\equiv V\circ {\cal P}$ is an LOCC protocol that creates an approximation of $\rho_{AB}^{\otimes n}$, using $\log d_n^{*}$ bits of private key. Hence, its rate is an upper bound on~$K_C$. We now compute the latter rate. As we will see, it asymptotically reaches the value of $K_F(\rho_{AB})$.

In total, the rate of the protocol ${\cal P}'$ is
\begin{align}
    &\frac{\log_2 d_n^{*}}{n} \notag \\
    & = \frac{1}{n}\sum_{i=1}^{k} \left\lceil{l_i^+(n)} [S(A_K)_{\psi_i}+\eta_i + 2\delta_i]\right\rceil  \\
    & \leq \frac{1}{n}\sum_{i=1}^{k} {l_i^+(n)} [S(A_K)_{\psi_i}+\eta_i + 2\delta_i] + \frac{k}{n}\\
    & \leq \sum_{i=1}^{k} {\frac{l_i^+(n)}{n}} [S(A_K)_{\psi_i}+\eta_i + 2\delta_i] + \frac{d_{AB}^2+1}{n}\\
    &\leq \sum_{i=1}^{k} {p_i(1+\delta_0)} [S(A_K)_{\psi_i}+\eta_i + 2\delta_i]  \nonumber \\
    &\qquad+\frac{k[S(A_K)_{\psi_i}+\eta_i + 2\delta_i] }{n}
    + \frac{d_{AB}^2+1}{n}\\
    & \leq \sum_{i=1}^{k} {p_i(1+\delta_0)} [S(A_K)_{\psi_i}+\eta_i + 2\delta_i]  \nonumber \\
    &\qquad + \frac{(d_{AB}^2+1)(\log(d_{AB}) +3)}{n}
    + \frac{d_{AB}^2+1}{n}\\
    & \rightarrow_{n\rightarrow \infty}\sum_{i=1}^{k} p_i(1+\delta_0)[S(A_K)_{\psi_i} +2\delta_i +\eta_i] 
    \\
    & = (1+\delta_0)\sum_{i=1}^k p_i S(A_K)_{\psi_i} + (1+\delta_0)\sum_{i=1}^k p_i (\eta_i +2\delta_i) \\
    & = \sum_{i=1}^k p_i S(A_K)_{\psi_i} + \delta_0\sum_{i=1}^k p_i S(A_K)_{\psi_i} \nonumber \\
    & \qquad + (1+\delta_0)\sum_{i=1}^k p_i (\eta_i +2\delta_i)\\
     & = K_F(\rho_{AB}) + \delta_0\sum_{i=1}^k p_i S(A_K)_{\psi_i} \nonumber\\
     & \qquad + (1+\delta_0)\sum_{i=1}^k p_i (\eta_i +2\delta_i).
    \end{align}
 In the above, the first inequality is by an upper bound by $1$ to each $\lceil .\rceil$ term. Next, we bound $k$ from above by $d_{AB}^2+1$ by Lemma \ref{lem:finite}.
Further, we upper bound the $\lceil . \rceil$ in the definition of each $l_i^+(n)$ and upper bound the term $S(A_K)_{\psi_i} +\eta_i +2\delta_i$ by $\log d_{AB}+3$ by noting that local entropy cannot exceed the logarithm of the local dimension, and without loss of generality, $\eta_i,\delta_i$ can be considered as $\leq 1$. We finally take the limit of large $n$ and reorder terms.

We have proven then, that for all $\varepsilon>0$, there exists a sufficiently large $N \in {\mathbb N}$ such that for each $n>N$, a private state $\gamma_{d_n^{*}}(\Phi^+)$ and a LOCC protocol ${\cal P}'$ producing, by acting on the latter state, an output state that is at least $(d_{AB}^{2}+1)\varepsilon$-close  in trace norm to $\rho^{\otimes n}$. Moreover, since $K_C^{\varepsilon}(\rho_{AB}^{\otimes n})$ is by definition the infimum over such protocols, then $(1/n)K_C^{\varepsilon}(\rho_{AB}^{\otimes n}) \leq \frac{\log_2 d_n^{*}}{n}$. Combining it with the latter chain of inequalities, we obtain
\begin{align}
&\limsup_{n\to\infty} \frac{1}{n}K_C^{\varepsilon}(\rho_{AB}^{\otimes n}) \notag \\
& \leq \limsup_{n\to\infty} \frac{\log_2 d_n^{*}}{n} \\
& \leq
K_F(\rho_{AB}) + \delta_0\sum_{i=1}^k p_i S(A_K)_{\psi_i}  \nonumber\\
& \qquad + (1+\delta_0)\sum_{i=1}^k p_i (\eta_i +2\delta_i) .
\end{align}
This holds for arbitrarily small $\delta_i>0$ (so as $\eta_i$, since $\eta_i = \delta_iS(A_K)_{\psi_i}$) ($i=0,\ldots k$). Hence 
\begin{equation}
    \limsup_{n\to\infty} \frac{1}{n}K_C^{\varepsilon}(\rho_{AB}^{\otimes n})\leq K_F(\rho_{AB}).
\end{equation}
But this holds true for each $\varepsilon > 0$, and so $K_C(\rho_{AB})\leq K_F(\rho_{AB})$.

The latter inequality is proved under the assumption that $K_F$ is attained by an ensemble of $k$ generalized private states with $k$. If it is not the case, then $K_F$ is by the definition of the infimum is a limit of quantities $\sum_{i=1}^{k(m)}p_i^{(m)}S_{A_K}(\gamma(\psi_i)^{(m)})$ with $k(m)\leq d_{AB}^{2}+1$ (see Lemma~\ref{lem:finite}) and $m\rightarrow \infty$. We can then fix $\kappa>0$ arbitrarily small for which there exists $t_0$ such that $K_F(\rho_{AB}) +\kappa = \sum_{i=1}^{k(t_0)}p_i^{(t_0)}S_{A_K}(\gamma(\psi_i)^{(t_0)})$, and prove by the above argument that $K_F +\kappa\geq K_C$. Since this holds for every $\kappa >0$, we have $K_F\geq K_C$ in general (without assumption about the attainability of the infimum in the definition of $K_F$).

By analogy with~\cite{Hayden_2001}, the same argument as above proves $K_F(\rho^{\otimes t}_{AB})/t \geq K_C(\rho^{\otimes t}_{AB})/t$ for every finite, fixed positive integer $t \in \mathbb{Z}^+$.
Since $K_C$ is additive, i.e. $K_C(\rho^{\otimes t}_{AB}) = t K_C(\rho_{AB})$ (see Lemma~\ref{lem:k-c-additive}), then
\begin{align}
    \lim_{t\to\infty} \frac{1}{t} K_F(\rho^{\otimes t}_{AB}) & \geq \lim_{t\rightarrow \infty} \frac{1}{t} K_C(\rho^{\otimes t}_{AB}) \\
    & =
    \lim_{t\rightarrow \infty} \frac{1}{t} t K_C(\rho_{AB}) \\
    & =K_C(\rho_{AB}).
\end{align}
So we have proved that $K^{\infty}_F(\rho_{AB})\geq K_C(\rho_{AB})$.
We have used above the fact, that $K_C$ is additive on tensor product of a bipartite state. We prove this below, partially in the spirit of the proof for distillable entanglement given in~\cite{Eisert2022}.
\end{proof}

\begin{lemma}\label{lem:k-c-additive}
$K_C$ is additive; i.e., for every positive integer $r \in \mathbb{Z}^+$,
    \begin{equation}
        K_C(\rho^{\otimes r}) = r K_C(\rho).
    \end{equation}
\end{lemma}

\begin{proof}
    \begin{align}
        K_{C}(\rho^{\otimes r}) & =  \sup_{\varepsilon \in (0,1)} \limsup_{n\rightarrow\infty} \frac{1}{n} K_{C}^{\varepsilon}(\rho^{\otimes nr}) \\
        & = 
         r \sup_{\varepsilon \in (0,1)} \limsup_{n\rightarrow\infty} \frac{1}{rn} K_{C}^{\varepsilon}(\rho^{\otimes nr}) \\
        & = r K_{C}(\rho).
    \end{align}
    The second equality holds by the next lemma.
\end{proof}
\begin{lemma}
    For each $r \in \mathbb{Z}^+$, the following equality holds
    \begin{multline}
        \sup_{\varepsilon \in (0,1)} \limsup_{n\rightarrow\infty} \frac{1}{n} K_{C}^{\varepsilon}(\rho^{\otimes n}) \\
        = \sup_{\varepsilon \in (0,1)} \limsup_{n\rightarrow\infty} \frac{1}{rn} K_{C}^{\varepsilon}(\rho^{\otimes nr}).
    \end{multline}
\end{lemma}

\begin{proof}
    Let $\varepsilon>0$. Since $\{rn\}_n$ is a subsequence of a sequence $\{n\}_n$, then 
    \begin{equation}
        \limsup_{n\rightarrow\infty} \frac{1}{rn} K_{C}^{\varepsilon}(\rho^{\otimes nr}) \leq \limsup_{n\rightarrow\infty} \frac{1}{n} K_{C}^{\varepsilon}(\rho^{\otimes n}).
    \end{equation}
    We will show by contradiction that the converse inequality holds. Assume that there exists $M > 0$ such that
    \begin{align}
        \eta_{\operatorname{sup}}^r & \coloneqq \limsup_{n\rightarrow\infty} \frac{1}{rn} K_{C}^{\varepsilon}(\rho^{\otimes nr}) \\
        & < M < \limsup_{n\rightarrow\infty} \frac{1}{n} K_{C}^{\varepsilon}(\rho^{\otimes n}).
    \end{align}
    Let $0 < \varepsilon_1 < (M-\eta_{\operatorname{sup}}^r)/2$. 
    From the definition of the limit superior, there exists $N$ such that, for all $m\in \mathbb{Z}^+$ and $n \coloneqq  r\cdot m > N$,
    \begin{equation}
         \frac{1}{n} K_C^{\varepsilon}(\rho^{\otimes n}) \leq \eta_{\operatorname{sup}}^r + \varepsilon_1 < M.
    \end{equation}
    Let us denote as $\{\overline{n}_k\}_k$, with $\overline{n}_k > N$, a sequence of all indices such that
    \begin{equation}
        M < \frac{1}{\overline{n}_k} K_C^{\varepsilon}(\rho^{\otimes \overline{n}_k}).
    \end{equation}
    Observe that for each $\overline{n}_k$ there exists $s_k \leq r$ such that $\overline{n}_k + s_k = r\cdot l \eqqcolon \underline{n}_k > N$. This defines a subsequence $\{\underline{n}_k\}_k$ for which \begin{equation}
        \left|\frac{1}{\underline{n}_k} K_C^{\varepsilon}(\rho^{\otimes \underline{n}_k})- \eta_{\operatorname{sup}}^r\right| \leq \varepsilon_1.
    \end{equation}
    We can choose then a subsequence $\{\underline{n}_{k_l}\}_l$ such that 
    \begin{equation}
        \eta \coloneqq  \lim_{l\to\infty} \frac{1}{\underline{n}_{k_l}} K_C^{\varepsilon}(\rho^{\otimes \underline{n}_{k_l}})  \leq \eta_{\operatorname{sup}}^r + \varepsilon_1 < M,
    \end{equation}
    and also choose the corresponding subsequence $\{\overline{n}_{k_l}\}_l$. For simplicity, from now on we will denote $\{\underline{n}_{k_l}\}_l$ as $\{i_k\}$ and $\{\overline{n}_{k_l}\}_l$ as $\{n_k\}_k$.\\
    
    Now we summarize the current construction and reduce redundancy in the notation. We chose two precisely prepared subsequences $\{n_k\}_k$, $\{i_k\}_k$ such that
    \begin{align}
        &\forall_k \, i_k = n_k + s_k, \quad s_k\leq r \label{eq:i-k-is-n-k-plus-1}, \\
        &\forall_k \, \frac{1}{i_k} K_{C}^{\varepsilon}(\rho^{\otimes i_{k}}) < \eta_{\operatorname{sup}}^r+\varepsilon_1 < M < \frac{1}{n_k} K_{C}^{\varepsilon}(\rho^{\otimes n_{k}}), \\
        \label{eq:kc-i-k-less-kc-n-k}
        &\lim_{k\to \infty} \frac{1}{i_k} K_{C}^{\varepsilon}(\rho^{\otimes i_{k}}) \textrm{ (denoted as $\eta$) } \leq \eta_{\operatorname{sup}}^r+\varepsilon_1.
    \end{align}
        Now we can write $\rho^{\otimes i_k} = \rho^{\otimes n_k}\otimes \rho^{\otimes s_k}$. Since $K_C^{\varepsilon}$ is an infimum over $\Lambda\in$ LOCC and $\gamma_{d}\in \IR$, then for each $\delta_{i_k}>0$ there exist $\Lambda_{i_k}\in$ LOCC, $\gamma_{d_{i_k}}$ such that $\Lambda_{i_k}(\gamma_{d_{i_k}}) \approx_{\varepsilon} \rho^{\otimes i_k}$ and
    \begin{equation}
        \frac{1}{i_k} K_{C}^{\varepsilon}(\rho^{\otimes i_k}) + \delta_{i_k} \geq \frac{\log d_{i_k}}{i_k}.
    \end{equation}
    But $\Lambda_{i_k}$, composed with a partial trace that will trace out $s_k$ redundant systems, is an LOCC protocol $\Lambda_{n_k}$ such that $\Lambda_{n_k}(\gamma_{d_{i_k}}) \approx_{\varepsilon} \rho^{\otimes n_k}$. Since it is a particular protocol, then we have
    \begin{equation}
        \frac{\log d_{i_k}}{n_k} \geq \frac{1}{n_k} K_{C}^{\varepsilon}(\rho^{\otimes n_k}).
    \end{equation}
    Combining above inequalities we obtain
    \begin{equation}
        K_{C}^{\varepsilon}(\rho^{\otimes i_k}) + i_k\delta_{i_k} \geq \log d_{i_k} \geq K_{C}^{\varepsilon}(\rho^{\otimes n_k}).
    \end{equation}
    But this holds true for all $\delta_{i_k} > 0$, so we conclude that 
    \begin{align}
        K_{C}^{\varepsilon}(\rho^{\otimes i_k}) \geq K_{C}^{\varepsilon}(\rho^{\otimes n_k})\nonumber\\
        \frac{1}{n_k}K_{C}^{\varepsilon}(\rho^{\otimes i_k}) \geq \frac{1}{n_k}K_{C}^{\varepsilon}(\rho^{\otimes n_k}),\label{eq:kc-eps-i-k-geq-kc-eps-i-k}
    \end{align}
    which is true for all $k$. Combining~\eqref{eq:kc-i-k-less-kc-n-k},~\eqref{eq:kc-eps-i-k-geq-kc-eps-i-k} and~\eqref{eq:i-k-is-n-k-plus-1} we obtain
    \begin{equation}
        \frac{1}{i_k-s_k}K_{C}^{\varepsilon}(\rho^{\otimes i_k}) \geq \frac{1}{n_k}K_{C}^{\varepsilon}(\rho^{\otimes n_k}) > \frac{1}{i_k}K_{C}^{\varepsilon}(\rho^{\otimes i_k}).
    \end{equation}
    It only remains to observe that
    \begin{multline}
        \left(\frac{1}{i_k-s_k}K_{C}^{\varepsilon}(\rho^{\otimes i_k}) - \frac{1}{i_k}K_{C}^{\varepsilon}(\rho^{\otimes i_k})\right) = \\
        \frac{s_k}{i_k-s_k}\left(\frac{1}{i_k}K_{C}^{\varepsilon}(\rho^{\otimes i_k})\right) \to 0\cdot\eta = 0
    \end{multline}
    as $k\to\infty$. So by the sandwich rule
    \begin{align}
        &\lim_{k\to\infty} \frac{1}{n_k}K_{C}^{\varepsilon}(\rho^{\otimes n_k}) \leq \eta_{\operatorname{sup}}^r + \varepsilon_1 < M \\
        &< \frac{1}{n_k}K_{C}^{\varepsilon}(\rho^{\otimes n_k}) \quad \forall_k,
    \end{align}
    which is a contradiction. So 
    \begin{equation}
        \limsup_{n\rightarrow\infty} \frac{1}{rn} K_{C}^{\varepsilon}(\rho^{\otimes nr}) \geq \limsup_{n\rightarrow\infty} \frac{1}{n} K_{C}^{\varepsilon}(\rho^{\otimes n})
    \end{equation}
    for each $\varepsilon > 0$.
\end{proof}

\begin{remark} \label{rem:broaden-Kf-definition}
    As noted in Remark \ref{rem:broaden-Kc-definition}, the existence of entangled but zero distillable key states would result in a redefinition of the key cost. Similarly, $K_F$ would change to a smaller quantity $K_F'$. Still, however, the following meaningful bound $K_C' \leq K_C \leq K^{\infty}_F$, resulting from our findings, would hold.
\end{remark}

\section{Key cost of generalized private states}\setcurrentname{Key cost of generalized private states}
\label{sec:key_cost_of_generalized_private_states}

In this section, we prove several lemmas,
which allow us to conclude that the key
of formation, key cost, and distillable key are all equal for (irreducible) generalized private states.
We begin with a series of lemmas,
which allow us to prove the main result of this section, that is, Theorem~\ref{thm:key_cost_for_private}.

\begin{lemma}
\label{lem:local_entropies}
For $\gamma(\psi)$  a generalized private state, the following equality holds
\begin{equation}
S_{A_K}[\gamma(\psi)]=S_{A_K}(\psi).
\end{equation}
\end{lemma}

\begin{proof}
Consider that
\begin{align}
    S_{A_K}[\gamma(\psi)] & = S(\Tr_{B_KA_SB_S}[\gamma(\psi)]) \\
 & = S(\Tr_{B_KA_SB_S} [{\widetilde{U}}^{\dagger} \gamma(\psi) {\widetilde{U}}])\\
 & = S(\Tr_{B_KA_SB_S} [\psi \otimes \rho_{A_SB_S}]) \\
 &=S(\Tr_{B_K} [\psi]) \\
 & = S_{A_K}(\psi). 
\end{align}
In the above, the first equality is by definition. The second comes from the fact that the twisting $U$ is a unitary that can be substituted with one
acting only on systems $B_KA_SB_S$: $\widetilde{U}=\sum_{i}|i\>\!\<i|_{B_K} \otimes U_i^{A_SB_S}$ in the case of generalized private states
(for this idea, see~\cite{Horodecki_smallD}), 
and because the trace does not depend on the basis in which it is performed (here, the trace over systems $B_KA_SB_S$ does not depend on the basis $\widetilde{U}^{\dagger}$ which is defined on systems $B_KA_SB_S$ only). The last equality follows from the definition of the generalized private state. 
\end{proof}

It was claimed, yet not argued explicitly in~\cite{HorLeakage}, that generalized private states are irreducible when the conditional states $U_i \rho_{A_SB_S} U_i^{\dagger}$ are separable for each $i \in [d_k]$. We provide an explicit argument below.
\begin{lemma}
    \label{lem:key_for_gamma}
    A strictly irreducible generalized private state $\gamma(\psi)$ 
    satisfies 
    \begin{equation}   K_D(\gamma(\psi))=S_{A_K}(\psi)=E_R(\gamma(\psi))=E_R^{\infty}(\gamma(\psi)).
    \end{equation}
\end{lemma}

\begin{proof}
Consider an arbitrary generalized private state $\gamma(\psi) = U(\psi_{A_KB_K}\otimes \rho_{A_SB_S})U^\dag$, with $U = \sum_{i}|ii\>\!\<ii|_{A_KB_K}\otimes U_i^{A_SB_S}$. We first show that $K_D(\gamma(\psi))\geq S_{A_K}(\psi)$ and then that $E_R(\gamma(\psi)) \leq S_{A_K}(\psi)$. Since $E_R\geq E^{\infty}_R\geq K_D$~\cite{pptkey,keyhuge}, the claim follows by proving these two inequalities.

To show that $K_D(\gamma(\psi)) \geq S_{A_K}(\psi)$, we note that, after the measurement of system~$A_K$ of $\gamma(\psi)$ in the computational basis,
the system $A_KB_K$ is in state $\rho_{A_KB_K}\coloneqq \sum_{i}\mu_i |ii\>\!\<ii|_{A_KB_K}$
such that $H(\{\mu_i\})= S_{A_K}(\psi)=I(A_K;B_K)_{\rho_{A_KB_K}}$.
Therefore, the Devetak--Winter protocol~\cite{Devetak_2005} applied to systems
$A_KB_K$ produces $I(A_K;B_K)-I(A_K;E) = I(A_K;B_K)$ of key, as 
$I(A_K;E)=0$ by the construction of the $\gamma(\psi_i)$.
This proves $K_D(\gamma(\psi))\geq S_{A_K}(\psi)$.

To show that $E_R(\gamma(\psi)) \leq S_{A_K}(\psi)$, let us choose a state
$\sigma_{A_KB_KA_SB_S} \coloneqq  \sum_i \mu_i |ii\>\!\<ii|_{A_KB_K}\otimes U_i \rho_{A_SB_S} U_i^{\dagger}$, which is, by the assumption of strict irreducibility of $\gamma(\psi)$, separable. We then observe that 
\begin{align}
    D(\gamma(\psi)\|  \sigma) & = D(\psi_{A_KB_K}\otimes \rho_{A_SB_S}\|  \sigma_{A_KB_K}\otimes \rho_{A_SB_S})\\
    & = D(\psi_{A_KB_K}\|  \sigma_{A_KB_K})\\
    & =S_{A_K}(\psi),
\end{align}
where the first equality
follows from the fact that the relative entropy is invariant 
under joint application of a unitary transformation $U^{\dagger}(\cdot)U$ to both its arguments. The second comes from the identity
\begin{equation}
D(\rho\otimes \sigma\|  \rho'\otimes \sigma) = D(\rho\|  \rho'),    
\end{equation}
and the last stems from a direct computation. By the above argument, we have by the definition of $E_R$ that $E_R\leq S_{A_K}(\psi)$ as the choice of $\sigma_{A_KB_KA_SB_S}$ may be suboptimal when taking the infimum of the relative entropy ``distance'' from separable states. 
\end{proof}

Due to Lemma~\ref{lem:key_for_gamma}, we can rephrase Definition~\ref{def:formation} as follows.
\begin{definition}
\label{def:k_f_k_d}
The key of formation of a state $\rho$ is defined as
\begin{equation}
K_F(\rho) \coloneqq \inf_{\left\{\sum_k p_k \gamma(\psi_k)=\rho\right\}} \sum_k p_k K_D(\gamma(\psi_k)),
\end{equation}
where the infimum is taken over strictly irreducible generalized private states.
\end{definition}

We are now ready to show that the key of formation of a strictly irreducible private state $\gamma(\psi)$ is equal to $S_{A_K}(\psi)$.

\begin{lemma}
\label{lem:K_F_for_gamma}
For every strictly irreducible generalized private state $\gamma(\psi)$,
\begin{equation}
K_F(\gamma(\psi)) = S_{A_K}(\psi).    
\end{equation}
\end{lemma}

\begin{proof}
First, by the definition of $K_F$, $\gamma(\psi)$ itself is a
 valid (singleton) decomposition into irreducible generalized private states.
Hence $K_F(\gamma(\psi))\leq S_{A_K}(\psi)$. On the other hand, for $\kappa >0$, let 
the  ensemble $\{p_i, \gamma(\psi_i)\}$ of $\gamma$ be $\kappa$-optimal, giving $\sum_ip_i S(A_K)_{\psi_i}=K_F+\kappa$. Then
\begin{align}
    K_F(\gamma)+\kappa & = \sum_i p_i K_D(\gamma(\psi_i)) \\
    & \geq K_D\!\left(\sum_i p_i \gamma(\psi_i)\right) \\
    & =
    K_D(\gamma) \\
    & = S_{A_K}(\psi),
\end{align}
where the inequality is due to the fact that $K_D$ is convex
on mixtures $\sum_i p_i \sigma_i$ of states satisfying
$K_D(\sigma_i)=E_R(\sigma_i)$ due to convexity of $E_R$ and the fact that it is an upper bound on $K_D$ (see Proposition 4.14 in~\cite{karol-PhD}). Since the above inequality is true for all $\kappa>0$ by the definition of $K_F$, we obtain the thesis.
\end{proof}

We will compute the entanglement
of formation and entanglement cost
for generalized private states.
In analogy to pure states, 
we will have $K_D(\gamma)=K_C(\gamma)=K_F(\gamma)$, as stated below. However, the key is not equal to the entropy of a whole subsystem of a bipartite state~$\gamma$
but the entropy of a part of the subsystem (system $A_K$ only).

We will need one more lemma to show the aforementioned equivalence. 
\begin{lemma}
\label{lem:kf_greater_than_key}
For every bipartite state $\rho_{AB}$, the following inequalities hold
\begin{equation} K_F(\rho_{AB})\geq K_F^{\infty}(\rho_{AB}) \geq K_D(\rho_{AB}).
\end{equation}
\end{lemma}

\begin{proof}
There are multiple ways to show this. One 
stems from the fact that $K_F\geq K^{\infty}_F$ by subadditivity of $K_F$ (Corollary~\ref{cor:additive}), and the fact that
$K_F^{\infty}\geq K_C\geq K_D$ (Theorem~\ref{thm:regularization}).
\end{proof}

We are ready to state the main result of this section. 
\begin{theorem}
\label{thm:key_cost_for_private}
For a strictly irreducible generalized private state $\gamma(\psi)\equiv \gamma_{A_KA_SB_KB_S}$ with key part $A_KB_K$, the following equalities hold
\begin{align}
    K_C(\gamma)=K_D(\gamma)&=E_R^{\infty}(\gamma)= E_R(\gamma)=K_{F}(\gamma)\nonumber \\
&=K_{F}^\infty(\gamma)=S_{A_K}(\gamma)=S_{A_K}(\psi).
\end{align}
\end{theorem}

\begin{proof}
The last equality $S_{A_K}(\gamma(\psi))=S_{A_K}(\psi)$ follows from Lemma~\ref{lem:local_entropies}.
Furthermore, by Lemma~\ref{lem:key_for_gamma},
we have
\begin{equation}
 K_D(\gamma(\psi))=E_R^{\infty}(\gamma(\psi))=E_R(\gamma(\psi))=S_{A_K}(\psi).   
\end{equation}
We have further that $K_F(\gamma(\psi))=S_A(\psi)$ by Lemma
\ref{lem:K_F_for_gamma}. It is thus sufficient
to prove that $K_C(\gamma)= K_F(\gamma)$. To see this, we note first that by the 
above considerations, $K_F(\gamma)=K_D(\gamma)$. Further we note that $K_F$ satisfies $K_F\geq K_F^{\infty}(\gamma) \geq K_D(\gamma)$ by Lemma~\ref{lem:kf_greater_than_key}. This fact implies that 
\begin{equation}
K_F(\gamma) \geq K_F^{\infty}(\gamma) \geq K_D(\gamma) =K_F(\gamma),     
\end{equation}
and hence $K_F(\gamma)=K_F^{\infty}(\gamma)$. We then use Theorem~\ref{thm:regularization}, which states that $K_F^{\infty}\geq K_C$ and reuse the argument stated in the proof of Lemma~\ref{lem:kf_greater_than_key}, which states that $K_F^{\infty} \geq K_C \geq K_D$. Combining all the above arguments we obtain
\begin{equation}
    K_F(\gamma) \geq K_F^{\infty}(\gamma) \geq K_C(\gamma) \geq K_D(\gamma) = K_F(\gamma),
\end{equation}
which finishes the proof.
\end{proof}

\begin{remark}\label{rem:squashed}
    In entanglement theory, every proper entanglement monotone lies between the distillable entanglement and the entanglement cost~\cite{H3}. However, not every proper entanglement monotone lies between the distillable key and the key cost. An example of this contrasting behavior involves the squashed entanglement which for the so-called flower states~\cite{locking} (that are strictly irreducible private bits) amounts to $E_{\operatorname{sq}}(\gamma_{A_KB_KA_SB_S})=1+\frac{\log |A_S|}{2} > 1=K_C(\gamma_{A_KB_KA_SB_S})$ for $|A_S|\geq 2$~\cite{ChristandlWinter05}. This is also a counterexample which shows that $E_{\operatorname{sq}}$ does not satisfy the equality from Theorem \ref{thm:key_cost_for_private}.
\end{remark}

\section{Yield-cost relation for privacy}\setcurrentname{Yield-cost relation for privacy}
\label{sec:yield_cost_relation}

In this section we derive the yield-cost relation for privacy, following the results for entanglement in~\cite{W21}. We note here that in~\cite{RyujiRegulaWilde} the yield-cost relation for a number of resource theories has been developed. However, inspecting the assumptions of Theorem~8 of~\cite{RyujiRegulaWilde} shows that the resource theory of private key (under the assumption that SEP is the free set of it)  does satisfy one of them. Indeed, although it holds true that $D_{\min}=D_{\max}$ (see~\eqref{eq:min-h} and~\eqref{eq:max-rel} for the definitions) for private states,
the equality $D_{\min,\operatorname{aff(SEP)}}=D_{\max}$ does not hold, where \begin{align}
    D_{\min,\operatorname{aff(SEP)}}(\rho) & \coloneqq \inf_{\sigma\in\operatorname{aff(SEP)}}D_{\min}(\rho\Vert \sigma)\\
    & =\inf_{\sigma\in\operatorname{aff(SEP)}}D^{\varepsilon=0}_h(\rho\Vert \sigma)
\end{align}
and $\operatorname{aff(SEP)}$ is the affine hull of the set $\SEP$. Indeed, like in the case of entanglement theory, $D_{\min,\operatorname{aff(SEP)}}=0$ in this case,
since the set $\SEP$ spans the whole space of self-adjoint operators acting on a finite-dimensional Hilbert space (see~\cite{Regula2020}). Whether one can use Theorem 8 of~\cite{RyujiRegulaWilde} is still open due to possible relation $D_{\min,\SEP}(\gamma_{d_k,d_s})=D_{{r},\SEP}(\gamma_{d_k,d_s})$ for all $\gamma_{d_k,d_s}\in{\rm SIR}$, where \begin{multline}
    D_{{r},\SEP}(\rho)\coloneqq \\
    \inf\bigg\{ \log_2(1+r): \frac{\rho+r\sigma}{1+r}\in\SEP, \sigma\in\SEP \bigg\}
\end{multline}
is a measure of the robustness of entanglement~\cite{VT99} of the state $\rho$. However, our findings provide a tighter bound on $K_D^{\varepsilon}$ by a function of $K_C^{\varepsilon'}$ for $\varepsilon=\varepsilon'<\frac{1}{2}$ than the relation given via the Theorem 8 of~\cite{RyujiRegulaWilde}, if it holds.

Before we state the yield-cost relation for the one-shot case, let us define the one-shot distillable key.
\begin{definition}
Fix $\varepsilon\in [0,1]$.
The one-shot distillable key $K_{D}^{\varepsilon}(\rho)$ of a state $\rho_{AB}$ is defined as 
\begin{equation}
K_{D}^{\varepsilon}(\rho) \coloneqq \sup_{\substack{\mc{L} \in \LOCC,\\ \gamma_{d_k,d_s} \in \IR}} \left\{
\begin{array}[c]{c}
\log_2 d_k :
\\ \frac{1}{2}\norm{\mc{L}(\rho)- \gamma_{d_k,d_s}}_1\leq \varepsilon
\end{array}
\right\},
\end{equation}
where the supremum is taken over every LOCC channel~$\mathcal{L}$ and every strictly irreducible private state $\gamma_{d_k,d_s}$ with an arbitrarily large, finite shield dimension $d_s\geq 1$. 
\end{definition}

\begin{theorem}
\label{thm:second_low}
For every bipartite state $\rho$, the following inequality holds
\begin{equation}
K_D^{\varepsilon_2}(\rho)\leq K_C^{\varepsilon_1}(\rho) +\log_2\!\left(\frac{1}{1-(\varepsilon_1+\varepsilon_2)}\right) .
\end{equation}
for $\varepsilon_1+\varepsilon_2 <1$ and $\varepsilon_1,\varepsilon_2 \in [0,1]$.
\end{theorem}
\begin{proof}
We first prove that for every $\gamma_{d_k,d_s}$ that is  irreducible, the following inequality holds
\begin{equation}
\begin{aligned}
    E_R^\varepsilon(\gamma_{d_k,d_s})& =\inf_{\sigma \in\operatorname{SEP} } D^{\varepsilon}_h(\gamma_{d_k,d_s}\Vert\sigma)\\
    & \leq \log_2 d_k - \log_2 (1-\varepsilon).
\end{aligned}    
\label{eq:main}
\end{equation}
Let us fix then $\gamma_{d_k,d_s} = U (\Phi^+\otimes \rho_{A_SB_S})U^{\dagger}\in\mathrm{SIR}$, where $U=\sum_i |ii\>\!\<ii|\otimes U_i$ and $U_i\rho_{A_SB_S}U_i^{\dagger} \in \SEP$ for all $i \in \{0,\hdots,d_k-1\}$.

Now, to see the above inequality, let us first note that it suffices to prove the following lower bound
\begin{equation}
    \sup_{\sigma\in \SEP} \inf_{\Omega \in T} \Tr\left(\Omega \sigma\right)  \geq \frac{1-\varepsilon}{d_k},
    \label{eq:maineq}
\end{equation}
where
\begin{equation}
T \coloneqq \{0\leq \Omega \leq {\id}_{d_k^2\times d_s^2}, \Tr(\Omega \gamma_{d_k,d_s}) \geq 1-\varepsilon\} .   
\end{equation}
Indeed, by multiplying  both sides of~\eqref{eq:main} by $-1$, we get
\begin{equation}
    \sup_{\sigma \in {\SEP}} \log_2\!\left(\inf_{\Omega \in T} \Tr(\Omega \sigma)\right)\geq - \log_2 d_k + \log_2 (1-\varepsilon).
\end{equation}
Now, taking the power of two on both sides, we get
\begin{equation}
    \sup_{\sigma \in {\SEP}} \inf_{\Omega \in T} \Tr(\Omega \sigma)\geq \frac{ 1-\varepsilon}{d_k}.
\end{equation}
The above inequality follows from Theorem~\ref{thm:lower_bound} below and the fact that the relative $\varepsilon$-hypothesis divergence is an upper bound on the one-shot key, by using the idea of the proof of Theorem $1$ in~\cite{W21}.

Consider two LOCC channels, ${\mathcal L}_{A_1^KB_1^KA_1^SB_1^S\rightarrow AB}$ and ${\mathcal L}_{AB\rightarrow A_2^KB_2^KA_2^SB_2^S}$ such that
${\mathcal L}_{A_1^KB_1^KA_1^SB_1^S\rightarrow AB}$ maps $\gamma_{d_k^1,d_s^1}$ up to $\varepsilon_1$ in normalized trace distance into $\rho_{AB}$ and ${\mathcal L}_{AB\rightarrow A_2^KB_2^KA_2^SB_2^S}$ maps $\rho_{AB}$ into $\gamma_{d_k^2,d_s^2}$ up to $\varepsilon_2$ in normalized trace distance. By composing these maps,
we get a new one satisfying
\begin{multline}
     \frac{1}{2} \big\|{\mathcal L}_{AB\rightarrow A_2^KB_2^KA_2^SB_2^S}\circ{\mathcal L}_{A_1^KB_1^KA_1^SB_1^S\rightarrow AB}(\gamma_{d_k',d_s'}) - \\
     \gamma_{d_k'',d_s''}\big\|_1 
    \leq   \varepsilon_1 + \varepsilon_2.
\end{multline}
It implies that 
\begin{equation}
    \log_2 d_k'' \leq \log_2 d_k' + \log_2\!\left(\frac{1}{\varepsilon'}\right),
\end{equation}
with $\varepsilon' \coloneqq  \varepsilon_1 + \varepsilon_2$ because
\begin{align}
    \log_2 d_k'' & \leq {K_D^{\varepsilon'}(A_1^KA_1^S:A_1^KB_1^S)}_{\gamma_{d_k',d_s'}}\\
    & \leq E_R^{\varepsilon'}({\gamma_{d_k',d_s'}}) \\
    & \leq 
    \log_2 d_k' + \log_2\!\left(\frac{1}{1-\varepsilon'}\right),
\end{align}
where the first inequality is due to the fact
that the distillation of $\gamma_{d_k'',d_s''}$ from
$\gamma_{d_k',d_s'}$ via $\rho_{AB}$ may be a sub-optimal process of distillation $\gamma_{d_k'',d_s''}$ from $\gamma_{d_k',d_s'}$ in general. The next inequality follows from \cite[Theorem 8]{Das2021}. The last one follows from Theorem~\ref{thm:lower_bound}. Taking supremum
over processes that produce $\gamma_{d_k'',d_s''}$ and
infimum over processes that create $\rho_{AB}$ from $\gamma_{d_k',d_s'}$, we obtain the statement of the theorem.
\end{proof}

\begin{theorem}
\label{thm:lower_bound}
For every strictly irreducible private state $\gamma_{d_k,d_s}= U(\Phi^+\otimes \rho_{A_SB_S})U^{\dagger}$,
the following inequality holds
\begin{equation}
    \sup_{\sigma\in \SEP} \inf_{\Omega \in T} \Tr\left(\Omega \sigma\right) \geq \frac{(1-\varepsilon)}{d_k},
    \label{eq:claim}
\end{equation}
where $T=\{0\leq \Omega \leq {\id}_{d_k^2\times d_s^2}, \Tr(\Omega \gamma_{d_k,d_s}) \geq 1-\varepsilon\}$.
\end{theorem}
\begin{proof}
Consider $\sigma =U(\sigma_k \otimes \rho_{A_SB_S})U^{\dagger}\in \SEP$ as an ansatz, where $\sigma_k = (1/d_k)^{-1} \sum_{i}  \op{ii}$. Recall that $U=\sum_{i=0}^{d_k-1} \op{ii}\otimes U_i$ and that, by the assumption of strict irreducibility, $U_i \rho_{A_SB_S} U_i^{\dagger}$ is separable for each $i$. Observe then that 
\begin{equation}
    \Tr( \Omega U(\Phi^+_{d_k} \otimes \rho_{A_SB_S})U^{\dagger}) \geq 1-\varepsilon
\end{equation}
is equivalent by cyclicity of trace to 
\begin{equation}
    \Tr (U^{\dagger}\Omega U (\Phi^+_{d_k} \otimes \rho_{A_SB_S})) \geq 1-\varepsilon.
\end{equation}
Denoting now $\widetilde{\Omega} \equiv U^{\dagger} \Omega U $,
note then that
\begin{align}
    \Tr( \Omega \sigma) & = \Tr (\Omega U(\sigma_k \otimes \rho_{A_SB_S}) U^{\dagger})\\
    & =
    \Tr (\widetilde{\Omega}(\sigma_k \otimes \rho_{A_SB_S})).
\end{align}
These equalities follow by the choice of $\sigma$ and cyclicity of trace.
We have then that the left-hand side (LHS) of the formula~\eqref{eq:maineq} is lower bounded (by the particular choice of $\sigma$) as follows
\begin{equation}
    \inf_{\widetilde{\Omega} \in T_1} \Tr(\widetilde{\Omega} (\sigma_k \otimes \rho_{A_SB_S})) ,
    \label{eq:task}
\end{equation}
where
\begin{equation}
    T_1 \coloneqq \{0\leq \widetilde{\Omega} \leq {\id}_{d_k^2\times d_s^2}, \Tr(\widetilde{\Omega}( \Phi^+_{d_k}\otimes \rho_{A_SB_S})) \geq 1-\varepsilon\}.
\end{equation}

The above quantity is an instance of semi-definite optimization. Indeed, we can cast it as follows
\begin{equation}
\begin{aligned}
&\inf_{\widetilde{\Omega}}\, \Tr(\widetilde{\Omega} (\sigma_k \otimes \rho_{A_SB_S})) \\
&     0\leq \widetilde{\Omega} \leq {\id}_{d_k^2\times d_s^2},\\
& \Tr(\widetilde{\Omega}( \Phi^+_{d_k}\otimes \rho_{A_SB_S})) \geq (1-\varepsilon)
\end{aligned}
\end{equation}
which by weak duality admits a lower bound of the form 
\begin{equation}
    \begin{aligned}
    &\sup_{y,Y} y(1-\varepsilon) -\Tr(Y)\label{eq:last_bound}\\
&  Y \geq (y\Phi^+_{d_k} -\sigma_k)\otimes \rho_{A_SB_S},\\
& y, Y \geq 0.
\end{aligned}
\end{equation}
Indeed, to see this, consider that
\begin{align}
& \inf_{\widetilde{\Omega}\geq0}\left\{
\begin{array}
[c]{c}%
\operatorname{Tr}[\widetilde{\Omega}\left(  \sigma_{k}\otimes\rho_{A_SB_S}\right)  ]:\\
\widetilde{\Omega}\leq\id,\\
\operatorname{Tr}\!\left[  \widetilde{\Omega}\left(  \Phi_{d_{k}}^{+}\otimes
\rho_{A_SB_S}\right)  \right]  \geq1-\varepsilon
\end{array}
\right\}  \nonumber\\
& =\inf_{\widetilde{\Omega}\geq0}\left\{
\begin{array}
[c]{c}%
\operatorname{Tr}[\widetilde{\Omega}\left(  \sigma_{k}\otimes\rho_{A_SB_S}\right)  ]+\\
\sup_{Y\geq0}\operatorname{Tr}\!\left[  Y\left(  \widetilde{\Omega
}-\id\right)  \right]  +\\
\sup_{y\geq0}y\left( \begin{array}[c]{c} 1-\varepsilon\\
-\operatorname{Tr}\!\left[  \widetilde
{\Omega}\left(  \Phi_{d_{k}}^{+}\otimes\rho_{A_SB_S}\right)
\right]  \end{array}\right)
\end{array}
\right\}  \\
& =\inf_{\widetilde{\Omega}\geq0}\sup_{\substack{Y\geq0,\\y\geq0}}\left\{
\begin{array}
[c]{c}%
\operatorname{Tr}[\widetilde{\Omega}\left(  \sigma_{k}\otimes\rho_{A_SB_S}\right)  ]+\\
\operatorname{Tr}\!\left[  Y\left(  \widetilde{\Omega}-\id%
\right)  \right]  +\\
y\left(  1-\varepsilon-\operatorname{Tr}\!\left[  \widetilde{\Omega}\left(
\Phi_{d_{k}}^{+}\otimes\rho_{A_SB_S}\right)  \right]  \right)
\end{array}
\right\}  \\
& \geq\sup_{\substack{Y\geq0,\\y\geq0}}\inf_{\widetilde{\Omega}\geq0}\left\{
\begin{array}
[c]{c}%
\operatorname{Tr}[\widetilde{\Omega}\left(  \sigma_{k}\otimes\rho_{A_SB_S}\right)  ]+\\
\operatorname{Tr}\!\left[  Y\left(  \widetilde{\Omega}-\id%
\right)  \right]  +\\
y\left(  1-\varepsilon-\operatorname{Tr}\!\left[  \widetilde{\Omega}\left(
\Phi_{d_{k}}^{+}\otimes\rho_{A_SB_S}\right)  \right]  \right)
\end{array}
\right\}  \\
& =\sup_{\substack{Y\geq0,\\y\geq0}}\inf_{\widetilde{\Omega}\geq0}\left\{
\begin{array}
[c]{c}%
y\left(  1-\varepsilon\right)  -\operatorname{Tr}[Y] +\\
\operatorname{Tr}\!\left[  \widetilde{\Omega}\left(  Y-\left(  y\Phi_{d_{k}%
}^{+}-\sigma_{k}\right)  \otimes\rho_{A_SB_S}\right)  \right]
\end{array}
\right\}  \\
& =\sup_{\substack{Y\geq0,\\y\geq0}}\left\{
\begin{array}
[c]{c}%
y\left(  1-\varepsilon\right)  -\operatorname{Tr}[Y]+\\
\inf_{\widetilde{\Omega}\geq0}\operatorname{Tr}\!\left[  \widetilde{\Omega
}\left(  \begin{array}[c]{c}
Y-(  y\Phi_{d_{k}}^{+}-\sigma_{k})\\  \otimes\rho
_{A_SB_S}\end{array}\right)  \right]
\end{array}
\right\}  \\
& =\sup_{Y\geq0,y\geq0}\left\{
\begin{array}
[c]{c}%
y\left(  1-\varepsilon\right)  -\operatorname{Tr}[Y]:\\
Y\geq\left(  y\Phi_{d_{k}}^{+}-\sigma_{k}\right)  \otimes\rho_{A_SB_S}%
\end{array}
\right\}  .
\end{align}
The first equality follows by the introduction of Lagrange multipliers and because%
\begin{equation}
\sup_{Y\geq0}\operatorname{Tr}\!\left[  Y\left(  \widetilde{\Omega
}-\id\right)  \right]  =+\infty
\end{equation}
if and only if $\widetilde{\Omega}\not \leq \id$ and%
\begin{equation}
\sup_{y\geq0}y\left(  1-\varepsilon-\operatorname{Tr}\!\left[  \widetilde
{\Omega}\left(  \Phi_{d_{k}}^{+}\otimes\rho_{A_SB_S}\right)
\right]  \right)  =+\infty
\end{equation}
if and only if $\operatorname{Tr}\!\left[  \widetilde{\Omega}\left(  \Phi
_{d_{k}}^{+}\otimes\rho_{A_SB_S}\right)  \right]  \not \geq
1-\varepsilon$. The second equality follows by moving suprema to the outside.
The inequality follows from the max-min inequality. The third equality follows
from some basic algebra. The penultimate equality follows from moving the
infimum inside. The final equality follows because%
\begin{equation}
\inf_{\widetilde{\Omega}\geq0}\operatorname{Tr}\!\left[  \widetilde{\Omega
}\left(  Y-\left(  y\Phi_{d_{k}}^{+}-\sigma_{k}\right)  \otimes\rho
_{A_SB_S}\right)  \right]  =-\infty
\end{equation}
if and only if $Y\not \geq \left(  y\Phi_{d_{k}}^{+}-\sigma_{k}\right)
\otimes\rho_{A_SB_S}$.

We now observe that any choice of $y$ and $Y$ in~\eqref{eq:last_bound}
gives another lower bound on the above quantity,
as it is a supremum over these variables.
We will choose $Y=0$ and $y=\frac{1}{d_k}$.
We need to argue that this choice satisfies
the constraints. To show this, it suffices to prove
that
\begin{equation}
    \left(\frac{1}{d_k} \Phi^+_{d_k} - \sigma_k\right)\otimes \rho_{A_SB_S} \leq Y=0.
\end{equation}
Since $\rho_{A_SB_S} \geq 0$, as $\rho_{A_SB_S}$ is a state, the above holds if and only if
\begin{equation}
    \frac{1}{d_k} \Phi^+_{d_k} - \sigma_k\leq 0.
\end{equation}
To prove this, we observe that the state $\sigma_k$ is equivalent to a maximally entangled state measured locally by a von Neumann measurement in the computational basis. Such a measurement
can be simulated by applying random unitary transformations that perform {\it dephasing}.
We take $d_k$ unitary transformations of the form
~\cite{Groisman2005}: 
 \begin{equation}
     V_l \coloneqq  \sum_{j=0}^{d-1}e^{2\pi i l j/d_k} |j\>\!\<j|.
 \end{equation}
 Given this choice, it follows that
\begin{equation}
    \sigma_k = \frac{1}{d_k}\sum_{l=0}^{d_k-1}(V_l\otimes {\id})\Phi^+_{d_k} (V_l\otimes {\id})^{\dagger}.
\end{equation}
We also note that the states $\op{\psi_l}\equiv (V_l\otimes {\id})\Phi^+_{d_k} (V_l^{\dagger}\otimes {\id})$ are mutually orthogonal~\cite{FerraraChristandl}, and form a subset of the basis of maximally entangled states in $d_k\otimes d_k$.
Thus expanding matrix $\frac{1}{d_k}\Phi^+_{d_k} - \sigma_k$ in this basis we see that
\begin{equation}
    \frac{1}{d_k}\Phi^+_{d_k} -\frac{1}{d_k}\Phi^+_{d_k} -\frac{1}{d_k}\sum_{l>0}^{d_k-1}(V_l\otimes {\id})\Phi^+_{d_k} (V_l^{\dagger}\otimes {\id}),
\end{equation}
which is negative, as we aimed to argue, since the first two terms cancel out, and the reminder has negative
eigenvalues in the considered basis. By the choice of 
$y$ and~\eqref{eq:last_bound} we conclude that the infimum in~\eqref{eq:task} is bounded from below by $\cfrac{1-\varepsilon}{d_k}$ as claimed. 
\end{proof}

As a corollary of Theorem~\ref{thm:second_low}, it follows that all strictly irreducible private states are to a large extent ``units'' of privacy, i.e., their cost is $\approx$ $\log_2 d_k$. We state this fact in the corollary below.
\begin{corollary}
\label{prop:pbitbounds}
Let $\gamma_{d_k,d_s}$ be a strictly irreducible private state. Then
\begin{equation}
\log_2 d_k - \log_2\!\left(\frac{1}{1-\varepsilon}\right)\leq    K_C^{\varepsilon}(\gamma_{d_k,d_s}) \leq \log_2 d_k.
\end{equation}
\end{corollary}

\begin{proof}
The  inequality on the RHS stems from the definition of $K_C^{\varepsilon}$, while the inequality on the LHS is due to Theorem~\ref{thm:second_low}, by picking $\varepsilon_1 = \varepsilon$ and $\varepsilon_2 = 0$, and the fact that $K_D^{0}(\gamma_{d_k,d_s})=\log_2 d_k$.
\end{proof}

\section{Key cost of a quantum device}\setcurrentname{Key cost of a quantum device}

\label{sec:behaviors}

We now define a quantity for a quantum device, which is analogous to the key cost
 of a state. We follow the idea of introducing measures of quantum non-locality based on entanglement measures, as provided in~\cite{CFH21}~(see also~\cite{HWD21}). Non-locality measures  obtained in this way are called reduced entanglement measures. For a given device, the reduced entanglement measure $E$ is obtained by taking the minimal value of $E$ over all quantum states that lead to the device when one can vary over measurement settings. It is natural to consider reduced key cost and reduced key of formation obtained in this way.

A device, i.e., a conditional probability distribution of two outputs $(a,b)$ given two inputs $(x,y)$, denoted as $P(ab|xy)$ is quantum if and only if there exists a state~$\rho_{AB}$ such that 
\begin{equation}
    \Tr (\rho_{AB} ({\mathrm M}^{a}_x\otimes{\mathrm M}^{b}_y)) = P(ab|xy).
\end{equation}
We denote
such a device  as $({\mathcal M},\rho_{AB})$, where ${\mathcal M} = {\mathrm M}^{a}_x\otimes{\mathrm M}^{b}_y$.
Below, we define the cost of a device for the bipartite case, with the multipartite case following along similar lines.

\begin{definition}
\label{def:di_kc_box}
Let $\varepsilon>0$. The key cost of a quantum device $P(ab|xy)$ is equal to the reduced single-shot key cost
\begin{multline}
     K_C^{\varepsilon}(P(ab|xy)) \equiv K_C^{\varepsilon \downarrow}(P(ab|xy)) \coloneqq \\
     \inf_{\substack{\rho_{AB},\\
    {\mathrm M}^{a}_x\otimes{\mathrm M}^{b}_y}} \{ K_C^{\varepsilon}(\rho_{AB}) :
     \Tr( \rho_{AB} {\mathrm M}^{a}_x\otimes{\mathrm M}^{b}_y) = P(ab|xy)\}.
\end{multline}
\end{definition}

Besides the technical origin of the above definition of a single-shot key cost of a device, it has a clear operational interpretation. 
Namely, it can be viewed as the minimal
cost (in terms of device-dependent secure key) incurred by a mistrustful vendor upon creation of the device. 
Let us note here that the dishonest provider
of a device $P(a,b|x,y)$ during production process is only
limited by the future statistics of inputs and outputs of the
device, which motivates taking the infimum
over states and measurements
compatible with $P(a,b|x,y)$.

We now establish a nontrivial bound for all devices
that admit a realization on states with a positive partial transpose (PPT states). The PPT states satisfy $\id_A\otimes T_B(\rho_{AB})\equiv \rho_{AB}^{\Gamma_B}\geq 0$, where $T_B$ is the transposition map applied to system $B$. The set of all PPT states is denoted as $\PPT$. We say
that $P(ab|xy) \in \PPT$ iff $P(ab|xy)=\Tr(\rho_{AB} {\mathrm M}^{a}_x\otimes{\mathrm M}^{b}_y)$ for some measurements ${\mathrm M}^{a}_x\otimes{\mathrm M}^{b}_y$ and a state $\rho_{AB}$ with positive partial transpose. We call the state $\rho_{AB}$ a PPT realization of $P(ab|xy)$. 
\begin{proposition}
For all $\varepsilon>0$ and for every $P(ab|xy) \in \PPT$ with  PPT realization $\rho_{AB}$, the following inequality holds
\begin{equation}
    K^{\varepsilon}_C (P(ab|xy))\leq \min \{K_C^{\varepsilon}(\rho_{AB}),K_C^{\varepsilon} (\rho_{AB}^{\Gamma_B})\}.
\end{equation}
\end{proposition}
\begin{proof}
Following~\cite{CFH21} it suffices to notice that, by the identity $\Tr(XY) = \Tr\!\left(X^{\Gamma}Y^{\Gamma}\right)$ and the fact that $\rho_{AB}^{\Gamma_B} \geq 0$, we have that, for every
realization $(\sigma_{AB},{\mathcal N})$ with 
${\mathcal N} = {\mathrm N}^{a}_x\otimes{\mathrm N}^{b}_y$ of a device $P(ab|xy)$, the equality
$(\sigma, {\mathcal N})=(\sigma^{\Gamma_B},{\mathcal N}^{\Gamma_B})$ holds.
\end{proof}
We can also define the device-independent key cost 
of a state.
\begin{definition} 
\label{def:di_kc_state}
The device-independent key cost of a bipartite state $\rho_{AB}$ is defined as
\begin{equation}
K_C^{\varepsilon}(\rho_{AB}) \coloneqq  \sup_{\mathcal M} K_C^{\varepsilon}(\rho_{AB},{\mathcal M}).
\end{equation}
\end{definition}
In particular, for the states $\rho$ on ${\mathbb C}^{2 d_s}\otimes {\mathbb C}^{2d_s}$ given in~\cite{smallkey} for which it is proven in~\cite{CFH21}, that device-independent key is (strictly) upper bounded by $K_{\operatorname{D}}(\rho^{\Gamma})$, there is
\begin{equation}
    K_C^{\varepsilon}(\rho_{AB}) \leq K_C^{\varepsilon}(\rho_{AB}^{\Gamma}) \leq \frac{1}{\sqrt{d_s}+1}.
    \label{eq:smallcost}
\end{equation}
These states have vanishing device-independent key cost with increasing dimension $d_s$. Interestingly, however, they have a device-dependent key rate equal to one. It corresponds to the fact that these states tend to be indistinguishable from their separable key-attacked version (i.e., $K_C$ is zero) where distinguishability is performed using LOCC operations~\cite{HorodeckiMurta, karol-PhD}.
 One can also define an asymptotic version of the key cost of a device, which, apart from an analogy of~\eqref{eq:smallcost}, can be related to the device-independent key via Eqs.~(2) and (5) of~\cite{CFH21}
\begin{multline}
K_{\operatorname{DI}}(\rho)\leq K_D^{\downarrow}(\rho)\\ \leq K_C^{\downarrow}(\rho) 
\equiv\sup_{{\cal M}}\inf_{(\sigma,{\cal N})=(\rho,{\cal M})} K_C(\sigma).
\end{multline}

Analogous results can be obtained for the reduced key of formation, which is defined as
\begin{multline}
K_F(P(ab|xy))\coloneqq \\
\inf_{({\cal N},\sigma_{AB})} \{K_F(\sigma_{AB}): ({\cal N},\sigma_{AB}) = P(ab|xy)\}.
\label{eq:kf_dev}
\end{multline}
Furthermore, definitions analogous to Defs.~\ref{def:di_kc_box} and \ref{def:di_kc_state} lead to the key cost for a given set of parameters tested in a DIQKD protocol. These are obtained by replacing the infimum over
states $\rho_{AB}$ and measurements
${\cal M}$ compatible with the device 
$P(a,b|x,y)$ by the infimum over those compatible with the parameters of the game score
(e.g. CHSH) and quantum bit error rate
value (for analogous approach in context of device-independent key, see~\cite{Kaur2022}).
\begin{remark}
We note here that the measure defined in~\eqref{eq:kf_dev} need not coincide with the key of formation of a device, which would be an analog of $K_F$ for quantum states. The latter analog would involve the infimum over splittings into devices that contain an ideal device-independent key (denoted as IDIK). While devices lying on the boundary of a quantum set (i.e. having quantum realization) belong to IDIK, it is unclear if they exhaust IDIK.
\end{remark}
	
\section{Proof of subadditivity of \texorpdfstring{$K_F$}{KF}}\setcurrentname{Proof of subadditivity of \texorpdfstring{$K_F$}{KF}}
\label{sec:subadditivity_of_key_of_formation}

It is easy to observe that the 
key of formation is subadditive.
We state it in a more general form 
for any convex roof measure.
Let $S$ denote a subset of all finite-dimensional quantum states with the following properties:
\begin{enumerate}
    \item $S$ is closed under tensor product,
    \item every quantum state $\rho$ can be written as a convex combination of elements from $S$.
\end{enumerate}
Note that, as an example, the set of all pure states satisfies the conditions above. A convex-roof measure $E_S$, with respect to the set $S$ of $\rho$, is then
\begin{equation}
    E_S(\rho) = \inf_{\{(p_i, \rho_i)\}_i } \sum_i p_i f(\rho_i)
\end{equation}
for some real non-negative valued function $f$, where the infimum is taken over every ensemble decomposition $\sum_i p_i \rho_i = \rho$ and $\rho_i \in S$ for all $i$.

\begin{observation}
\label{obs:subadd}
For a convex-roof measure $E_S$ taken with
respect to set $S$ that is closed under 
tensor products, $E_S$ is subadditive;
i.e., $E_S(\rho\otimes \rho')\leq E_S(\rho)+E_S(\rho')$.
\end{observation}

\begin{proof}
Let $\{(p_i, \sigma_i)\}_i$ and $\{(q_j, \tau_j)\}_j$ be two ensemble decompositions 
that attain $E_S(\rho)$ and $E_S(\rho')$
respectively. Then, from the facts that
$S$ is closed under tensor products
and $\sigma_i,\tau_j\in S$, we have that $\{(p_iq_j,\sigma_i\otimes\tau_j)\}_{i,j}$ is
a valid ensemble decomposition of $\rho\otimes \rho'$.
However, it can be suboptimal, concluding the proof.
\end{proof}
We have then the following corollary.
\begin{corollary}
\label{cor:additive}
 $K_F$ is subadditive.
\end{corollary}

\begin{proof}
It follows from Observation~\ref{obs:subadd} and the fact that the set of generalized private states is closed under tensor products.
The fact that strictly irreducible private states are closed under tensor products was observed in the proof of \cite[Theorem 29 (Supplementary Material)]{FerraraChristandl}.

We give the proof for generalized private states for the sake of completeness.
To see the above, consider a tensor-product
$\gamma^{(1)}_{A_1^KA_1^SB_1^KB_1^S}\equiv\gamma^{(1)}$ with $\gamma^{(2)}_{A_2^KA_2^SB_2^KB_2^S}\equiv\gamma^{(2)}$ 
where
\begin{align}
    \gamma^{(1)} & \coloneqq  U_1|\psi_1\>\!\<\psi_1|_{A_1^KB_1^K}\otimes \rho_{A_1^SB_1^S}U_1^{\dagger},\\
    \gamma^{(2)}& \coloneqq U_2|\psi_2\>\!\<\psi_2|_{A_2^KB_2^K}\otimes \rho_{A_2^SB_2^S}U_2^{\dagger},\\
    U_1 & \coloneqq  \sum_{i}|ii\>\!\<ii|\otimes U^{(1)}_i , \\
    U_2 & \coloneqq  \sum_{i}|ii\>\!\<ii|\otimes U^{(2)}_i.
\end{align}
Consider local unitary transformations
$V_{A_1^S,A_2^K}\otimes V_{B_1^S,B_2^K}$ that swap the shield systems  $A_1^SB_1^S$ of the  private state $\gamma^{(1)}_{A_1^KA_1^SB_1^KB_1^S}$ with (properly embedded to equal dimensions) the key part of $A_2^KB_2^K$ of $\gamma^{(2)}_{A_2^KA_2^SB_2^KB_2^S}$. It follows that
the state 
\begin{multline}
    \left(V_{A_1^S,A_2^K}\otimes V_{B_1^S,B_2^K}\right)\left(\gamma^{(1)}\otimes \gamma^{(2)}\right) \times \\
    \left(V_{A_1^S,A_2^K}^{\dagger}\otimes V_{B_1^S,B_2^K}^{\dagger}\right)
\end{multline}
is a generalized private state of the form
\begin{multline}
W \Bigg(|\psi_1\>\!\<\psi_1|_{A_1^KB_1^K}\otimes|\psi_2\>\!\<\psi_2|_{A_2^KB_2^K}\\
\otimes\rho_{A_1^SB_1^S}\otimes \rho_{A_2^SB_2^S} \Bigg) W^{\dagger}\\
\equiv
    W |\psi_1\>\!\<\psi_1|_{A_1^KB_1^K}\otimes|\psi_2\>\!\<\psi_2|_{A_2^KB_2^K}\otimes\rho_{A_1^SA_2^SB_1^SB_2^S} W^{\dagger} ,
\end{multline}
where 
 \begin{equation}
     W\coloneqq \sum_{i,j}|ii\>\!\<ii|_{A_1^KB_1^K}\otimes |jj\>\!\<jj|_{A_2^KB_2^K}\otimes U_i^{(1)}\otimes U_j^{(2)},
 \end{equation}
 which completes the proof.
\end{proof}

\section{Outlook and Discussion}

We have introduced two quantities central to the resource theory of privacy: private key cost and private key of formation. We have done so assuming the faithfulness of the distillable key (i.e., nonexistence of entangled states with zero distillable key). One of our fundamental contributions is that the regularized key of formation of a bipartite state upper bounds the key cost  $(K_F^{\infty}\geq K_C)$.  Equality holds (i.e., $K_F^{\infty}= K_C$) if $K_F$ is an asymptotically continuous entanglement monotone. While we have shown the monotonicity of $K_F$ with respect to some basic operations (see Section~\ref{sec:key_of_formation}), the monotonicity (on average) under von Neumann measurements and asymptotic continuity beyond trivial cases, remain important open problems.
The importance of these problems also stems from the fact that they would imply non-trivial results for the key cost of quantum channels. 

Let us note here that $K_C$, similar to most operational entanglement measure, is difficult to compute. However, the bounds $E_R^\infty\leq K_C \leq K_F^\infty \leq K_F$ provide a way to estimate it. It is then an interesting open problem to determine for which states  other than strictly irreducible generalized private states, can $K_F$ be found analytically or computed numerically, beyond the trivial upper bound $K_F\leq E_F$ or the lower bound $E_R^{\infty}\leq K_F$ established here.

It is striking that, while showing a function to be entanglement monotone is standard by now~\cite{Horodecki2009}, all known methods from entanglement theory do not immediately apply to show that $K_F$ is an asymptotically continuous LOCC monotone. For a detailed discussion see Section \ref{sec:monotonicity}.

We consider it salient to generalize our findings to the settings of conference key agreement~\cite{CKA_Murta,Das2021}.
It is also of prime importance to compute $K_C$ for some states other than generalized private states.

We stress that, from a purely mathematical perspective,
rather than from an operational viewpoint, the presented results remain valid even if entangled but zero distillable key states were to
exist (see Remarks~\ref{rem:broaden-Kc-definition} and~\ref{rem:broaden-Kf-definition}).

{\it Note added.} Upon finishing this article, we became aware
that the most general definition of the key cost, denoted as $ K^s_C$ and called strict key
cost (corresponding to the $K_C'$ mentioned in Remark \ref{rem:broaden-Kc-definition}), was proposed by Stefan B{\"a}uml~\cite{Stefan} in his master's thesis, along with some initial observations. Ref.~\cite{Stefan} includes also the main result of~\cite{IR-pbits}, which characterized 
irreducible private states.

\begin{acknowledgements}
KH acknowledges the Fulbright Program and Cornell ECE for hospitality during the Fulbright scholarship at the School of Electrical and Computer Engineering of Cornell University.
 KH acknowledges partial support by the Foundation for Polish Science (IRAP project, ICTQT, contract no.\ MAB/2018/5, co-financed by EU within Smart Growth Operational Programme). The `International Centre for Theory of Quantum Technologies' project (contract no.\ MAB/2018/5) is carried out within the International Research Agendas Programme of the Foundation for Polish Science co-financed by the European Union from the funds of the Smart Growth Operational Programme, axis IV: Increasing the research potential (Measure 4.3). SD acknowledges support from the Science and Engineering Research Board (now ANRF), Department of Science and Technology, Government of India, under Grant No. SRG/2023/000217 and the Ministry of Electronics and Information Technology (MeitY), Government of India, under Grant No. 4(3)/2024-ITEA. SD also thanks IIIT Hyderabad for the Faculty Seed Grant. KH acknowledges National Science Centre, Poland, grant Opus 25, UMO-2023/49/B/ST2/02468. MMW acknowledges support from the National Science Foundation under Grant No.~2315398.
\end{acknowledgements}

\bibliographystyle{quantum}
\bibliography{main.bib}

\clearpage

\end{document}